%% file: 0.VexIR2Vec_Main.tex
\documentclass[acmsmall,screen,nonacm]{acmart}
\usepackage{preamble}

\AtBeginDocument{%
  \providecommand\BibTeX{{%
    \normalfont B\kern-0.5em{\scshape i\kern-0.25em b}\kern-0.8em\TeX}}}

\begin{document}
\input{preamble-defns}

\title{\vexirtovec{}: An Architecture-Neutral Embedding Framework for Binary Similarity}

\author{S. VenkataKeerthy}
\email{cs17m20p100001@iith.ac.in}
\orcid{0000-0003-1393-7321}
\author{Soumya Banerjee}
\email{cs22mtech12011@iith.ac.in}
\orcid{0009-0003-5772-2469}
\author{Sayan Dey}
\email{cs22mtech02005@iith.ac.in}
\orcid{0000-0002-4115-0588}
\author{Yashas Andaluri}
\email{cs17b21m000001@iith.ac.in}
\orcid{0000-0003-1180-4197}
\author{Raghul PS}
\email{compiler.intern.24003@cse.iith.ac.in}
\orcid{0009-0003-3059-5025}
\author{Subrahmanyam Kalyanasundaram}
\email{subruk@cse.iith.ac.in}
\orcid{0000-0001-9094-3368}
\affiliation{%
  \institution{IIT Hyderabad}
  \city{Hyderabad}
  \country{India}
}

\author{Fernando Magno Quint\~ao Pereira}
\email{fernando@dcc.ufmg.br}
\orcid{0000-0002-0375-1657}
\affiliation{%
  \institution{UFMG}
  \city{Belo Horizonte}
  \country{Brazil}
}
\author{Ramakrishna Upadrasta}
\email{ramakrishna@cse.iith.ac.in}
\orcid{0000-0002-5290-3266}
\affiliation{%
  \institution{IIT Hyderabad}
  \city{Hyderabad}
  \country{India}
}

\renewcommand{\shortauthors}{S. VenkataKeerthy, et al.}

\begin{abstract}
    \input{1.abstract}

\end{abstract}



\maketitle

\input{2.introduction}
\input{3.background}
\input{4.approach}
\input{5.experiments}

\input{6.ablation}
\input{7.relatedWorks}
\input{8.conclusions}


\bibliographystyle{ACM-Reference-Format}
\bibliography{references}

\appendix
\input{9.appendix}

\end{document}

%% file: preamble-defns.tex
\definecolor{olivegreen}{RGB}{0,153,0}
\definecolor{awesome}{rgb}{1.0, 0.13, 0.32}
\definecolor{light-gray}{gray}{0.80}
\definecolor{light-blue}{RGB}{207, 250, 241}
\definecolor{light-green}{RGB}{227, 247, 183}%
\definecolor{light-yellow}{RGB}{250, 247, 183}
\definecolor{light-orange}{RGB}{250, 235, 207}

\mdfdefinestyle{mystyle}{
    backgroundcolor=white!20
}

\newcommand{\urk}[1]{\textcolor{red}{\textbf{URK:$\bigstar$}#1}}
\newcommand{\urkfixme}[1]{\textcolor{red}{\textbf{URK FIXME: $\bigstar$}#1}}
\newcommand{\vk}[1]{\textcolor{blue}{\textbf{VK:$\bigstar$}#1}}
\newcommand{\vksuggest}[2]{%
    \textcolor{blue}{\sout{#1}}
    \textcolor{blue}{\ #2}
}

\newcommand{\vexirtovec}{{\sc VexIR2Vec}\xspace}

\definecolor{beaublue}{rgb}{0.74, 0.83, 0.9}
\definecolor{lightgray}{RGB}{245, 243, 242}
\definecolor{lightblue}{RGB}{237, 246, 250}
\definecolor{lightyellow}{RGB}{243, 247, 225}
\definecolor{lightred}{RGB}{250, 237, 237}
\definecolor{lightorange}{RGB}{252, 245, 235}
\definecolor{auburn}{rgb}{0.43, 0.21, 0.1}

\newtheorem{observation}[theorem]{Observation}

\newcommand{\binutils}{{\sc binutils}}

\newcommand{\embed}[1]{$\llbracket \mathbf{#1} \rrbracket$}

\NewDocumentCommand{\embedEqn}{m o o}{%
  \llbracket \mathbf{#1}%
  \IfValueTF{#2}{%
    \StrLeft{#2}{1}[\firstchar]%
    \IfStrEq{\firstchar}{^}{%
      {#2}%
    }{%
      \IfStrEq{\firstchar}{_}{%
        {#2}%
      }{%
        #2%
      }%
    }%
  }{}%
  \IfValueT{#3}{_{#3}}%
  \rrbracket
}

\newcommand{\embedFInit}{\embedEqn{F^{\scaleto{\mathcal{V}}{3.5pt}}}}
\newcommand{\embedFFinal}{\embedEqn{F}}

\newcommand{\llangle}{\left\langle}
\newcommand{\rrangle}{\right\rangle}

\newcommand{\triplet}[3]{\langle #1, #2, #3 \rangle}
\newcommand{\vocab}[1]{\mathcal{V}_{lookup}(#1)}
\newcommand{\angr}{\texttt{angr}\xspace}
\newcommand{\vexir}{VEX-IR\xspace}

\newcommand{\optengine}{{\sc VexINE}\xspace}
\newcommand{\optenginelong}{\vexir Normalization Engine\xspace}

\newcommand{\vexnet}{{\sc VexNet}\xspace}

\definecolor{ygreen}{HTML}{d6e39d}
\definecolor{sgreen}{HTML}{bee69c}

\newcommand{\first}[1][]{\cellcolor{sgreen}}
\newcommand{\second}[1][]{\cellcolor{ygreen}}

\lstdefinelanguage{VEX}
{
  morekeywords={
    opcd, mov, add, mul, sub, set, load, store, eq, neq, lt, gt, le, ge,
    and, or, xor, not, shl, shr, put, puti, get,  geti, ext, trunc, ldle, 32uto64, 64to32, add32, stle 
  },
  sensitive=false,
  morecomment=[l]{//},
}

\lstset{
  language=VEX,
  commentstyle=\itshape \bfseries \color{auburn},
  xleftmargin=11.5pt, 
  backgroundcolor=\color{white},   
  basicstyle=\ttfamily\footnotesize,        
  breakatwhitespace=false,         
  breaklines=true,                 
  captionpos=b,                    
  escapechar=|,
  escapeinside=||,
  frame=none,	                   
  keywordstyle=\color{blue},       
  numbers=left,                    
  numbersep=10pt,                   
  numberstyle=\tiny, 
  rulecolor=\color{black},         
  numberblanklines=false,
  mathescape,
}

\lstset{
  language=[x86masm]Assembler,
  commentstyle=\itshape \bfseries \color{auburn},
  xleftmargin=11.5pt, 
  backgroundcolor=\color{white},   
  basicstyle=\ttfamily\scriptsize,        
  breakatwhitespace=false,         
  captionpos=b,                    
  escapechar=|,
  escapeinside=||,
  frame=none,	                   
  keywordstyle=\color{blue},       
  numbers=left,                    
  numbersep=10pt,                   
  numberstyle=\tiny, 
  rulecolor=\color{black},         
  numberblanklines=false,
  mathescape,
}

\newcommand{\highlightLines}[2]{%
  \colorbox{#1}{%
    \begin{minipage}{\dimexpr\linewidth-2\fboxsep}%
      #2%
    \end{minipage}%
  }%
}

\let\origthelstnumber\thelstnumber
\makeatletter

\lst@Key{countblanklines}{true}[t]%
    {\lstKV@SetIf{#1}\lst@ifcountblanklines}

\lst@AddToHook{OnEmptyLine}{%
    \lst@ifnumberblanklines\else%
       \lst@ifcountblanklines\else%
         \advance\c@lstnumber-\@ne\relax%
       \fi%
    \fi}

\newcommand*\Suppressnumber{%
  \lst@AddToHook{OnNewLine}{%
    \let\thelstnumber\relax%
     \advance\c@lstnumber-\@ne\relax%
    }%
}

\newcommand*\Reactivatenumber[1]{%
  \setcounter{lstnumber}{\numexpr#1-1\relax}
  \lst@AddToHook{OnNewLine}{%
   \let\thelstnumber\origthelstnumber%
   \refstepcounter{lstnumber}%
  }%
}


\newcommand{\diagonalstripes}{%
  \fill[white, pattern=north east lines, pattern color=white] (0,0) rectangle (4,0.5);
}

\newenvironment{rationale}[1][Rationale]{\proof[\sc #1]}{\endproof}
\renewcommand{\qedsymbol}{$\blacksquare$}

%% file: 1.abstract.tex
Binary similarity involves determining whether two binary programs exhibit similar functionality, often originating from the same source code.
In this work, we propose \vexirtovec, an approach for binary similarity using \vexir, an architecture-neutral Intermediate Representation (IR). 
We extract the embeddings from sequences of basic blocks, termed \textit{peepholes}, derived by random walks on the control-flow graph. 
The peepholes are normalized using transformations inspired by compiler optimizations. 
The \optenginelong (\optengine) mitigates, with these transformations, the architectural and compiler-induced variations in binaries while exposing semantic similarities.
We then learn the vocabulary of representations at the entity level of the IR using the knowledge graph embedding techniques in an unsupervised manner. 
This vocabulary is used to derive function embeddings for similarity assessment using \vexnet, a feed-forward Siamese network designed to position similar functions closely and separate dissimilar ones in an $n$-dimensional space.
This approach is amenable for both diffing and searching tasks, ensuring robustness against Out-Of-Vocabulary (OOV) issues.

We evaluate \vexirtovec on a dataset comprising $2.7M$ functions and $15.5K$ binaries from $7$ projects compiled across $12$ compilers targeting x86 and ARM architectures.
In diffing experiments, \vexirtovec outperforms the nearest baselines by $40\%$, $18\%$, $21\%$, and $60\%$  in cross-optimization, cross-compilation, cross-architecture, and obfuscation settings, respectively. In the searching experiment, \vexirtovec achieves a mean average precision of $0.76$, outperforming the nearest baseline by $46\%$.
Our framework is highly scalable and is built as a lightweight, multi-threaded, parallel library using only open-source tools. \vexirtovec is $\approx 3.1$--$3.5 \times$ faster than the closest baselines and orders-of-magnitude faster than other tools.

%% file: 2.introduction.tex
\section{Introduction}
\label{sec:introduction}

Binary similarity is the task of determining whether two binary programs exhibit similar functionality,
 often originating from the same source code.
 Solutions to this problem enable applications in
 vulnerability analysis~\cite{gao2018vulseeker, shigang2020cyberVulnerability},
 malware detection~\cite{farhadi2014binclone},
 plagiarism identification~\cite{Ming2016PlagDetection},
 copyright authentication~\cite{wei2018bcfinder},
 profile matching~\cite{Panchenko19,Wang00},
 code lifting~\cite{Martinez23,VenkataKeerthy-2023-packet_processing},
 and redundancy elimination~\cite{Xue18}.

This paper focuses on two flavors of the binary similarity problem: diffing and searching.
Diffing aims to identify similar or dissimilar regions (e.g., basic blocks, functions) between two binaries. 
Searching involves retrieving a binary function from a large pool of binary functions that is similar to the query code. 
These tasks become challenging in an  \emph{adversarial setting}, where binaries compiled from the same source code vary due to several factors:
(i) Compiler choice (e.g., Clang, GCC, ICC);
(ii) Compiler version (e.g., Clang 6.0.0 vs. Clang 17.0.0);
(iii) Compiler optimizations (e.g., GCC \texttt{-O0} vs. GCC \texttt{-O3});
(iv) Target architecture (e.g., x86, ARM); and
(v) Use of obfuscation (e.g., control-flow flattening or dead control flow).
These variations in compilation configurations can significantly alter the binaries in terms of syntax, semantics, and structure, as we discuss in detail in Section~\ref{sub:motivation}.

\paragraph{State-of-the-art Solutions to Binary Similarity in Adversarial Settings}
Binary similarity tools typically employ various representations to bridge the gap between raw binary and a more analyzable format.
These representations include
assembly languages~\cite{ding2019asm2vec, massarelli2019safe, pei2020trex},
virtualized bytecodes~\cite{chandramohan2016bingo, yaniv2017pldi},
or abstract syntax trees~\cite{yang2021asteria}.
Once converted, several techniques are used to identify code similarities.
These approaches are inherently heuristic because the general problem of determining program equivalence is undecidable~\cite{Rice53}.
Common heuristics leverage
graph matching techniques~\cite{bindiff},
hashing of code sequences~\cite{pewny2015sp, yaniv2017pldi, Wang00, Ayupov24}, or Machine Learning (ML)---an approach that presently enjoy great popularity.
ML models have emerged as the dominant approach due to their ability to learn complex relationships within code~\cite{xu2017gemini, zuo2018innereye, ding2019asm2vec, duan2020deepbindiff, massarelli2019safe, pei2020trex, Yu2020OrderMatters}.

Machine learning approaches for binary similarity require converting code into a numerical vector suitable for use as the model's input.
These vectors can encode either handcrafted features or learned representations. Feature-based approaches~\cite{Feng2016Genius, eschweiler2016discovre, Qasem2023Binfinder-AsiaCCS, Damasio23} use manually defined program characteristics, such as the number of opcodes, loops, and function calls, to represent the binary. In contrast, distributed representations are learned through representation learning techniques~\cite{replearning-review}. This \textit{learned representation} is a real-valued vector of a chosen dimensionality, conventionally referred to as an \textit{embedding}~\cite{ding2019asm2vec, duan2020deepbindiff, massarelli2019safe, wang2023sem2vec}. Embeddings capture complex relationships within the code that may not be easily captured by hand-crafted features.

\paragraph{Limitations of previous work}

Although much progress has been achieved in the domain of binary similarity, we perceive a few limitations in the current state-of-the-art solutions:

\begin{enumerate}
\item 
The existing approaches that use assembly code for modeling binary similarity achieve good results~\cite{duan2020deepbindiff, wang2023sem2vec, li2021palmtree, ahn2022BinShot}, but are trained for specific architectures. Hence, they cannot be used to compare binaries targeting new architectures that were not included in the training.

\item 
Binary comparison tools that determine a similarity score between pairs of binaries~\cite{duan2020deepbindiff,
 zuo2018innereye, zeek2018plas} face challenges in scalability, particularly for searching tasks. While effective for diffing, their reliance on pairwise comparisons makes them impractical for searching large datasets due to the resulting quadratic worst-case time complexity.

\item 
Scalability is also an issue in binary similarity tools that are based on the modern language models~\cite{devlin2019bert, Liu2019Roberta, Vaswani2017attentionTransformers}.
Techniques such as Oscar~\cite{peng2021oscar}, PalmTree~\cite{li2021palmtree}, jTrans~\cite{Wang22jTrans}, Sem2Vec~\cite{wang2023sem2vec} and kTrans~\cite{zhu2023ktrans} require very high training time, even when using clusters of GPUs~\cite{marcelli2022usenix}.
For instance, Oscar and jTrans were trained using $8$ V100 GPUs each, while kTrans was trained using $4$ V100 GPUs.
Nevertheless, training times often span weeks.

\item 
Encoding program Control-Flow Graphs (CFGs) using Graph Neural Networks (GNNs) presents scalability challenges due to the high computational cost inherent to GNNs~\cite{GNNBook2022, marcelli2022usenix}. This limitation also affects tools like DeepBinDiff~\cite{duan2020deepbindiff}, which rely on DeepWalk and matrix factorization~\cite{yang2015DeepWalk} for CFG modeling. In our experiments with DeepBinDiff (Sec.~\ref{sub:scalability}), training times reached approximately 7.5 hours per epoch, hindering its application to binaries exceeding 300KB within a two-hour analysis (``timeout'') window.

\item 
A final shortcoming that we perceive in previous work regards availability. As shown by \citet{haq2021binsurvey}, several approaches do not release their software, limiting the reproducibility of scientific results~\cite{marcelli2022usenix}.
Moreover, some of these tools~\cite{yang2021asteria, ding2019asm2vec, xu2017gemini, Luo23}   are limited by the use of licensed/proprietary disassemblers such as IDA-Pro~\cite{IDAPro}.
\end{enumerate}

\paragraph{\vexirtovec{}: The Contribution of this Work}
To address these limitations, we propose \vexirtovec{}, an embedding approach that represents binary functions as continuous, $n$-dimensional distributed vectors.
The design of \vexirtovec{} embodies five characteristics, which we describe below.

\begin{description}
\item [Architecture-Neutral:] 
The \vexirtovec{} embedding of a binary is extracted from its \vexir{} intermediate representation, which is architecture neutral~\cite{wang2017angr}.
Thus, as Section~\ref{sub:vexir} explains, \vexirtovec{} can be used to compare binaries compiled to targets
such as x86 and ARM.

\item [Structural Encodings:] The \vexirtovec{} embedding is extracted from sequences of basic blocks taken from the function's control-flow graph.
These straight-line sequences---henceforth called \textit{peepholes}---are produced via random walks (Section~\ref{sub:generating-peepholes}).
Random walks reveal structural properties,
emphasizing
frequently connected blocks and blocks nested within loops.

\item [Normalizing Transformations:] 
The peepholes are amenable to {\it normalizing transformations} (Section~\ref{sub:normalization}).
Normalizations are rewriting rules inspired by compiler optimizations\footnote{
The normalizations described in Section~\ref{sub:normalization} include register optimizations, copy propagation, constant propagation, constant folding, redundancy elimination, and load-store elimination.
These rules are applied to local sequences of instructions (the peepholes) without global code knowledge; hence, they are unsound.
However, soundness is not important in this context: the normalizations are designed to reduce the differences in IR generated from different architectures and compilers.
}.
They remove \textit{uninteresting} syntactic details from the peepholes that do not contribute to revealing their essential semantics.
We implement these normalizations on \vexir as a library called \optengine.

\item [Learned Embeddings:] Vocabulary of \vexirtovec{} is learned from the IR entities---opcodes, types, arguments---using representation learning techniques with simple feed-forward networks~\cite{replearning-review,knowledge-graph-embedding-survey}.
Thus, unlike approaches that learn the representations of instructions~\cite{marcelli2022usenix,zuo2018innereye}, \vexirtovec avoids Out-Of-Vocabulary (OOV) issues.
Learning this vocabulary is a \textit{one-time pre-training} step; once learned, the vocabulary is independent of the binary similarity task (Section~\ref{sec:phaseII-embeddings}).

\item [Application-Independent:] \vexirtovec{} can be adapted to different tasks, such as diffing or searching.
To this end, we designed \vexnet, a Siamese network~\cite{koch2015siamese} which \textit{fine-tunes} the vocabulary to represent functions as points in an embedding space, where similar functions are closer to each other while dissimilar functions are far apart (Section~\ref{sec:phaseIII-model}).

\end{description}

The intuition behind the design of \vexirtovec{} is that decomposing two similar functions into a sufficiently large number of peepholes is likely to result in many of these peepholes having similar semantics.
This semantics can be inferred as the composition of the semantics of the individual entities that make up the peephole.
The implementation of this intuition into an actual tool is able to address the limitations of previous works in terms of scalability, precision, and availability.

\paragraph{Scalability}
The experiments in Section~\ref{sec:performance-evaluation} demonstrate the practicality of \vexirtovec.
The peephole extraction algorithm is linear on the number of basic blocks that constitute the CFG (Section~\ref{sub:generating-peepholes}) as opposed to the other approaches that use GNNs and matrix factorizations.
As the normalizations are applied to the straight-line peepholes, the time taken is linear in the length of the peephole.
As we use simple models, the training time of \vexnet{} is about $5$--$8$ seconds per epoch, resulting in an improvement of about $1080 \times$--$5000 \times$ in comparison to systems like SAFE~\cite{massarelli2019safe} and BinFinder~\cite{Qasem2023Binfinder-AsiaCCS}.
Inference, i.e., the usage of the trained system, is equally fast.
\vexirtovec is $\approx 3.1\times$ faster than SAFE~\cite{massarelli2019safe},
$\approx 3.5 \times$ faster than BinFinder~\cite{Qasem2023Binfinder-AsiaCCS}, and orders-of-magnitude faster than DeepBinDiff~\cite{duan2020deepbindiff}.

\paragraph{Precision}
The evaluation in Section~\ref{sec:performance-evaluation} shows that \vexirtovec{} is more precise than BinDiff~\cite{bindiff}, DeepBinDiff~\cite{duan2020deepbindiff}, SAFE~\cite{massarelli2019safe}, BinFinder~\cite{Qasem2023Binfinder-AsiaCCS}.
We also compare our approach with the representations created by using the histograms of opcodes, originally 
proposed by \citet{Damasio23} for source-codes. 
Our evaluation uses a dataset made of $2.7M$ functions and $15.5K$ binaries built from $7$ projects (Findutils~\cite{findutils}, Diffutils~\cite{diffutils}, Coreutils~\cite{coreutils}, cURL~\cite{curl}, Lua~\cite{lua}, PuTTY~\cite{putty}, and Gzip~\cite{gzip}) compiled with $12$ different compilers targeting x86 and ARM.
In the diffing scenario, \vexirtovec's F1 Score outperforms the nearest baselines in cross-optimization, cross-compiler, cross-architecture, and obfuscation settings by $40\%$, $18\%$, $21\%$, and $60\%$ respectively. 
In the searching scenario, \vexirtovec outperforms the nearest baseline by $46\%$, obtaining a mean average precision of $0.76$ across different configurations.


\paragraph{Summary of Contributions}
The high precision and scalability metrics reported in Sections~\ref{sec:performance-evaluation} and ~\ref{sec:ablation} result from a number of contributions, which we summarize as follows:
\begin{description}
\item[Insights:] We show that the combination of representation learning, extraction of program structure via random walks, and normalization of the code sequences is an effective way to solve binary similarity tasks.

\item[Representation:] We show that \vexir{} is a suitable intermediate representation to solve binary similarity problems in an adversarial setting.

\item[Methodology:] \vexirtovec{} introduces a decoupled approach to solve binary similarity, which separates the task of learning a vocabulary from the task of training a model to solve binary similarity problems such as diffing or searching.

\item[Implementation:] We build \vexirtovec, which consists of the \optenginelong (\optengine), an embedding extractor, and a tunable model (\vexnet). \vexirtovec is a parallel library written in Python, available as a command line tool and via a web interface.

\end{description}

These contributions are detailed throughout the paper.
Section~\ref{sec:background} provides background and motivation, highlighting the challenges of binary similarity.
Section~\ref{sec:overview} introduces our proposed approach, demonstrating how it addresses these challenges.
The next three sections delve into \vexirtovec's technical details.
Section~\ref{sec:phaseI-processing} discusses the first phase, where binaries are decomposed into smaller units.
Section~\ref{sec:phaseII-embeddings} explores the vocabulary learning phase, and Section~\ref{sec:phaseIII-model} details the fine-tuning for similarity.
Evaluation of \vexirtovec's effectiveness is presented in the subsequent sections.
 Section~\ref{sec:experimental-setup} describes the experimental setup, followed by results and ablation studies in Sections~\ref{sec:performance-evaluation} and~\ref{sec:ablation}, respectively.
Section~\ref{sec:relatedworks} reviews related work, highlighting the novelty of our approach.
Finally, Section~\ref{sec:conclusion} summarizes the paper's findings and potential future directions.

%% file: 3.background.tex
\section{Background and Motivation}
\label{sec:background}

In this section, we introduce the key ideas that we discuss in this paper.
We begin by describing the \vexir in Section~\ref{sub:vexir}, the intermediate format we rely on for analyzing binary files.
Then, in Section~\ref{sub:motivation}, we explain why binary similarity analysis is a difficult problem.
In Section~\ref{sub:solution}, we show how our work deals with the challenges mentioned in Section~\ref{sub:motivation}, introducing one of the central notions of this paper: {\it peepholes} of basic blocks.

\subsection{\vexir: The Intermediate Representation}
\label{sub:vexir}

The process of binary similarity analysis may optionally involve lifting the assembly code into some higher level program representation~\cite{ding2019asm2vec, massarelli2019safe, pei2020trex} such as Abstract Syntax Tree (AST)~\cite{yang2021asteria} or an Intermediate Representation (IR)~\cite{chandramohan2016bingo, yaniv2017pldi}.
Once the binary is converted to the desired representation, different graph/hash matching techniques or ML approaches can be used to solve the underlying binary similarity problem.

\begin{figure}[ht]
\centering
\includegraphics[width=\linewidth]{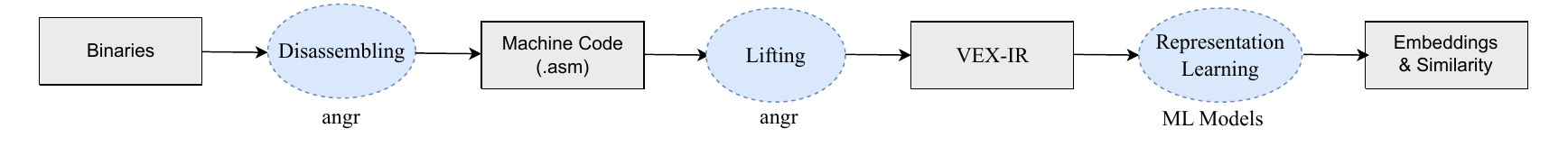}
\caption{\vexirtovec uses \angr for disassembling the binaries and obtaining the \vexir; Function Embeddings for binary similarity tasks are obtained by training simple Feed Forward Neural Networks.}
\Description{Overview}
\label{fig:Configurations}
\end{figure}

The program representation used in this paper to solve binary similarity is \vexir, the IR used by Valgrind~\cite{nicholas2007valgrind} and \angr~\cite{wang2017angr}. 
\vexir is derived from the assembly code; however, it abstracts out many architecture-specific details, such as register names.
Instead of using machine registers, instructions refer to variable names, which are in the static single assignment (SSA) form~\cite{ssa:cytron1991efficiently,rastello2022ssa}.
Thus, each variable name has only one definition site.
\vexir{} presents another high-level characteristic: it is typed.
Nevertheless, \vexir{} also preserves some machine-specific information, such as side effects, pointer sizes, instruction flags, and calling conventions, for instance.
As we explain in Section~\ref{sub:solution}, we opted to work with \vexir{} because this format is architecture-neutral and open source; hence, it is amenable to cross-architecture binary similarity analysis.
Figure~\ref{fig:Configurations} shows how \vexir{} related with the different phases of the binary diffing techniques that this paper introduces.

\begin{figure}[!ht]
\centering
\includegraphics[width=0.8\linewidth]{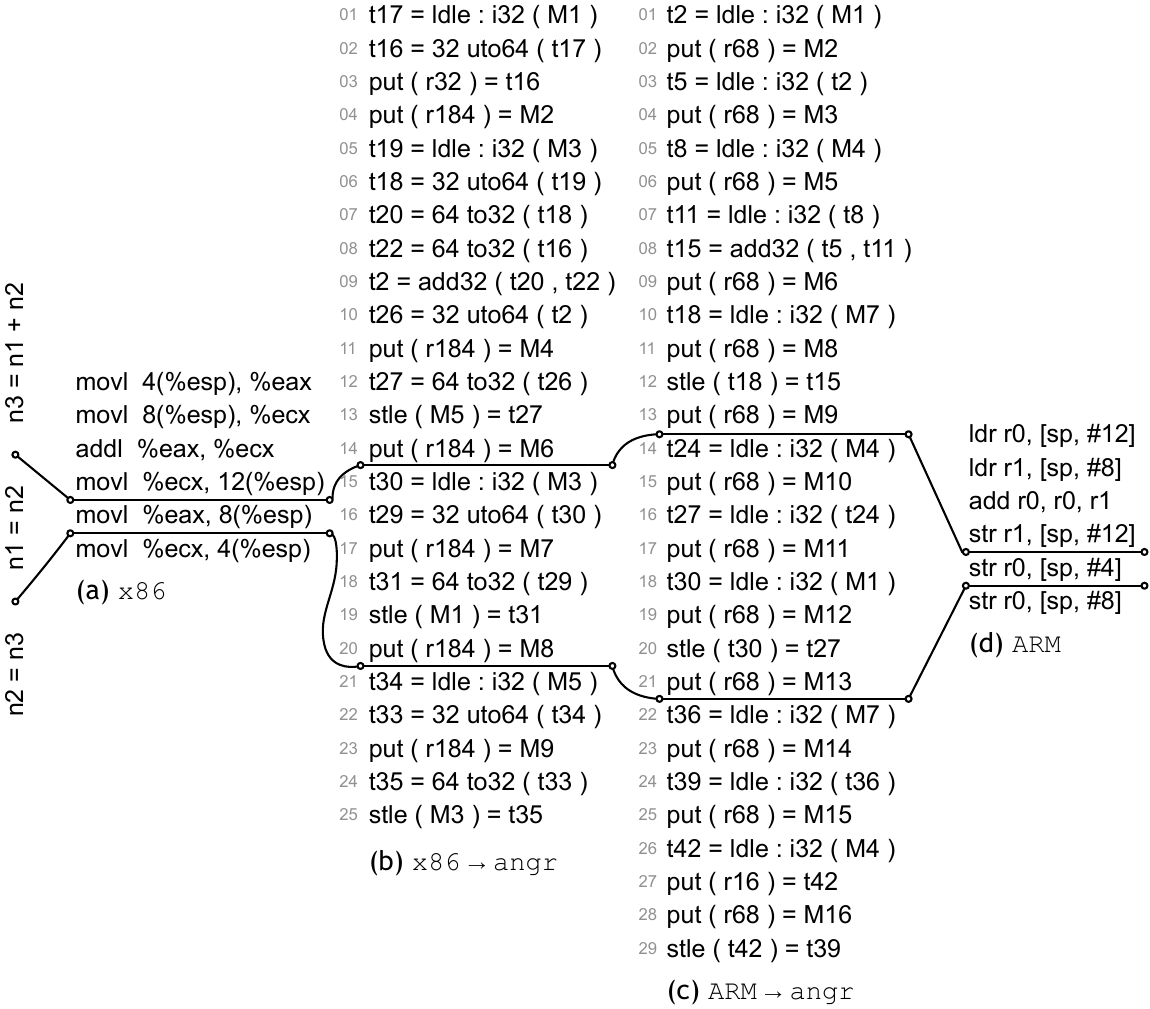}
\caption{\vexir corresponding to the section of the program \texttt{n3 = n1 + n2; n1 = n2; n2 = n3;} to compute the Fibonacci series generated by \angr from x86 and ARM binaries.}
\Description{\vexir example - x86 and ARM}
\label{fig:vexir-example}
\end{figure}

\begin{example}
\label{ex:vexir}
Figures~\ref{fig:vexir-example} (b) and (c) show two \vexir{} programs.
These two programs implement the same sequence of three operations.
They differ because they were produced out of different assembly codes.
Each line of code represents an instruction in Figures~\ref{fig:vexir-example} (a) and (d).
Names such as \texttt{t26} represent variables.
\vexir{} programs can use an unbounded surplus of these names.
Indeed, the static single-assignment property ensures that each variable is defined
with a new name.
Memory addresses and registers, in contrast, do not follow this property.
Thus, a register such as \texttt{r184} in Figure~\ref{fig:vexir-example} (a) can be affected multiple times.
The same is true for memory addresses such as \texttt{M1} or \texttt{M2}.
\end{example}

\subsection{Challenges in Binary Similarity}
\label{sub:motivation}

The binary similarity problem is challenging for several reasons.
The same source code might yield very different binaries due to the differences in compilers, optimization levels, target architectures, or obfuscations0.
Any one of these factors might change the instructions present in a program, as well as the control flow created by these instructions.
This section discusses some of these challenges.

\subsubsection{Compiler and Optimization Levels}
It is well understood that different compilers generate different assemblies and binaries for the same source code.
Such divergences are due to several factors, including the underlying optimizations and internal cost models used by these compilers at different stages of compilation.
Moreover, even binaries produced by the same compiler out of the same source code can present substantial differences, depending on the optimization level used during code generation.

\paragraph{Structural Differences}
The control-flow graph of a program determines the possible paths through which execution can flow along that program.
Many binary diffing tools rely on structural properties of control-flow graphs to compare programs~\cite{Karamitas18,Bourquin13,xu2017gemini,li2021palmtree,duan2020deepbindiff}.
These tools first construct a CFG for each function in the binary, then compare them using graph isomorphism algorithms~\cite{Ullmann76} or ML algorithms and models to identify similar or identical subgraphs.
However, such algorithms leverage structural properties that can be significantly affected by how the binary code is produced.
In particular, optimizations like loop unrolling, function inlining, and tail call elimination cause substantial changes to a program's CFG.
Example~\ref{ex:diffControlFlow} illustrates this issue by comparing control-flow graphs produced by GCC and Clang under different optimization levels.

\begin{figure}[ht]
    \centering
    \subfloat[\footnotesize{x86 - GCC10 \texttt{-O0}}]{{\includegraphics[width=0.22\linewidth, valign=b]{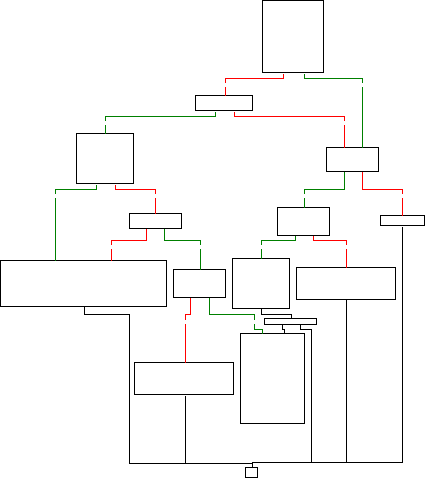} }}%
    \quad
    \subfloat[\footnotesize{x86 - GCC10 \texttt{-O3}}]{{\includegraphics[width=0.22\linewidth, valign=b]{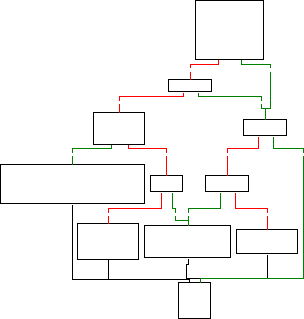} }}
    \quad
    \subfloat[\footnotesize{x86 - Clang10 \texttt{-O0}}]{{\includegraphics[width=0.22\linewidth, valign=b]{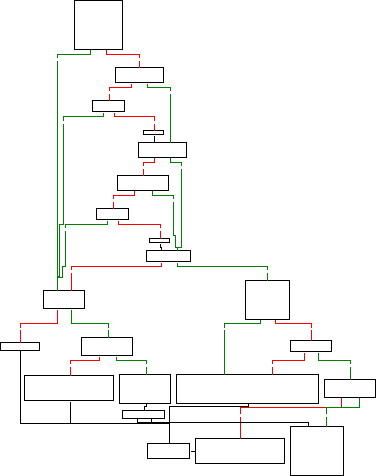} }}%
    \quad
    \subfloat[\footnotesize{x86 - Clang10 \texttt{-O3}}]{{\includegraphics[width=0.22\linewidth, valign=b]{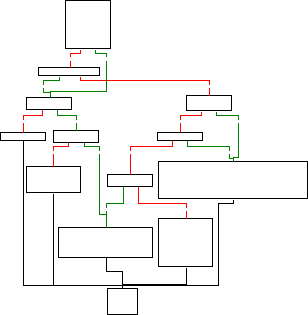} }}
    \Description{The CFG of \texttt{adjust\_relative\_path} method from \texttt{elfedit} of binutils project generated by GCC 10 and Clang 10 compilers under different optimizations}
    \caption{The CFG of \texttt{adjust\_relative\_path} method from \texttt{elfedit} of binutils project generated by GCC 10 and Clang 10 compilers under different optimizations}
    \label{fig:cfg-similarity}
\end{figure}

\begin{example}
\label{ex:diffControlFlow}
Figure~\ref{fig:cfg-similarity} shows the CFGs corresponding to \texttt{adjust\_relative\_path} method of Binutils generated by GCC and Clang compilers at \texttt{-O0} and \texttt{-O3} optimization levels. 
The size of the blocks is proportional to the number of instructions they contain. As it can be observed, the topology of the CFGs, the number of basic blocks, and their size vary significantly across different optimizations and compilers.
\end{example}

\paragraph{Instruction-Level Differences} 
Many binary diffing tools use the hashing of code sequences to compare snippets of binary code~\cite{Jin12,pewny2015sp,yaniv2017pldi,Wang00,Ayupov24}.
However, just like the structural properties previously mentioned, the instructions that make up the function also vary depending on how code is generated.
Differences arise due to the heuristics used by the compilers to perform instruction selection and scheduling.
In this regard, compilers reorder instructions differently across architectures to minimize pipeline stalls and improve execution efficiency.
Furthermore, in architectures such as x86, the same operation can often be expressed in many ways.
Such differences can be dramatic, as Example~\ref{ex:inst-differences} demonstrates.

\begin{example}
\label{ex:inst-differences}
The heatmap in Figure~\ref{fig:inst-freq-dist} shows the frequency distribution of the x86 instructions in the binaries of FindUtils generated by Clang and GCC.
The figure shows that these compilers use different instructions within and across optimization levels.
For instance, the \texttt{endbr64} instruction is emitted $4,682$ times by GCC, whereas Clang generated this instruction only $71$ times.
Similarly, instructions like \texttt{inc} and \texttt{dec} are mostly emitted by both compilers when they optimize for code size at the \texttt{-Os} level.
Even among the vector instructions, Clang and GCC use different instructions as Figure~\ref{fig:inst-freq-dist}
demonstrates.
    
\end{example}

\begin{figure}[ht]
    \centering
    \includegraphics[width=0.9\linewidth]{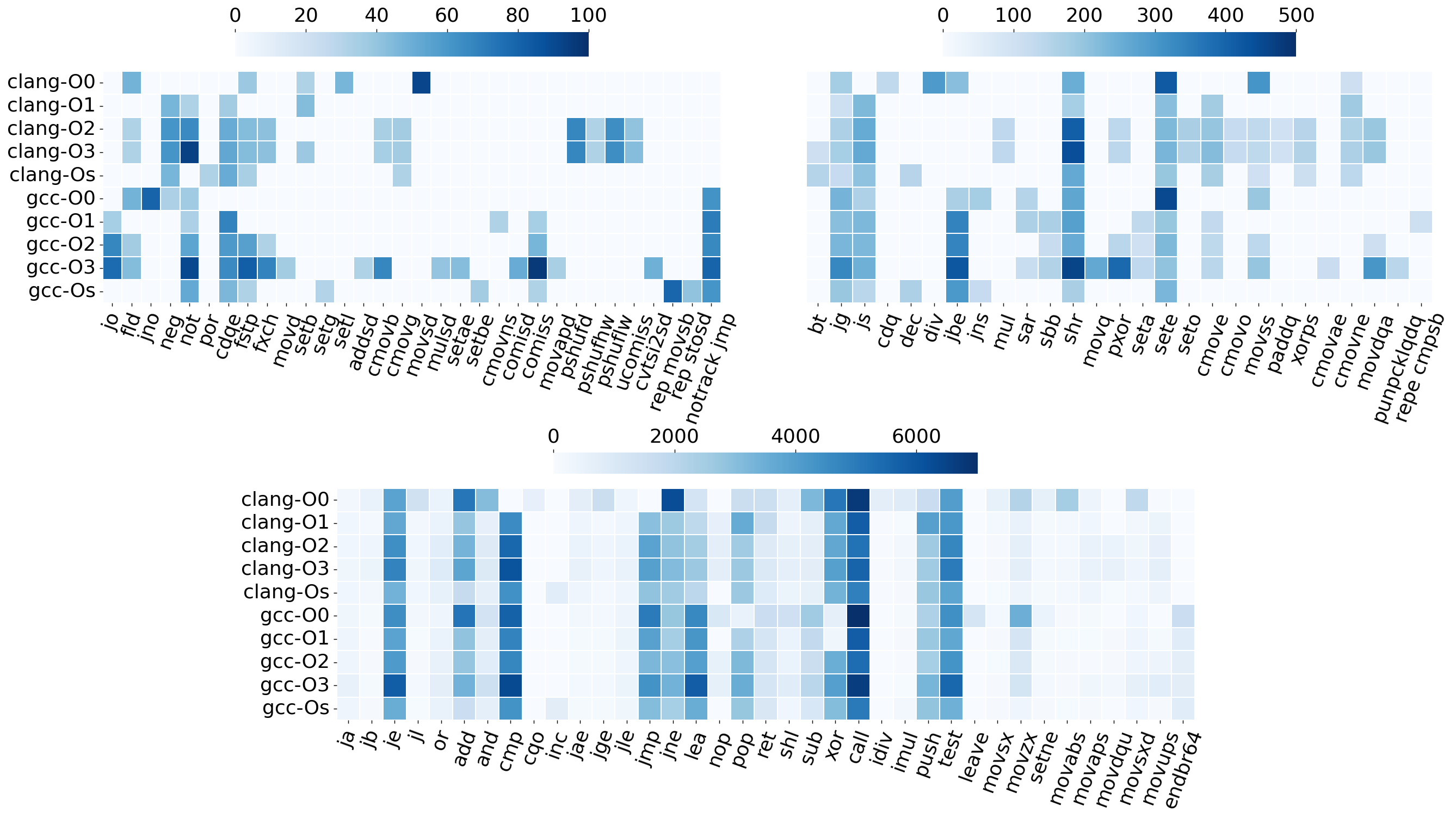}
    \caption{Heatmap showing the frequency distribution of the instructions across the binaries from FindUtils generated by Clang 10 and GCC 10 compilers with different optimization levels.}
    \Description{Heatmap showing the frequency distribution of the instructions across the binaries from FindUtils generated by Clang 10 and GCC 10 compilers with different optimization levels.}
    \label{fig:inst-freq-dist}
\end{figure}

\subsubsection{Target Architecture}
The challenges in cross-architecture binary similarity stem from the differences in the syntax of assembly code and the structure of the induced CFG.
Even when this target-specific assembly code is lifted to a common intermediate representation, the problem of matching them is still difficult because executables written in different assembly codes are unlikely to yield the same structure and instructions once converted to a common intermediate representation.
Many factors might explain such differences:

\begin{enumerate}
\item Binaries compiled for different architectures may exhibit variations in memory layout, including the organization of data structures and the allocation of functions in memory.

\item Different architectures may use different byte orders (endianess) to
store multi-byte data. 

\item System calls and API functions can vary between architectures.

\item Binaries may dynamically link to different versions of system libraries.

\item Compilers for different architectures may optimize code differently.
\end{enumerate}

\begin{example}
\label{ex:vexir-example}
Going back to Figure~\ref{fig:vexir-example}, on Page~\pageref{fig:vexir-example}, we see the intermediate representation that \angr generates for two binary programs that were compiled from the same source code: one of them runs in x86, the other in ARM.
Noticeable differences between these intermediate representations include:

\begin{itemize}
	\item[-] \textbf{Redundant extensions/truncations:} Variable $t17$ in Figure ~\ref{fig:vexir-example}(b) is extended from $32$ to $64$ bits in \texttt{L2}, followed by truncation back to $32$ bits in \texttt{L8} before any use.

	\item[-] \textbf{Redundant load-stores:}
	In \texttt{L13} of Figure ~\ref{fig:vexir-example}(b), the variable \texttt{t27} is stored in the memory location $M5$, and in \texttt{L21} the same memory location is loaded to a new variable $t34$ and used before any write to the memory location. Such operations are typically introduced during the code generation phase of an SSA-based IR.

	\item[-] \textbf{Special Register Updates:} The updates to special registers like instruction/stack pointer (\texttt{put} operations on $r184$ and $r68$ in Figure ~\ref{fig:vexir-example}(b)/(c)) in a straight-line code.

	\item[-] \textbf{Indirect Memory Accesses:} Indirect memory accesses arising out of architectural semantics like $t5$ in \texttt{L8} of Figure ~\ref{fig:vexir-example}(c) vs. $t20$ in \texttt{L9} of Figure ~\ref{fig:vexir-example}(b).

\end{itemize}
\end{example}

\subsection{Normalized Instruction Traces to Minimize Code Differences}
\label{sub:solution}

The objective of our work is to learn to represent the functions in the binary as embeddings, such that semantically similar ones are mapped onto vectors that are spatially close in a multi-dimensional Euclidean space.
To achieve this goal, we aim to extract program embeddings from a representation that achieves the following properties:

\begin{description}
\item[Shape Sensitiveness:] Instructions that run more often should be given higher importance. 
In other words, instructions that exist in confluence points (the post-dominators of branches) or within the cycles of the CFG should influence the embeddings more heavily.

\item[Normalizable:] 
The effect of redundant instructions as those seen in Example~\ref{ex:vexir-example} should be discarded as much as possible.
In this regard, the instructions should be \textit{normalizable}, meaning that such redundancies should be pruned prior to the construction of the vector representation of programs.

\item[Flow Insensitiveness:] Recent findings have shown that embeddings that map binaries to vectors should be flow-insensitive to resist changes in the order of instructions~\cite{Damasio23,Gorchakov23}.
This requirement is motivated by Example~\ref{ex:diffControlFlow} 
which illustrates how the ordering of instructions and the control-flow graph of programs change once these 
programs are compiled to target different architectures.
\end{description}

To achieve these three goals, our approach builds from the {\it sequences of Basic blocks}, which we term as \textit{peepholes}.
Because this notion is fundamental to the presentation that follows,
Definition~\ref{def:peephole} states it formally.

\begin{definition}[Peephole]
\label{def:peephole}
Let $G = (V, E)$ be the control-flow graph of a program, where $V$ is a set
of {\it basic blocks}---instructions that always run in sequence.
An edge $v_i \rightarrow v_j \in E$ denotes a branch from $v_i, v_j \in V$.
A peephole is a list of $k$ contiguous basic blocks that form a path in $G$.
\end{definition}

%% file: 4.approach.tex
\section{Overview of \vexirtovec{}}
\label{sec:overview}

Figure~\ref{fig:vex-embed} outlines the approach that this paper introduces to solve the binary similarity problem.
Section~\ref{sub:three_steps} provides a brief overview of this approach, and Section~\ref{sub:desirable_properties} explains how it meets the list of requirements previously enumerated in Section~\ref{sub:solution}.

\subsection{From Programs to Vectors to Tasks}
\label{sub:three_steps}

Figure~\ref{fig:vex-embed} illustrates the overview of our pipeline designed to transform the binary representation of a program into a fixed-size vector.
Our approach comprises three main phases:
\begin{enumerate}[label=\Roman*:]
    \item Decomposing functions into a set of normalized peepholes (Section~\ref{sec:phaseI-processing}).
    \item Pre-training the vocabulary to embed these peepholes as vectors (Section~\ref{sec:phaseII-embeddings}).
    \item Fine-tuning the resulting vectors to address downstream tasks such as binary diffing and searching (Section~\ref{sec:phaseIII-model}).
\end{enumerate}

\begin{figure}[!ht]
    \centering
    \resizebox{\textwidth}{!}{
    \includegraphics{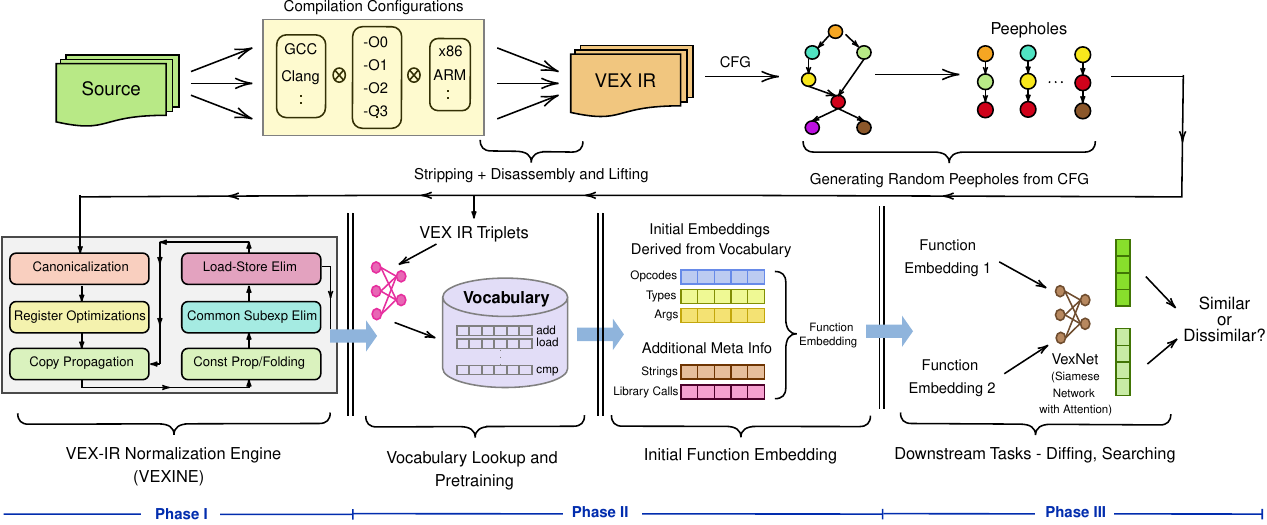}
    }
    \caption{Overview of \vexirtovec{}: The Control-Flow Graphs of the functions from \vexir are generated from different compilation configurations. 
    Then, the peepholes are obtained and normalized by \optengine, the \optenginelong by using different normalizing transformations.
    Function embedding is obtained from the embeddings derived from the Opcodes, Types, and Arguments, along with the strings and external library calls used in the function.
    This embedding is used as input to train \vexnet, a simple feed-forward network in the Siamese setting designed to obtain the final embedding of the function. 
    }
     \Description{Overview of \vexirtovec{}: The Control-Flow Graphs of the functions from \vexir are generated from different compilation configurations. 
    Then, the peepholes are obtained and normalized by \optengine, the \optenginelong by using different normalizing transformations.
    Function embedding is obtained from the embeddings derived from the Opcodes, Types, and Arguments, along with the strings and external library calls used in the function.
    This embedding is used as input to train \vexnet, a simple feed-forward network in the Siamese setting designed to obtain the final embedding of the function.}
    \vspace{-\baselineskip}
    \label{fig:vex-embed}
\end{figure}

\paragraph{Phase I: From Binaries to Normalized Peepholes}
In Phase I, we extract the \vexir from binaries generated with different compiler configurations.
Functions are then decomposed into a set of peepholes derived from their corresponding Control-flow Graphs (CFG), as detailed in Section~\ref{sub:generating-peepholes}.
The resulting intermediate representation (\vexir) undergoes normalization through \optenginelong, \optengine.
This engine applies transformations inspired by traditional compiler optimizations, effectively \textit{normalizing} the \vexir representation of each peephole.
Prior to invoking these passes, we canonicalize the input IR to abstract away details that are unnecessary to solve the binary similarity problem.
Sections~\ref{sub:canonicalization} and~\ref{sub:normalization} elaborate on the normalization engine, its associated normalizations, and its role in decluttering and simplifying the IR.

\paragraph{Phase II: From Peepholes to Embeddings}
In Phase II, we delve into the pre-training process, wherein we learn---without supervision---the vocabulary of \vexir from a corpus of binaries.
This vocabulary facilitates the mapping of the \textit{entities} in the intermediate representation (IR), such as opcodes, types, and arguments to $n$-dimensional vector representations.
This step of learning a vocabulary occurs once in an \textit{offline} manner. Subsequently, we derive $n$-dimensional representations for peepholes and functions using the learned vocabulary.

\paragraph{Phase III: Fine-Tuning Embeddings for Downstream Tasks}
In Phase III, we train \vexnet, a simple feed-forward siamese network with global attention that learns to combine the embeddings generated at the entity level of the functions for addressing the binary similarity problem.
The objective of the model is to represent functions in an Euclidean space, grouping similar functions together while setting apart dissimilar ones.
Once trained, the representations generated by the model can be applied to various downstream tasks, such as binary diffing and searching.

\subsection{Meeting the List of Desirable Properties}
\label{sub:desirable_properties}

The series of steps outlined in Figure~\ref{fig:vex-embed} was conceived as a way to meet the three requirements previously stated in Section~\ref{sub:solution}; namely, shape-sensitiveness, normalization, and flow-insensitiveness.

\paragraph{\textbf{Achieving Shape-Sensitiveness via Random Walks}} We generate peepholes as random walks on the CFG of the functions. Given three parameters: a function $F$, a length of the peephole $k$, and a coverage criterion $c$, peepholes of maximum length $k$ are produced via random walks on the CFG of $F$, where each walk crosses $k$ basic blocks.
Production stops once each basic block is visited at least $c$ times.

\textit{Rationale:} The random walk tends to visit highly connected basic blocks more often. Thus, peepholes are more likely to contain blocks with high in and out degrees in the control-flow graph.

\paragraph{\textbf{Achieving Normalization through Unsound Optimizations}} Before being represented as embeddings, peepholes are normalized using transformations inspired by conventional compiler optimizations.
Currently, we apply the following transformations on each peephole: register optimization, copy propagation, constant propagation and folding, common expression elimination, and load-store elimination.
Our normalizing optimizations are unsound~\cite{Lockwood-Morris-Semantics-10.1145/512927.512941}: they are applied to each peephole; hence, they can eliminate operations that are \textit{dead} within a peephole, but that could be necessary outside it.

\textit{Rationale:}
The objective of \vexirtovec{} is to solve binary similarity. These optimizing transformations act as a means of normalization.
Thus, we want to be able to eliminate as much redundancy from the sequence of instructions that constitute a peephole as possible so as to preserve the essential relations between variables and operations. In this sense, soundness is not essential.

\paragraph{\textbf{Achieving Flow-Insensitiveness through Commutative Accumulation}} Peepholes are transformed into vectors in two steps.
First, each peephole is transformed into a single vector via a {\it mapping
function} that converts instructions to vectors.
Then, all these vectors are added together.

\textit{Rationale:} Addition is commutative; hence, accumulation based on summation ensures that the ordering of instructions plays no hole in the construction of the final vector.
However, instructions that appear in many peepholes (as a result of the random walk) will contribute more to the final vector.

\input{4a-PI-optimization}
\input{4b-PII-embeddings}
\input{4c-PIII-model}

%% file: 4a-PI-optimization.tex
\section{From Binaries to Normalized Peepholes}
\label{sec:phaseI-processing}

In the effort to map binaries to vectors, we first convert them
to a set of peepholes, following Definition~\ref{def:peephole}.
In this section, we explain how we generate such peepholes.
This process happens in multiple steps, starting with the \vexir{} representation of the binary and terminating with a collection of peepholes at the end of Phase I, as shown in Figure~\ref{fig:vex-embed}.
Firstly, Section~\ref{sub:generating-peepholes} shows how the peepholes are generated from the \vexir{} representation of the functions.
Secondly, Section~\ref{sub:canonicalization} explains how peepholes are rewritten to abstract out the different syntax-oriented details that map to the same semantics in \vexir.
Finally, Section~\ref{sub:normalization} explains how this initial set of peepholes is further normalized via a number of transformations.
These transformations, which are inspired by typical compiler optimizations, are applied by our \optenginelong (\optengine).
This final collection of peepholes is the input of Phase II (Figure~\ref{fig:vex-embed}).

\subsection{Generating Peepholes}
\label{sub:generating-peepholes}

\vexir{} programs form {\it Control-Flow Graphs}: graph-like structures whose vertices denote basic blocks, and edges denote possible program flows.
The peepholes of Definition~\ref{def:peephole} is extracted from a control-flow graph via a random walk, which is implemented by Algorithm~\ref{algorithm:random-walk}.
Algorithm~\ref{algorithm:random-walk} translates a function $F$ into a set of peepholes ($\mathcal{P}$), where each peephole $\pi \in \mathcal{P}$ is a sequence of a fixed number ($k$) of basic blocks. 
Algorithm~\ref{algorithm:random-walk} ensures that each basic block in the CFG is \textit{covered} at least $c$ times across $\mathcal{P}$.
As Theorem~\ref{theo:peephole} demonstrates, the expected running time of this algorithm is proportional to the product of this parameter, $c$, and the number of basic blocks in the program.

\begin{algorithm}
\SetAlgoLined\DontPrintSemicolon
\SetKwProg{myalg}{Algorithm}{}{end}
\SetKwRepeat{Do}{do}{while}
    \textbf{Inputs} \\
    \begin{itemize}
    \item $G = (V, E)$: Control-flow graph
    \item $k$: Maximum length of a peephole
    \item $c$: Minimum number of visits per basic block
    \end{itemize}
    \textbf{Output} \\
    \begin{itemize}
    \item $\mathcal{P}$: Set of peepholes
    \end{itemize}
    \textbf{Variables} \\
    \begin{itemize}
        \item $U$: worklist of basic blocks from which the starting block is selected. Initially,  $U = V$ \\
        \item $count$: Number of visits per basic block. Initially, $count[v] = 0$ for every $v \in V$
    \end{itemize}
    \While{$|U| > 0$ }{
        $v$ := randomly sampled start block from $U$ \\
        $\pi$ := random path of length $k$, starting at $v$ \\
        $\mathcal{P}$.append($\pi$) \\

        \For {each block $b\ \in\ \pi $}
        {
            $count[b] := count[b] + 1$ \\
        }
        $U$ := $ \left \{\exists \ v \in V : count[v] < c \right\}$ \\
    }
    \Return $\mathcal{P}$
\caption{Generating peepholes via random walks on the control-flow graph}
\label{algorithm:random-walk}
\end{algorithm}

\begin{theorem}
\label{theo:peephole}
If Algorithm~\ref{algorithm:random-walk} is invoked on a VEX IR function with $|V|$ basic blocks in the control-flow graph and parameters $k$ and $c$, then it terminates in at most $c |V|$ steps.
\end{theorem}

\begin{proof}
The worklist $U$, from which the start basic block is selected at each iteration, is updated to include only the basic blocks $v \in V $ that satisfy the condition $count[v] < c$.
This ensures that at least one basic block in $\pi$ that was not visited $c$ times earlier is visited in each iteration. 
Hence, there is a monotonic decrease in the size of $U$, due to which the algorithm is guaranteed to terminate.
This gives us a worst-case bound of $c|V|$ on the number of steps needed.
The worst case occurs when the start basic block is the only one in $\pi$ that is also in $U$ (that has not been visited $c$
times previously).  
\end{proof}

\begin{observation}
\label{obs:peephole}
From our experiments, we note that the expected number of peepholes generated by using Algorithm~\ref{algorithm:random-walk} is a value close to $c |V|/2$. 
\end{observation}

\begin{rationale}

In practice (from the experiments shown in Appendix~\ref{sec:graph-study}), we observe that as the value of $k$ (maximum length of the peephole) increases, the number of peepholes needed tends to a value close to $c |V|/2$ (as opposed to the worst case $c |V|$). Intuitively, this means that in every step, roughly one more basic block from $U$ is covered apart from the start basic block.

Below, we give an intuitive explanation 
of this. 

\begin{itemize}
    \item The average ratio of the number of edges and the number of blocks from the CFGs used in our dataset is around $1.3$-$1.4$. This implies that the graph is sparse and has the structure of a directed tree with some additional edges.
    \item We also observe that the length of the longest path is about one-third of the number of basic blocks. If the randomly chosen start vertex of the peephole belongs to the first half of the longest path, then we are assured to hit at least $1/6$ fraction of the basic blocks (assuming the peephole length $k$ is long enough). The probability of hitting the first half of a particular longest path is again $1/6$. Hence, with probability $1/6$, we hit at least $|V|/6$ many basic blocks in addition to the starting block. This alone adds a contribution of $|V|/36$ (excluding the starting block) to the expected number of blocks to be covered.
    \item  If the starting block is not a leaf in $G[U]$ (induced graph on the remaining basic blocks), we are assured that there will be at least one more block that will be covered. 
    
\end{itemize}

In our experiments, we set $c=2$, and hence, the expected number of peepholes generated is a value close to $|V|$.
\end{rationale}

The final \vexirtovec{} embeddings are derived from these peepholes.  
By design, the instructions that occur in different peepholes will have higher contributions to the final \vexirtovec{} 
vector---a direct consequence that such instructions occur more often in different \textit{execution-contexts}.
This observation motivates the design of Algorithm~\ref{algorithm:random-walk}, 
which builds peepholes via random walks.
Thus, the basic blocks that are more connected in the function's control-flow graph are likely to appear in more peepholes.
Therefore, they tend to bear more weight on the final vector that represents a
function.
It can be seen that the number of peepholes produced by Algorithm~\ref{algorithm:random-walk} might vary for the same function across different runs due to the randomness induced by the algorithm.
However, the expected number of peepholes can be derived from the parameters $c$ and $k$, as Theorem~\ref{theo:peephole} (and Observation~\ref{obs:peephole}) shows.

\begin{figure}[htb]
\centering
\includegraphics[width=1\linewidth]{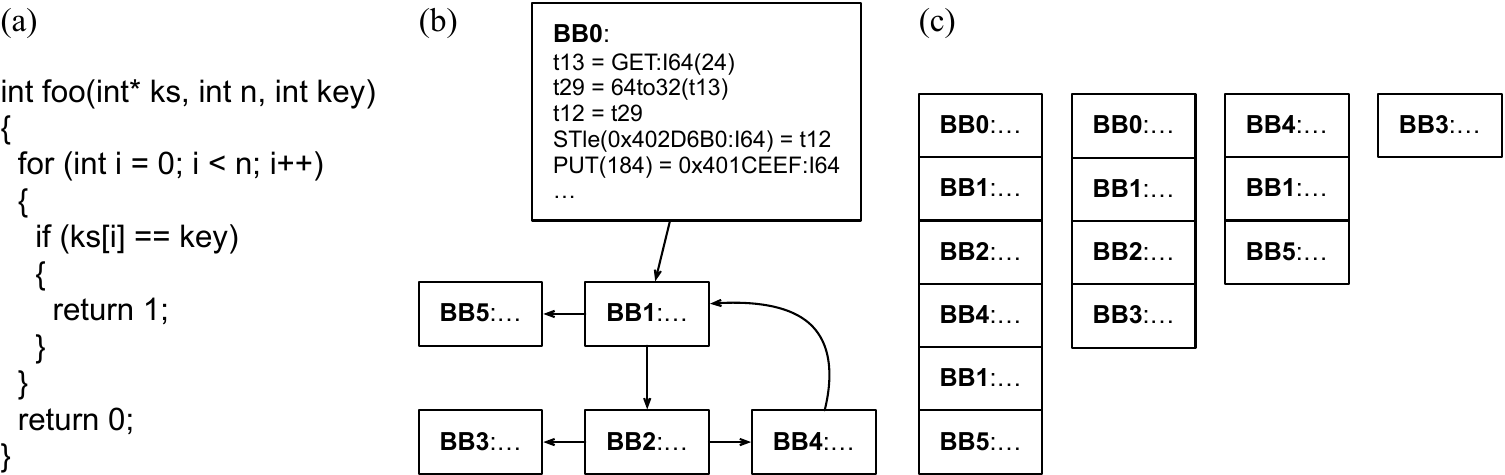}
\caption{(a) Program written in C.
(b) \vexir{} representation of the program, as obtained from an x86 binary.
(c) Peepholes produced by Algorithm~\ref{algorithm:random-walk}, using $k$ = 6 and $c$ = 2.}
\Description{(a) Program written in C.
(b) \vexir{} representation of the program, as obtained from an x86 binary.
(c) Peepholes produced by Algorithm~\ref{algorithm:random-walk}, using $k$ = 6 and $c$ = 2.}
\label{fig:examplePeepholes}
\end{figure}

\begin{example}
\label{ex:peephole1}
Figure~\ref{fig:examplePeepholes} illustrates how peepholes are produced.
This example assumes that Algorithm~\ref{algorithm:random-walk} is invoked on the function shown in Figure~\ref{fig:examplePeepholes} (a), with $k$ = 6 and $c$ = 2.
The peepholes are the sequences of basic blocks seen in Figure~\ref{fig:examplePeepholes} (c).
These peepholes are generated via a random walk on the graph that
Figure~\ref{fig:examplePeepholes} (b) shows.
Notice that the same basic block might appear in different peepholes.
Basic blocks can also appear multiple times in the same peephole if the control-flow graph contains loops.
Thus, highly connected blocks, or blocks in loops, such as
\texttt{BB1} in Figure~\ref{fig:examplePeepholes} (b), are likely to contribute more instances to the peepholes.
\end{example}

\subsection{Canonicalization of Peepholes}
\label{sub:canonicalization}

As seen in Example~\ref{ex:vexir-example}, \vexir{} is a rather prodigal format.
It contains about $1095$ opcodes and $18$ types of operands~\cite{shoshitaishvili2015PyVex}.
Each opcode is specialized for the type and the bit-width of the operation that it implements.
As an illustration, \vexir{} contains about $100$ different ways to write an addition.
This diversity is due to, among other things, types (integer, floating point, etc) and the size of the operands ($8$, $16$, $32$, and $64$ bits).
For example, the ARM instruction \texttt{adds R2, R2, \#8} becomes a sequence of five operations: \texttt{t0 = GET:I32(16); t1 = 0x8:I32; t3 = Add32(t0,t1); PUT(16) = t3; PUT(68) = 0x59FC8:I32}, once converted into \vexir. 

In this paper, vectors derived from \vexir{} instructions form the vocabulary and are used to compare binaries.
The syntactic diversity of instructions that implement the same semantics complicates the generation of this vocabulary.
Ideally, each different syntax should yield different semantics.
Thus, to approximate this goal, we canonicalize the \vexir{} representation of a
program.
This preprocessing step follows a number of simple rewriting rules, which we enumerate as follows:

\begin{enumerate}
\item Opcodes representing the same semantics are replaced with a single opcode.
For instance, we replace opcodes such as \texttt{Add8}, \texttt{Add16} and \texttt{Add32} with just \texttt{Add}.

\item The bitwidths and the endianness of the operands are masked out. For instance, we remove casts such as \texttt{32uto64} from the instructions.

\item Operations that use negative immediate values are converted to operations that use positive constants.
For instance, \texttt{add(-1, t20)} is replaced with \texttt{sub(+1, t20)}.

\item Types are normalized into four primitive types, namely: Integer, Float, Double, and Vector.

\item Following the third and fourth rewriting rules above, all the constants, variables, and registers are abstracted using a corresponding generic representation.

\item Instructions involving indirect memory accesses are replaced with direct memory accesses.
\end{enumerate}

The process of canonicalization of \vexir{} is analogous to the \textit{normalization} step performed in Natural Language Processing~\cite{jurafsky2000nlp-book}.
Canonicalization reduces the number of unique entities to learn, in turn facilitating effective learning.
Example~\ref{ex:canonicalization} illustrates this preprocessing, as done in the context of this work.

\begin{example}
\label{ex:canonicalization}
The canonicalized version corresponding to the example IR shown in Figure ~\ref{fig:vexir-example}(b) is shown in Figure ~\ref{fig:normalization-steps}(a). Notice that the little-endian opcodes \texttt{ldle/stle} are generically represented as \texttt{load/store} while masking out the endianness and bitwidths. 
\end{example}

\begin{figure}
    \centering
    \captionsetup{font=small, skip=8pt}
    \resizebox{0.215\textwidth}{!}{
        \begin{minipage}{0.285\linewidth}
            \captionsetup{font=footnotesize, skip=4pt} 
            \centering
            \input{figures/IR_Optimization/1-canonicalized}
            \caption*{(a) Canonicalized IR} 
        \end{minipage}
    }
    \vrule
    \hspace{0.5em}
    \resizebox{0.215\textwidth}{!}{
        \begin{minipage}{0.29\linewidth}
            \captionsetup{font=footnotesize, skip=4pt} 
            \centering
            \input{figures/IR_Optimization/2-getput-elim}
            \caption*{(b) After Register Optimizations} 
        \end{minipage}
    }
    \hspace{0.7em}
    \vrule
    \hspace{0.7em}
    \resizebox{0.215\textwidth}{!}{
        \begin{minipage}{0.285\linewidth}
            \captionsetup{font=footnotesize, skip=4pt} 
            \centering
            \input{figures/IR_Optimization/3-copy+const-prop}
            \caption*{(c) After Copy and Const prop} 
        \end{minipage}
    }
   \hspace{0.7em}
    \vrule
    \hspace{0.5em}
    \resizebox{0.215\textwidth}{!}{
        \begin{minipage}{0.285\linewidth}
            \captionsetup{font=footnotesize, skip=4pt} 
            \centering
            \input{figures/IR_Optimization/4-offset+expr-prop}
            \caption*{(d) After Common Subexp Elim} 
        \end{minipage}
    }
    \caption{Steps showing the normalizations performed by \optengine on the example shown in Figure ~\ref{fig:vexir-example}(a) for Binary Similarity}.
    \Description{Steps showing the optimizations performed by our Optimization Engine on the example shown in Figure ~\ref{fig:vexir-example}(a) for Binary Similarity}
    \label{fig:normalization-steps}
\end{figure}

The replacement of indirect memory accesses with direct accesses breaks architectural semantics.
However, memory locations are masked out in the process of mapping peepholes to vectors.
Thus, for the purposes of this paper, this kind of replacement is acceptable.
Example~\ref{ex:indirect_access} explains how this transformation is carried out.

\begin{example}
\label{ex:indirect_access}
The \vexir{} program generated from the ARM binary seen in Figure~\ref{fig:vexir-example}(c) shows several indirect accesses.
For instance, the operands $t5$ and $t11$ of the \texttt{add} instruction in \texttt{L8}, are derived from $M1 \rightarrow t2 \rightarrow t5$ and $M4 \rightarrow t8 \rightarrow t11$ respectively.
We replace such accesses by removing additional loads by introducing new memory addresses $M_x$ and $M_y$ as shown in Figure ~\ref{fig:after-normalization} (b). 
\end{example}

\subsection{Normalization of Peepholes by \optengine}
\label{sub:normalization}

Two IRs generated from binaries corresponding to the same source code might be vastly different due to compilation configurations, as explained in Section~\ref{sec:background}.
Additionally, the process of disassembling and lifting from the binary to \vexir{} can also introduce additional redundant instructions that are not of importance for binary similarity.
Thus, to get rid of such \textit{uninteresting} instructions, we normalize the IR via a catalog of transformations that this section discusses.

These transformations follow conventional compiler optimizations~\cite{muchnick1997advanced}.
We perform normalization directly on peepholes, not on functions.
This approach simplifies the implementation of the transformations because the peepholes are straight-line codes.
The normalizations are implemented as part of our \optenginelong (\optengine) module of \vexirtovec{}.
The \optengine{} module takes a peephole as input and produces a new \textit{normalized} version of it.
The implementation of \optengine{} relies on two principles:
\begin{enumerate}
\item Values used in a peephole without being defined within it are assumed to be parameters of the peephole and must be preserved.
\item Values defined within a peephole but not used within it are considered dead and can be eliminated.
\end{enumerate}
The effects produced by \optengine{} have two characteristics:
\begin{description}
\item [Local:] Optimizations affect only one occurrence of an instruction, even if this instruction appears in multiple peepholes. In other words, the elimination of an instruction from a peephole bears no effect on the existence of this instruction within other peepholes.
\item [Unsound:] Due to the local nature of the optimizations, they are unsound. In other words, we remove instructions from a peephole without worrying if these instructions could be used outside the peephole.
\end{description}

In the rest of this section, we describe different normalizations that \optengine{} performs on the canonicalized peepholes.

\begin{description}

\item [Register Promotion:] Memory addresses accessed via the \texttt{put} and \texttt{get} instructions are promoted to registers.
Notice that this transformation breaks the static single-assignment property
of \vexir. Consequently, a reaching-definition analysis becomes necessary to implement the transformations that follow.
It can be seen that the register promotion exposes redundant, useless writes to the same register.

\item [Redundant Write Elimination:] Redundant writes to the same register is removed. This transformation, guided by a reaching-definitions analysis, is useful in reducing the number of updates to special registers like instruction and stack pointers.

\item [Copy Propagation:] Redundant copy instructions are eliminated after a reaching-definition analysis by replacing copies with their sources.
As an example, Lines 6 and 7 in Figure ~\ref{fig:normalization-steps} (b) are replaced with a single copy.

\item [Constant Propagation and Folding:] Variables initialized with constants are eliminated, and the constants are written directly at their use sites.
Expressions that use only constants are replaced with the evaluated result.

\item [Common Subexpression Elimination:] Multiple evaluations of the same expression are replaced with a single variable with the one that holds the first computation of that variable.
This transformation is guided by an available-expression analysis to ensure that the value computed by an expression does not change in between computations.
Example~\ref{ex:opt-steps} provides more details about this transformation.

\item [Load-Store Elimination:] Redundant store instructions in pairs of load and
store operations are eliminated.
A pair \texttt{t = load(M); store(M) = t} is deemed redundant if \texttt{t} is
not updated in between the load and the store.
This transformation is guided by a reaching-definition analysis, as the representation is no longer in the static single-assignment format.

\item [Store-Store Elimination:] Two consecutive stores to the same memory location without any reads in between cause the first store instruction to be removed.

\end{description}

\begin{example}
\label{ex:opt-steps}
Figure~\ref{fig:normalization-steps} shows the successive effects of different normalizations.
First, register promotion and redundant write elimination removes multiple instances of \texttt{put}
instructions; hence, converting Figure~\ref{fig:normalization-steps} (a)
into  Figure~\ref{fig:normalization-steps} (b).
Only the last effect of such instructions, in Line \texttt{L23}, is preserved.
Secondly, a round of copy propagation maps Figure~\ref{fig:normalization-steps} (b) onto Figure~\ref{fig:normalization-steps} (c).
Chained sequences of assignments are eliminated, and definitions are
directly used where they are needed.
As an example, variable \texttt{t17}, defined in Line \texttt{L1}, is directly used in Line \texttt{L9}.
Finally, Figure~\ref{fig:normalization-steps} (d) illustrates a usage of common subexpression elimination.
The redundant \texttt{load} on \texttt{M3} in \texttt{L15} is eliminated, and the use of \texttt{t30} in \texttt{L19} is replaced with \texttt{t19} defined in \texttt{L5}.
\end{example}

The combination of canonicalization rules and normalizations substantially simplifies the instructions in the peepholes that are used to produce vectors.
The resultant codes, even when obtained from different instruction sets, tend to be more similar.

\begin{figure}[ht]
    \centering
    \captionsetup{font=small, skip=8pt}
    \begin{minipage}{0.45\linewidth}
        \captionsetup{font=footnotesize, skip=4pt} 
        \centering
        \input{figures/IR_Optimization/5-load-store-elim}
        \caption*{(a) x86} 
    \end{minipage}
    \hspace{0.7em}
    \vrule
    \hspace{0.5em}
    \begin{minipage}{0.45\linewidth}
        \captionsetup{font=footnotesize, skip=4pt} 
        \centering
        \input{figures/IR_Optimization/arm-normalized}
        \caption*{(b) ARM} 
    \end{minipage}
    \caption{Normalized \vexir generated by \optengine for the example shown in Figure ~\ref{fig:vexir-example}.}
    \Description{Normalized \vexir generated by the \optengine for the example shown in Figure ~\ref{fig:vexir-example}.}
    \label{fig:after-normalization}
\end{figure}

\begin{example}
\label{ex:final_optimization}
Figure~\ref{fig:after-normalization} shows the x86 and the ARM versions of the peepholes earlier seen in Figure~\ref{fig:vexir-example}.
As it can be seen, the syntactic gap between these two programs is substantially shorter than the syntactic difference between the two peepholes seen in
Figure~\ref{fig:vexir-example}.
\end{example}

%% file: figures/IR_Optimization/1-canonicalized.tex
\begin{lstlisting}
t17 = load:i32(M1)
t16 = t17
put(r32) = t16
put(r184) = M2
t19 = load:i32(M3)
t18 = t19
t20 = t18
t22 = t16
t2 = add(t20,t22)
t26 = t2
put(r184) = M4
t27 = t26
store(M5) = t27
put(r184) = M6
t30 = load:i32(M3)
t29 = t30
put(r184) = M7
t31 = t29
store(M1) = t31
put(r184) = M8
t34 = load:i32(M5)
t33 = t34
put(r184) = M9
t35 = t33
store(M3) = t35
\end{lstlisting}

%% file: figures/IR_Optimization/2-getput-elim.tex
\begin{lstlisting}[countblanklines=false]
t17 = load:i32(M1)
t16 = t17
r32 = t16
          |\makebox[-6pt]{\color{lightgray}\rule[0ex]{16em}{2ex}}|
t19 = load:i32(M3)
t18 = t19
t20 = t18
t22 = t16
t2 = add(t20,t22)
t26 = t2
          |\makebox[-6pt]{\color{lightgray}\rule[0ex]{16em}{2ex}}|
t27 = t26
store(M5) = t27
          |\makebox[-6pt]{\color{lightgray}\rule[0ex]{16em}{2ex}}|
t30 = load:i32(M3)
t29 = t30
          |\makebox[-6pt]{\color{lightgray}\rule[0ex]{16em}{2ex}}|
t31 = t29
store(M1) = t31
          |\makebox[-6pt]{\color{lightgray}\rule[0ex]{16em}{2ex}}|
t34 = load:i32(M5)
t33 = t34
r184 = M9
t35 = t33
store(M3) = t35
\end{lstlisting}

%% file: figures/IR_Optimization/3-copy+const-prop.tex
\begin{lstlisting}[countblanklines=false]
t17 = load:i32(M1)
          |\makebox[-6pt]{\color{lightgray}\rule[0ex]{16em}{2ex}}|
r32 = t17
          |\makebox[-6pt]{\color{lightgray}\rule[0ex]{16em}{2ex}}|
t19 = load:i32(M3)
          |\makebox[-6pt]{\color{lightgray}\rule[0ex]{16em}{2ex}}|
          |\makebox[-6pt]{\color{lightgray}\rule[0ex]{16em}{2ex}}|
          |\makebox[-6pt]{\color{lightgray}\rule[0ex]{16em}{2ex}}|
t2 = add(t19,t17)
          |\makebox[-6pt]{\color{lightgray}\rule[0ex]{16em}{2ex}}|
          |\makebox[-6pt]{\color{lightgray}\rule[0ex]{16em}{2ex}}|
          |\makebox[-6pt]{\color{lightgray}\rule[0ex]{16em}{2ex}}|
store(M5) = t2
          |\makebox[-6pt]{\color{lightgray}\rule[0ex]{16em}{2ex}}|
t30 = load:i32(M3)
          |\makebox[-6pt]{\color{lightgray}\rule[0ex]{16em}{2ex}}|
          |\makebox[-6pt]{\color{lightgray}\rule[0ex]{16em}{2ex}}|
          |\makebox[-6pt]{\color{lightgray}\rule[0ex]{16em}{2ex}}|
store(M1) = t30
          |\makebox[-6pt]{\color{lightgray}\rule[0ex]{16em}{2ex}}|
t34 = load:i32(M5)
          |\makebox[-6pt]{\color{lightgray}\rule[0ex]{16em}{2ex}}|
r184 = M9
          |\makebox[-6pt]{\color{lightgray}\rule[0ex]{16em}{2ex}}|
store(M3) = t34
\end{lstlisting}

%% file: figures/IR_Optimization/4-offset+expr-prop.tex
\begin{lstlisting}[countblanklines=false]
t17 = load:i32(M1)
          |\makebox[-6pt]{\color{lightgray}\rule[0ex]{16em}{2ex}}|
r32 = t17
          |\makebox[-6pt]{\color{lightgray}\rule[0ex]{16em}{2ex}}|
t19 = load:i32(M3)
          |\makebox[-6pt]{\color{lightgray}\rule[0ex]{16em}{2ex}}|
          |\makebox[-6pt]{\color{lightgray}\rule[0ex]{16em}{2ex}}|
          |\makebox[-6pt]{\color{lightgray}\rule[0ex]{16em}{2ex}}|
t2 = add(t19,t17)
          |\makebox[-6pt]{\color{lightgray}\rule[0ex]{16em}{2ex}}|
          |\makebox[-6pt]{\color{lightgray}\rule[0ex]{16em}{2ex}}|
          |\makebox[-6pt]{\color{lightgray}\rule[0ex]{16em}{2ex}}|
store(M5) = t2
          |\makebox[-6pt]{\color{lightgray}\rule[0ex]{16em}{2ex}}|
          |\makebox[-6pt]{\color{lightgray}\rule[0ex]{16em}{2ex}}|
          |\makebox[-6pt]{\color{lightgray}\rule[0ex]{16em}{2ex}}|
          |\makebox[-6pt]{\color{lightgray}\rule[0ex]{16em}{2ex}}|
          |\makebox[-6pt]{\color{lightgray}\rule[0ex]{16em}{2ex}}|
store(M1) = t19
          |\makebox[-6pt]{\color{lightgray}\rule[0ex]{16em}{2ex}}|
t34 = load:i32(M5)
          |\makebox[-6pt]{\color{lightgray}\rule[0ex]{16em}{2ex}}|
r184 = M9
          |\makebox[-6pt]{\color{lightgray}\rule[0ex]{16em}{2ex}}|
store(M3) = t34
\end{lstlisting}

%% file: figures/IR_Optimization/5-load-store-elim.tex
\begin{lstlisting}[belowskip=-0.6 \baselineskip, backgroundcolor = \color{lightorange}]
t17 = load:i32(M1)
r32 = t17
t19 = load:i32(M3)
t2 = add(t19,t17)
\end{lstlisting}

\begin{lstlisting}[belowskip=-0.6 \baselineskip, backgroundcolor = \color{lightblue},firstnumber=5]
store(M1) = t19
r184 = M9
\end{lstlisting}

\begin{lstlisting}[backgroundcolor = \color{lightyellow},firstnumber=7]
store(M3) = t2
\end{lstlisting}

%% file: figures/IR_Optimization/arm-normalized.tex
\begin{lstlisting}[belowskip=-0.55 \baselineskip, backgroundcolor = \color{lightorange}]
t5 = load:i32(Mx)
t11 = load:i32(My)
t15 = add(t5,t11)
\end{lstlisting}

\begin{lstlisting}[belowskip=-0.55 \baselineskip, backgroundcolor = \color{lightblue}, firstnumber=4]
store(Mz) = t15
r68 = Mp
\end{lstlisting}

\begin{lstlisting}[backgroundcolor = \color{lightyellow},firstnumber=6]
store(Mx) = t11
\end{lstlisting}

%% file: 4b-PII-embeddings.tex
\section{Pre-Training: From Peepholes to Embeddings}
\label{sec:phaseII-embeddings}

The peepholes obtained from the previous step are projected onto an $n$-dimensional Euclidean space as a vector of real numbers. This $n$-d embedding is suitable for geometric, specifically Euclidean distance comparison. Given a normalized peephole, the goal is to learn an embedding such that vectors in close proximity in the Euclidean space are likely to be semantically similar; conversely, if the vectors are distant, the functions are likely to be semantically different. This rationale is illustrated in Figure~\ref{fig:exVectorSpace}. Upon learning such a projection, the problem of binary similarity of functions is reduced to finding the nearest neighbors of the function in the Euclidean space.

\begin{figure}[htb]
\centering
\includegraphics[width=1\linewidth]{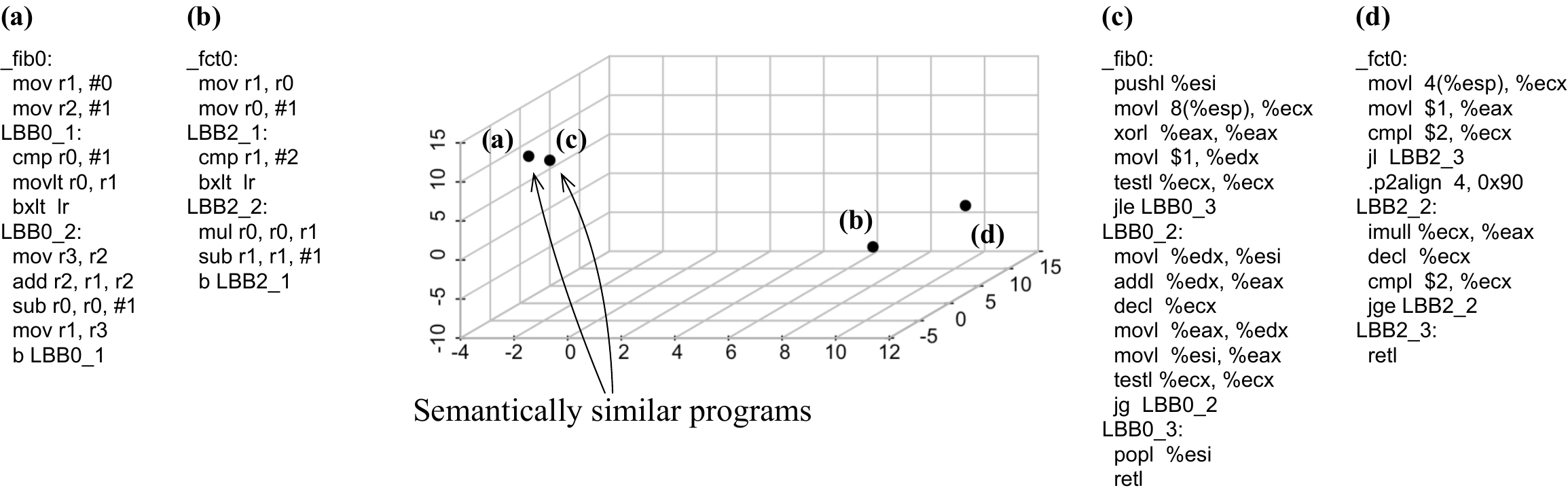}
\caption{An analogy emphasizing the geometric nature of \vexirtovec: similar binary programs are projected onto points that are spatially close.
This example uses a three-dimensional Euclidean space.
The actual implementation of \vexirtovec{} uses an Euclidean space of $128$ dimensions.}
\Description{An analogy emphasizing the geometric nature of \vexirtovec: similar programs are projected onto points that are spatially close.
This example uses a three-dimensional Euclidean space.
The actual implementation of \vexirtovec{} uses an Euclidean space of $128$ dimensions.}
\label{fig:exVectorSpace}
\end{figure}

Learning to represent functions as embeddings involves two steps: lightweight pre-training and fine-tuning.
The first stage, pre-training, consists of learning an initial \textit{task-independent} vocabulary function in an \textit{unsupervised} manner.
The second stage, fine-tuning, consists of constructing the final function embedding using the vocabulary function.

\paragraph{Knowledge Representation via Knowledge Graphs}
In the pre-training step, we learn a vocabulary of embeddings for each entity in the \vexir instruction. This is done by modeling instructions from a corpus of binaries as a simple Knowledge Graph. A knowledge graph~\cite{paulheim2017knowledge} is a graph that captures the relationships between the entities that form the system under analysis.
In our case, such entities are syntactic IR constructs: opcodes, types, values, etc.
Typically, a knowledge graph models \textit{entities} as vertices and the edges as the \textit{relationship} between the neighbors.

Upon obtaining the knowledge graph, the representations of entities and relationships can be derived by using knowledge graph embedding approaches~\cite{knowledge-graph-embedding-survey}.
The effectiveness of knowledge graph embeddings is well-studied in conventional NLP applications.
Additionally, previous work has demonstrated that this effectiveness remains consequential in the context of learning program embeddings~\cite{VenkataKeerthy-2020-IR2Vec}.
In this work, we propose to use knowledge graph embeddings to model the \vexir instructions.

Formally, a knowledge graph $\mathcal{G}$ can be represented as a set of triplets $\langle h, r, t \rangle$, where $h$ and $t$ denote the \textit{head} and \textit{tail} entities connected by a \textit{relation} $r$.
The embeddings of $h$, $r$, and $t$ in an $n$-d space are learned as translations from the head entity $h$ to the tail entity $t$ using the relationship $r$.
Several knowledge graph embedding strategies have been proposed to represent the entities and relations in a continuous $n$-d space~\cite{knowledge-graph-embedding-survey}.
These approaches learn a scoring function that returns a high value if $\langle h, r, t \rangle$ relation holds in $\mathcal{G}$, and a low value otherwise. 
This process results in learning the representations of the triplets (\embed{h} $\in \mathbb{R}^n$, \embed{r} $\in \mathbb{R}^n$ and \embed{t} $\in \mathbb{R}^n$)\footnote{Following previous work, we use \embed{.} to denote an embedding, i.e., a vector in an n-dimensional space.} in an $n$-d space, where the third component of the triplet could be identified given the other two components using vector operations.
We use a Translational Embedding model TransE~\cite{transe-Bordes:2013:TEM:2999792.2999923} to learn the representations of entities of \vexir.

\subsection{Mapping \vexir as Entities and Relations}
After normalization, the variables and constants are masked out with abstract placeholder tokens such as VAR and CONST.
We decompose the normalized \vexir instructions into entities and relations. Such code entities, in turn, are features that characterize program instructions.
Each entity is represented as a triplet formed by an {\it opcode}, a {\it kind}, and a {\it value}.
In this work, we recognize three \textit{kinds} of code entities, which we list as follows:

\begin{description}
\item [Arg:] the entity describes the argument of an instruction; e.g., $\langle \mathit{add}, \mathbf{Arg}_1, \mathit{VAR} \rangle$ indicates that the first argument of an addition operation is a variable (instead of a constant or a memory address, for instance).
\item [Type:] the entity describes the type of instruction; e.g., $\langle \mathit{add}, \mathbf{Type}, \mathit{INT} \rangle$ indicates that an addition instruction operates on integer values.
\item[Next:] the entity describes the successor relation between instructions; e.g.: $\langle \mathit{add}, \mathbf{Next}, \mathit{store} \rangle$ indicates that a store instruction follows the add instruction.
\end{description}

These three categories capture the relationships between two entities at a time, resulting in $\langle h, r, t \rangle$ triplets.
Such triplets are collected over a corpus of \vexir functions and are fed as an input to an ML model to learn the embeddings of the entities in the corpus. Such representation of entities results in a vocabulary.

\begin{example}
Figure~\ref{fig:peephole2Vex} shows the code entities derived from two instructions.
As it can be observed, each instruction is decomposed into multiple $\langle h, r, t \rangle$ triplets using opcodes, kind (type, arg or next), and values.
The decomposition captures the possible relationships defined by \textit{kind}.
For instance, entities of kind ``type'' capture the relation between an operation (ex.: a load or addition) and the type of the value that the operation manipulates (ex.: an integer).
\end{example}
\label{ex:peep2vec}

\begin{figure}[htb]
    \vspace{-0.5\baselineskip}
    \centering
    \includegraphics[width=0.95\linewidth]{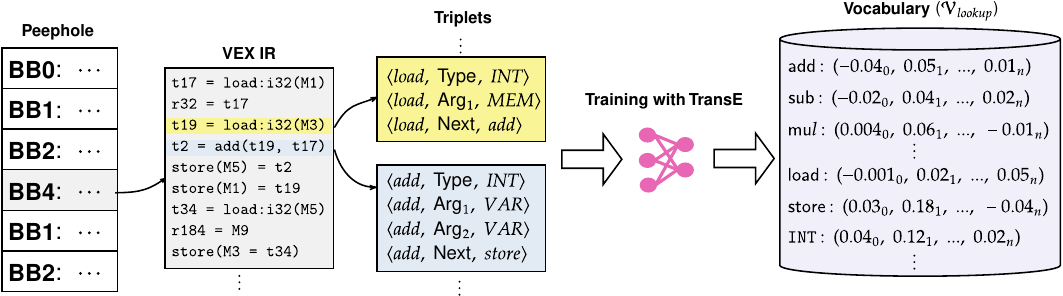}
    \caption{\vexir instructions from a corpus of binaries are decomposed into entities and relations to obtain triplets. These triplets are used to train the TransE model to obtain a vocabulary. 
    The vocabulary contains the $n$ dimensional representations of each of the entities of \vexir in the train set.}
    \Description{\vexir instructions from a corpus of binaries are decomposed into entities and relations to obtain triplets.}
    \label{fig:peephole2Vex}
    \vspace{-0.5\baselineskip}
\end{figure}

\subsection{The Vocabulary Function}
\label{sub:lookup}

A key element in \vexirtovec{} is the vocabulary function, $\mathcal{V}_{lookup}$.
This function maps code entities to vectors.
We use a Translational Embedding model TransE~\cite{transe-Bordes:2013:TEM:2999792.2999923} to learn the representations of entities of \vexir. The model, in this case, learns the relationship $\embedEqn{h} + \embedEqn{r} \approx \embedEqn{t}$, by minimizing the energy function, $E(h, r, t) = \left \Vert \embedEqn{h} + \embedEqn{r} - \embedEqn{t}\right \Vert_2^2$.
The $\approx$ notation indicates that the embedding of $h$ is closer to the embedding of $t$ upon adding the embedding of the relation $r$.
As relation $r$ is used as a translation from the entity $h$ to another entity $t$, this approach is called TransE.
We train a model to minimize $E$ using a margin-based ranking loss $\mathcal{L}$ given by the following equation.

\begin{equation}
\label{eq:lookup}
\mathcal{L} = \sum_{\langle h, r, t \rangle \in \mathcal{G}} \sum_{\langle h', r, t' \rangle \not\in \mathcal{G}} \mbox{max} \left( 0, \gamma + d(\embedEqn{h}+\embedEqn{r}, \embedEqn{t}) - d(\embedEqn{h'}+\embedEqn{r}, \embedEqn{t'}) \right)
\end{equation}

Equation~\ref{eq:lookup} ranges on a {\it Knowledge Graph} $\mathcal{G}$ created from the corpus of \vexir instructions described earlier.
To minimize the loss function $\mathcal{L}$, we map code entities $\triplet{h}{r}{t}$ to vectors that approximate the triangle equality; that is, $\embedEqn{h} + \embedEqn{r} \approx \embedEqn{t} \implies d(\embedEqn{h}+\embedEqn{r}, \embedEqn{t}) \approx 0$, where $d$ is Eucledian distance.
Likewise, if $\langle h, r, t\rangle \notin \mathcal{G}$, then it follows that $\embedEqn{h} + \embedEqn{r} \not\approx \embedEqn{t} \implies d(\embedEqn{h'} +\embedEqn{r}, \embedEqn{t}') > 0$.
Figure~\ref{fig:derivativeBasedProgramming} provides an intuition on this objective function.

\begin{figure}[htb]
    \centering
    \includegraphics[width=0.95\linewidth]{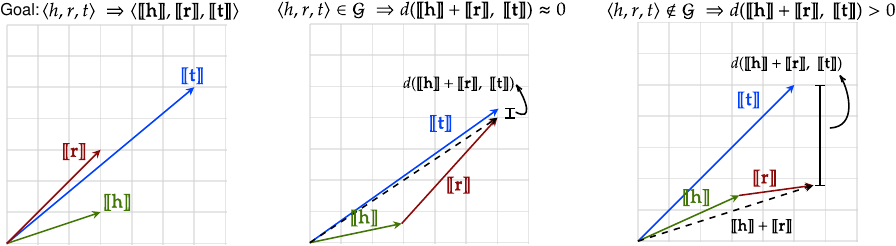}
    \caption{The construction of the vocabulary function $\mathcal{V}_{lookup}$ involves mapping entities $\langle h, r, t\rangle$ into triplets of vectors in a way that minimizes the distance between the vectors $\embedEqn{t}$ and
    $\embedEqn{h} + \embedEqn{r}$ whenever the triplet $\langle h, r, t\rangle$ belongs to the training set.}
    \Description{The construction of the vocabulary function $\mathcal{V}_{lookup}$ involves mapping entities $\langle h, r, t\rangle$ into triplets of vectors in a way that minimizes the distance between the vectors $\embedEqn{t}$ and
    $\embedEqn{h} + \embedEqn{r}$ whenever the triplet $\langle h, r, t\rangle$ belongs to the training set.}
    \label{fig:derivativeBasedProgramming}
    \vspace{-0.5\baselineskip}
\end{figure}

Additionally, $\mathcal{L}$ ensures that $d(\embedEqn{h}+\embedEqn{r}, \embedEqn{t})$ and $d(\embedEqn{h'}+\embedEqn{r}, \embedEqn{t'})$ are at least separated by a margin $\gamma$.
At the end of this learning process, we obtain a vocabulary $\mathcal{V}_{lookup}$ containing the representation of each entity of \vexir in an $n$-d Euclidean space $\mathbb{R}^n$. 
In the next section we shall explain how we derive the embedding of a program function out of the embedding of every entity that makes up this function.

\subsection{Mapping Functions to Vectors}
\label{sub:map}

As described in Section~\ref{sec:phaseI-processing}, each program function is decomposed into a set of peepholes after Algorithm~\ref{algorithm:random-walk}.
Each peephole represents a sequence of basic blocks.
Basic blocks are sequences of \vexir instructions.
Each one of these instructions can be decomposed into the following entities: Opcode, Type, and Arguments.
Thus, to construct the \vexir representation of a function, we combine the representation of these entities in the function's peepholes.
As seen in Section~\ref{sub:lookup}, the representation of each entity is obtained using the $\mathcal{V}_{lookup}$ map.

If a \vexir instruction $l$ is of format $[O^{(l)} T^{(l)} A_1^{(l)} A_2^{(l)} \cdots A_n^{(l)}]$, where $O^{(l)}$, $T^{(l)}$ and $A_i^{(l)}$ represent the opcode, type, and its $i^{th}$ argument, then its representation is obtained as $\embedEqn{O}[^{(l)}] = \vocab{O^{(l)}}, \embedEqn{T}[^{(l)}] = \vocab{T^{(l)}}$, and $\embedEqn{A_i}[^{(l)}] = \vocab{A_i^{(l)}}$. 
The instruction is embedded as a combination of the embeddings of these entities that form it:

\begin{equation}
\label{eq:entities}
    \llangle \embedEqn{O}[^{(l)}], \embedEqn{T}[^{(l)}], \embedEqn{A}[_1^{(l)}] + \embedEqn{A}[_2^{(l)}] + \cdots + \embedEqn{A}[_n^{(l)}] \rrangle
\end{equation}

If the function $F$ is decomposed into a set of normalized peepholes $\mathcal{P} = \{\pi_1, \pi_2, \ldots\}$, the representation of a peephole $\pi \in \mathcal{P}$ containing a set of instructions $I^\pi$ is computed as the pointwise addition of the representation of individual entities of its instructions in $I^\pi$.
\begin{equation}
     \embedEqn{\pi} = \left\langle \sum_{I \in I^\pi} \embedEqn{O}[^I], \sum_{I \in I^\pi} \embedEqn{T}[^I], \sum_{I \in I^\pi} \embedEqn{A}[^I] \right\rangle = \left\langle \embedEqn{O}[^\pi], \embedEqn{T}[^\pi], \embedEqn{A}[^\pi]  \right\rangle 
\end{equation}

The representation $\embedFInit$ of a function $F$ is computed by
summing the representations of entities across different peepholes:

\begin{equation}
    \embedFInit = \left \langle \sum_{\pi \in \mathcal{P}}\embedEqn{O}[^\pi], \sum_{\pi \in \mathcal{P}}\embedEqn{T}[^\pi], \sum_{\pi \in \mathcal{P}}\embedEqn{A}[^\pi]   \right\rangle 
\end{equation}

Therefore, $\embedFInit$, the embedding vector representing function $F$, is produced as a linear combination of the embeddings of the peepholes produced via Algorithm~\ref{algorithm:random-walk}.
This process of mapping functions to embeddings, parameterized by the $\mathcal{V}_{lookup}$ function, is summarized in Algorithm~\ref{algorithm:vec_gen}.

\begin{algorithm}
\DontPrintSemicolon
\SetAlgoNoLine

\textbf{Inputs} \\
\begin{itemize}
\item $\mathcal{P}$: set of normalized peepholes of the function $F$
\end{itemize}

\textbf{Output} \\
\begin{itemize}
\item $\embedFInit = \langle \embedEqn{O}, \embedEqn{T}, \embedEqn{A} \rangle$: A tuple of $n$ dimensional vectors corresponding to opcodes, types and arguments.
\end{itemize}

$\embedEqn{O}, \embedEqn{T}, \embedEqn{A} = \vec{0}$   // Initialize vectors to zero \\
\For{peephole $\pi \in \mathcal{P}$}{
    \For{instruction $l \in I_\pi$}{
        \embed{O} = \embed{O} + $\vocab{O^{(l)}}$\\
        \embed{T} = \embed{T} + $\vocab{T^{(l)}}$\\
        \For{Arg $A_i^{(l)} \in [A_1^{(l)} A_2^{(l)} \cdots A_n^{(l)}]$}{
            \embed{A} = \embed{A} + $\vocab{A_i^{(l)}}$\\
        }
    }
}

\Return $\langle \embedEqn{O}, \embedEqn{T}, \embedEqn{A} \rangle$\ as\ $\embedFInit$
\caption{Mapping functions to vectors.}
\label{algorithm:vec_gen}
\end{algorithm}

\paragraph{Function calls}
We also consider the function calls while obtaining the function vector ($\embedFInit$).
We compute the call graph of the binary. For each function call, the embedding of the callee is first computed and is used in the call site in place of the call instruction. 
This way of substituting the callee embedding in the call site has a similar effect to that of the inlining optimization performed during the compilation process. 
This approach also provides additional information while handling wrapper functions. 
We follow the normal embedding process for representing call instructions if the definition of the callee is not available in the current binary.

%% file: 4c-PIII-model.tex
\section{Fine-Tuning: Training Embeddings for Computing Similarity}
\label{sec:phaseIII-model}

The combined application of Algorithms~\ref{algorithm:random-walk} and~\ref{algorithm:vec_gen} on a function yields its initial representation $\embedFInit$. 
As seen in Algorithm~\ref{algorithm:vec_gen}, the initial function representation $\embedFInit$ is derived from the vocabulary function $\mathcal{V}_{lookup}$.
The construction of $\mathcal{V}_{lookup}$ depends only on the dataset corpus and is trained in an unsupervised manner; thus, it does not depend on any particular task to be solved by \vexirtovec.
Therefore, to be effectively used, the representation of $\embedFInit$ must be fine-tuned to obtain the actual \vexirtovec{} embeddings to solve particular tasks, such as binary diffing or searching.
This section explains this process of fine-tuning.

\subsection{\vexnet - Siamese Network with EAN modules}
\label{sub:vexnet}

This process of fine-tuning depends on a model that learns to combine the entities $\llangle \embedEqn{O}, \embedEqn{T}, \embedEqn{A} \rrangle$, along with metadata of the function like strings and external library calls.
This resultant embedding is projected on an embedding space in such a way that semantically similar functions are spatially closer than dissimilar functions.
We model this process of fine-tuning using the Siamese network~\cite{koch2015siamese} seen in Figure~\ref{fig:finetuning-network}.
This model is constructed using two identical Entity-Attention Network (EAN) modules.
The rest of this section explains this model, which we call \vexnet.

\begin{figure}[htb]
    \captionsetup[subfigure]{font=footnotesize,labelfont=footnotesize}
    \centering
    \subfloat[\footnotesize{Entity-Attention Network}]{{\includegraphics[width=0.55\linewidth, valign=M]{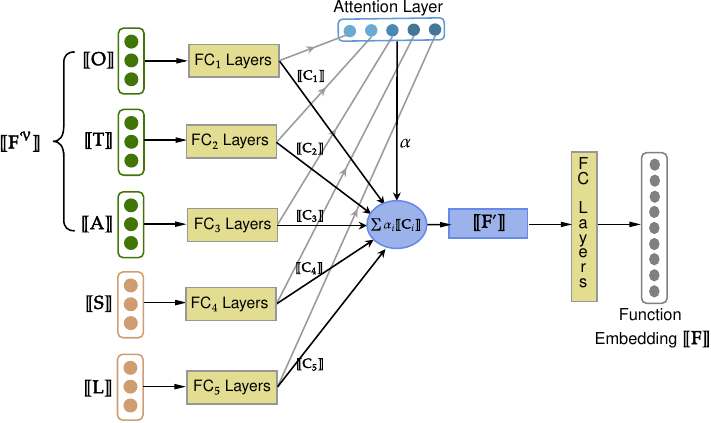} }}%
    \qquad
    \subfloat[\footnotesize{\vexnet - Siamese Configuration}]{{\includegraphics[width=0.37\linewidth, valign=M]{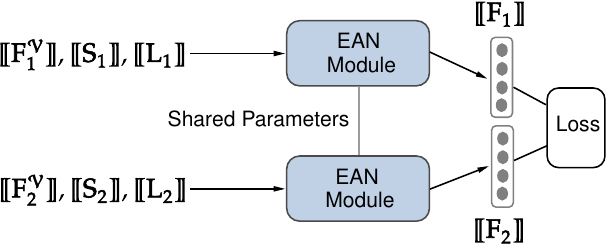} }}
    
    \caption{\vexnet - Siamese network used for fine-tuning the embeddings. This network consists of an Entity-Attention module that learns to combine $\langle \embedEqn{O}, \embedEqn{T}, \embedEqn{A} \rangle$ and $\embedEqn{S}, \embedEqn{L}$ to obtain a function embedding vector $\embedFFinal$ to learn similarity metric.}
    \Description{\vexnet - Siamese network Configuration}
    \label{fig:finetuning-network}
\end{figure}

A Siamese network is a neural network configuration that contains two identical subnetworks that share the same learning parameters and weights~\cite{koch2015siamese}.
We use the EAN module (Section~\ref{sub:ean}) in the Siamese configuration to learn an embedding $\embedFFinal$ that differentiates similar and dissimilar binary functions.
Thus, our training set consists of a collection of pairs of similar and dissimilar functions.
In other words, pairs of binary codes obtained from the same source but compiled with different compilers, optimization levels, or architecture targets constitute similar pairs.
Pairs that do not originate from the same source, in turn, form a set of dissimilar samples.
Normalized temperature-scaled cross-entropy loss (NT-Xent)~\cite{chen2020ContrastiveLearning} is used to train the model.
The model is trained end-to-end to fine-tune the initial embeddings $\embedFInit$ derived from the vocabulary to obtain the final \vexirtovec embeddings $\embedFFinal$ of the function.
This training process ensures that the embeddings of similar functions are positioned closer in the embedding space and the dissimilar functions are pushed farther apart.

\subsubsection{Obtaining strings and library call representations \embed{S} and \embed{L}}
\label{sss:strings}

Strings and calls to the external library methods from \textit{libc} and \textit{libstdc++} seldom change with the variations in compilers, optimizations, and the target architecture.
Hence, they can act as valuable additional information in the binary similarity analysis.
Our approach leverages this information to enhance the effectiveness of the solution.
To incorporate this data into our model, each string used in a function is represented using embeddings generated by a pre-trained model.
In this work, we use a lightweight fastText model~\cite{bojanowsk-2017-fasttext} to obtain embeddings in $\mathbb{R}^{n}$. The embeddings for each string are summed to produce a single representation for all strings used in the function, denoted as \embed{S}. A similar process is applied to represent the library calls. The function names of these calls are treated as string tokens, and their representations are obtained using the fastText model. As with the strings, the representation of library calls, \embed{L}, is obtained by summing the individual embeddings.

\subsection{Entity-Attention Network}
\label{sub:ean}

The EAN module (Figure~\ref{fig:finetuning-network}) is a simple feed-forward neural network with a global attention~\cite{bahdanau2015globalAttention} layer.
This global attention layer combines the entities of the initial embeddings $\embedFInit = \llangle \embedEqn{O}, \embedEqn{T}, \embedEqn{A} \rrangle$ in $\mathbb{R}^{3n}$ with the strings \embed{S} and library calls \embed{L} in $\mathbb{R}^{2n}$, projecting the resulting vector onto $\mathbb{R}^{n}$.
It uses this projection to obtain the final function embedding $\embedFFinal$ in  $\mathbb{R}^{n}$.
Embeddings of the entities and metadata are passed through a set of independent, Fully Connected (FC) layers. These layers process the raw input embeddings derived from the vocabulary and metadata into context vectors ($\embedEqn{C}[_{1}], \embedEqn{C}[_{2}], \embedEqn{C}[_{3}], \embedEqn{C}[_{4}], \embedEqn{C}[_{5}]$).
The final context vector \embed{C} is produced as follows:

\begin{equation*}
    \embedEqn{C}[_{1}] = \text{FC}_1(\embedEqn{O});\quad \embedEqn{C}[_{2}] = \text{FC}_2(\embedEqn{T});\quad \embedEqn{C}[_{3}] = \text{FC}_3(\embedEqn{A});
    \quad \embedEqn{C}[_{4}] = \text{FC}_4(\embedEqn{S});
    \quad \embedEqn{C}[_{5}] = \text{FC}_5(\embedEqn{L})
\end{equation*}

\begin{equation}
\label{eq:c_emb}
    \embedEqn{C} = \left[ \embedEqn{C}[_{1}], \embedEqn{C}[_{2}], \embedEqn{C}[_{3}], \embedEqn{C}[_{4}], \embedEqn{C}[_{5}] \right]
\end{equation}

The attention module learns to combine these context vectors into a single vector by learning the weights of each of the contexts.
An attention vector $\mathbf{u} \in \mathbb{R}^{n}$ is learned during the training and is used to compute the attention weights. 
Attention weight ($\alpha_{i}$) for each of the contexts ($\embedEqn{C}[_{i}]), 1 \leq i \leq 5$ is computed as the dot-product of the context and attention vector, followed by a softmax function over all the attention weights.

\begin{equation}
\alpha_i = \frac{\exp(\embedEqn{C}[_{i}]^T \cdot \mathbf{u})}{\sum_{j=1}^{5} \exp(\embedEqn{C}[_{j}]^T \cdot \mathbf{u})}
\end{equation}
This ensures the attention weights $\alpha_{i} \in (0, 1)$ and $\sum_{i=1}^{5} \alpha_{i} =1$. 
Finally, the attention weights are used to aggregate the individual context vectors into a single vector in $\mathbb{R}^{n}$ according to Equation~\ref{eq:FPrime}.

\begin{equation}
\label{eq:FPrime}
\embedFFinal[^\prime] = \sum_{i=1}^{5} \alpha_i \embedEqn{C}[_i]
\end{equation}  

This attention mechanism helps the network to automatically identify the relative importance of entities of the function.
In this sense, notice that we differ from previous work, which relied on a fixed heuristic~\cite{VenkataKeerthy-2020-IR2Vec} to weigh the relative importance of the different entities that make up a function.
The aggregated function vector ($\embedFFinal[^\prime]$) is then passed through another set of fully connected layers to obtain the final function embedding $\embedFFinal$.
This final embedding vector is a compact and more informative representation that captures the similarity metric of the binary function.
We call this embedding the \vexirtovec{} vector that represents the target function, as Definition~\ref{def:vexir2vec} emphasizes.

\begin{definition}[\vexirtovec]
\label{def:vexir2vec}
If $\embedFInit$ is the initial function representation derived from the vocabulary $\mathcal{V}$ (Algorithm~\ref{algorithm:vec_gen}) on a set of peepholes $\mathcal{P}$
(Algorithm~\ref{algorithm:random-walk}), then the embedding obtained from the \vexnet, $\embedFFinal$ is the final \vexirtovec representation of the function $F$.
\end{definition}

\subsection{Binary Similarity Tasks}

Upon training the \vexnet model, the final function embeddings \embed{F^\prime} can solve binary similarity tasks like diffing and searching.
Definition~\ref{def:binsim-tasks} formalizes these tasks.

\begin{definition}[Binary Similarity Tasks]
\label{def:binsim-tasks}
A binary similarity task is a problem involving two collections of binary programs produced from the same or different source codes, albeit using different compilation settings.
In this paper, we recognize two types of binary similarity tasks:

\begin{description}
\item [Diffing:] $\mathcal{M}_{\mathit{diff}}(A, B)$, where $A = \{a_1, a_2, \ldots, a_m\}$ and $B = \{b_1, b_2, \ldots, b_n\}$ are the two sets of \vexirtovec vectors for the functions in binaries $A$ and $B$, respectively.
$\mathcal{M}_{\mathit{diff}}$ produces a list $L_{\mathit{diff}}$ of $\mbox{min}(m, n)$ pairs $(a_i, b_j)$.
Each $a_i, 1 \leq i \leq m$ and each $b_j, 1 \leq j \leq n$ appear only once in $L_{\mathit{diff}}$.
The goal of $\mathcal{M}_{\mathit{diff}}$ is to match every function in $A$ to its equivalent in $B$.

\item [Searching:] $\mathcal{M}_{\mathit{search}}(A, b)$, where $A = \{a_1, a_2, \ldots, a_m\}$  is a collection of $m$ \vexirtovec vectors obtained from different set of binaries, and $b$ is a \vexirtovec vector.
$\mathcal{M}_{\mathit{search}}$ produces one vector $a_i, 1 \leq i \leq m$.
The goal of $\mathcal{M}_{\mathit{search}}$ is to find, within $A$, a vector $a_i$ representing a function that solves the same task as the function that $b$ represents.
\end{description}

\end{definition}

$\mathcal{M}_{\mathit{diff}}(A, B)$ assumes that each $a_i, 1 \leq i \leq m$, comes from a function extracted from a binary program $P_A$, being derived via the combination of Algorithms~\ref{algorithm:random-walk} and \ref{algorithm:vec_gen}, following Definition~\ref{def:vexir2vec}.
Similarly, each $b_i, 1 \leq i \leq n$, comes from a function extracted from a binary program $P_B$.
We assume that $P_A$ and $P_B$ are binary programs derived from the same source code. 
Whereas, $\mathcal{M}_{\mathit{search}}(A, b)$ assumes that each $a_i, 1 \leq i \leq m$ comes from the functions extracted from a collection of different binary programs. 
And, the query function $b$ is searched to retrieve similar functions among the pool of functions in $A$.

\begin{example}[Diffing]
    Let $P_s$ be a source program, and let $P_{ARM}$ and $P_{x86}$ be binary versions of $P_s$, compiled for ARM and x86.
    If we let $A$ be the set of all the \vexirtovec vectors extracted from $P_{ARM}$, and $B$ be the set of all the \vexirtovec vectors extracted from $P_{x86}$, then  $\mathcal{M}_{\mathit{diff}}(A, B)$ tries to match up the binary functions in $P_{ARM}$ and $P_{x86}$.
\end{example}

\begin{example}[Searching]
    Let $P_s$ be a source program, and let $P_{O0}$, $P_{O1}$, $P_{O2}$ and $P_{O3}$ be binary versions of $P_s$, compiled with different optimization levels. 
    If we let $A$ be the collection of \vexirtovec vectors extracted from a corpus of binaries including all functions from $P_{O1}$, $P_{O2}$ and $P_{O3}$, as well as other from unrelated programs $P'$.
    And, let $b$ be \vexirtovec representation of some function in $P_{O0}$, then $\mathcal{M}_{\mathit{search}}(A, b)$ searches for functions performing same task as $b$ in $A$.
\end{example}

%% file: 5.experiments.tex
\section{Experimental Setup}
\label{sec:experimental-setup}

This section describes the experimental setup used in the evaluation discussed in Sections~\ref{sec:performance-evaluation} and \ref{sec:ablation}.

\subsection{Dataset}
\label{sub:dataset}

The dataset used in this paper comes from the following projects: Findutils~\cite{findutils}, Diffutils~\cite{diffutils}, Coreutils~\cite{coreutils}, cURL~\cite{curl}, Lua~\cite{lua}, PuTTY~\cite{putty}, and Gzip~\cite{gzip}.
Binaries are generated by varying the compilation configuration, which includes architecture (x86 or ARM), compiler (Clang or GCC), compiler version (six versions per compiler), and optimization level (\texttt{-O0}, \texttt{-O1}, \texttt{-O2}, and \texttt{-O3}).
Thus, in total, we use $2 \times 2 \times 6 \times 4 = 96$ compiler configurations.
Experiments are performed on stripped binaries, which do not contain debug and symbol information. Stripping is done using the GNU \texttt{strip} utility. However, as we explain in Sections~\ref{sss:gt_diffing} and~\ref{sss:gt_searching}, we use unstripped binaries with debug information to aid in obtaining ground truth and identifying the same function across binaries. The size of individual unstripped binary files ranges from $9.2$KB to $7.6$MB.

We use $70\%$ of the binaries from Findutils, Diffutils, and Coreutils for training and the remaining $30\%$ for testing. All binaries from the other projects are used exclusively for testing and studying the generalizability of our approach. The training set includes binaries compiled using $6$ different versions of Clang (V$6.0$, $8.0.1$, $9.0.1$, $10.0.1$, $11.1.0$, and $12.0.1$) and $6$ different versions of GCC (V$6.4.0$, $7.5.0$, $8.4.0$, $9.4.0$, $10.5.0$, and $11.4.0$). In total, the training set contains $10K$ binaries with $1.5M$ functions.
The test set consists of about $5.5K$ binaries with $1.2M$ functions.
These binaries are produced with the same compiler configurations, except that to reduce the time required for experiments, we use only three versions of each compiler: Clang V$6.0$, $8.0.1$, and $12.0.1$, and GCC V$6.0$, $8.0$, and $10.0$. These test-set functions are used in the experiments described in Section~\ref{subsec:exp-diffing} and Section~\ref{subsec:exp-searching}. Appendix~\ref{appendix:dataset-desc} summarizes the total number of projects, files, and functions, along with their sizes.

We use \angr (V9.2.96) to disassemble the binaries and lift them to obtain the \vexir representation. Embeddings are constructed from the disassembled functions by creating peepholes following Algorithm~\ref{algorithm:random-walk}, setting the maximum length of peepholes ($k$) to $72$ and the minimum number of visits per basic block ($c$) to $2$. The resulting peepholes are normalized as described in Section~\ref{sub:normalization}, and embeddings are obtained using the approaches described in Section~\ref{sec:phaseI-processing} and Section~\ref{sec:phaseII-embeddings}.

\subsection{Training}
\label{sub:training}

Training is the process of building a function that maps sets of peepholes $\mathcal{P}$
to vectors $\embedFFinal$, as explained in Sections~\ref{sec:phaseII-embeddings} and~\ref{sec:phaseIII-model}.

\subsubsection{Pre-Training to obtain $\mathcal{V}_{lookup}$}
\label{sss:train_vocabulary}

Section~\ref{sub:lookup} defines the vocabulary function $\mathcal{V}_{lookup}$, which maps IR entities to 128-dimensional vectors. Each dimension is a 64-bit floating-point number. In this paper, $\mathcal{V}_{lookup}$ was learned from triplets extracted from two benchmark suites: SPEC CPU2006 and SPEC CPU2017~\cite{spec17-Bucek:2018:SCN:3185768.3185771}. 
These benchmarks were compiled using the Clang V6.0, V8.0.1, and V12.0.1, and GCC V6.0, V8.0, and V10.0 with different optimization levels (\texttt{O[0-3]}) targeting x86 and ARM.

We used \angr to extract triplets and an open-source implementation of TransE~\cite{han2018openke} to train the model. The dataset contained approximately $72M$ triplets, from which we identified $145$ unique entities and $10$ unique relations. The model converged after about 900 epochs using the Adam optimizer with a learning rate of $0.002$. The margin $\gamma$ was set to $3$, and the batch size was $256$. These hyperparameters were chosen based on tuning the validation loss.

\subsubsection{Fine-Tuning - Training \vexnet to obtain $\embedFFinal$}
\label{sss:train_final_emb}

As explained in Section~\ref{sub:dataset}, to train the \vexnet model (see Section~\ref{sec:phaseIII-model}), we use a dataset formed by Findutils, Diffutils, and Coreutils.
The inputs for the \vexnet model are the initial function embedding $\embedFInit = \llangle \embedEqn{O} \in \mathbb{R}^{1 \times 128}, \embedEqn{T} \in \mathbb{R}^{1 \times 128}, \embedEqn{A}  \in \mathbb{R}^{1 \times 128} \rrangle$ obtained from the vocabulary $\mathcal{V}_{lookup}$, along with the embedding of strings ($\embedEqn{S} \in \mathbb{R}^{1 \times 100}$) and external library calls ($\embedEqn{L} \in \mathbb{R}^{1 \times 100}$) obtained using the fastText model.

A single fully connected layer with SiLU activation~\cite{Hendrycks2016Silu} processes $\embedFInit, \embedEqn{S}, \embedEqn{L}$ (Eq.~\ref{eq:c_emb}).
Before passing to the model, the inputs are $L_2$ normalized\footnote{The $L_2$-norm is the vector norm, e.g., the square root of the sum of the squares of all the dimensions of the vector.}.
Internally, the fully connected layer converts the inputs into context vectors of $180$ dimensions each. These context vectors are attended to by the attention layer to obtain $\embedEqn{F}[^\prime]$ in $180$ dimensions (Eq.~\ref{eq:FPrime}).
$\embedEqn{F}[^\prime]$ is then processed by the final fully connected layer to obtain $\embedFFinal$ in $128$ dimensions. 

Deep metric learning models can suffer from \textit{collapsing}, where the model learns to embed all or most of the inputs very close to each other at a negligible distance, rendering them indistinguishable~\cite{tian2024collapseaware, xuan2020hard, schroff2015facenet}. To overcome this phenomenon, we use batch easy hard mining~\cite{xuan2020EasyHardMining} to select positive and negative samples for training and employ dropout ($0.02$) and batch normalization~\cite{pmlr-v37-ioffe15-batchnorm} as regularizers. 
We train the model with Adam optimizer with default parameters.

Other hyperparameters like batch size, learning rate, gamma (for Adam), and temperature (for NT-Xent loss) are tuned using RayTune~\cite{Liaw2018RayTune} to find the best-performing values.
For hyperparameter tuning, we use the ASHA scheduler~\cite{li2018asha} with Optuna search~\cite{akiba2019Optuna} with default parameters to select the best model based on the Mean Average Precision (MAP) score on the validation set.
We use an Ubuntu server with Intel Gold $6142$ processors ($32$ threads) and a V100 GPU ($32$ GB) for hyperparameter tuning. Each trial is tuned until convergence, allowing a maximum of $200$ epochs. This process results in the best model with a batch size of $4064$, a gamma of $0.817$, and a temperature of $0.05$. 
Training a single trial is lightweight.
For instance, training takes about $5$--$8$ seconds per epoch on an Ubuntu workstation with Xeon(R) W-1390P processor with $16$ threads, $64$ GB RAM, and a $10$ GB variant of an RTX $3080$ GPU.
In contrast, training SAFE and BinFinder on the same machine takes $\approx 1.5$ hours and $\approx 7.5$ hours per epoch, respectively.

\subsection{Baselines}

We compare \vexirtovec{} with five different tools: BinDiff~\cite{bindiff}, DeepBinDiff~\cite{duan2020deepbindiff}, SAFE~\cite{massarelli2019safe}, Histograms of Opcodes~\cite{Damasio23}, and BinFinder~\cite{Qasem2023Binfinder-AsiaCCS}.

\paragraph{BinDiff}
BinDiff extracts the control flow graph (CFG) information of functions in a binary. It compares the number of blocks, edges, and calls between functions in the source and target binaries to obtain an initial match. BinDiff identifies similar function pairs and provides a similarity score and confidence value. It uses disassemblers like IDA Pro~\cite{IDAPro} and Ghidra~\cite{Ghidra} to extract relevant information and convert it to the BinExport format. For our experiments, we use BinDiff (V7) to generate the BinExport files using Ghidra (V9.2.3) with the BinExport plugin (V12).

\paragraph{DeepBinDiff}
DeepBinDiff (DBD) uses an unsupervised learning technique to obtain embeddings from the control flow graph (CFG) of a function. A greedy matching algorithm is then used to provide block-level diffing between two binaries. We use the code provided by the authors of DeepBinDiff for our experiments. To perform function-level diffing, we follow the approach used by Codee~\cite{jia2021codee}, where two functions are considered similar if at least one pair of blocks between them is predicted to be similar.

\paragraph{SAFE}
SAFE uses Radare~\cite{Radare} to disassemble binary files and trains a word2vec model for instruction-to-embedding mapping. This mapping is used to generate an embedding for each function. SAFE does not process functions if Radare returns instructions with invalid opcodes. We skip such instructions and allow SAFE to construct the function embedding from the remaining instruction sequence.
We use the code and datasets provided by the authors of SAFE and the pre-trained model for the x86 experiments. However, the public distribution of SAFE does not include a pre-trained model for the ARM architecture. Therefore, we train the model using the cross-architecture dataset from SAFE to create a model for all our experiments involving ARM binaries.

\paragraph{Histograms of Opcodes}
In their recent study, \citet{Damasio23} demonstrated that representing programs using histograms of opcodes is as effective as other source code embedding techniques. They showed that using opcode frequencies of LLVM IR~\cite{LLVM-LangRef} as input features can effectively identify similar codes, even when compiler optimizations and obfuscation transformations are applied.
We implemented this technique on binaries using \angr and consider it as a baseline for comparison.
We create a $144$-dimensional feature vector, with each dimension representing the frequency of a specific opcode in \vexir of the function. We then train the \vexnet model using these histogram representations on our training dataset. The model's hyperparameters were fine-tuned based on validation loss, similar to the process described in Section~\ref{sss:train_final_emb}. 
The best model obtained through this process is then used for evaluation. In our experiments, we refer to this technique as OPC.

\paragraph{BinFinder}
BinFinder uses a similar setup as SAFE. We implemented the necessary feature extraction using \angr as described in the paper to create a functional setup. We trained the BinFinder model on our training set using the provided training scripts and then evaluated the model on our test set for comparison.

\section{Performance Evaluation}
\label{sec:performance-evaluation}

This section evaluates the performance of \vexirtovec, in terms of accuracy and scalability, on the two binary similarity experiments formalized in Definition~\ref{def:binsim-tasks}: function diffing and searching.
This evaluation happens in an adversarial setting, involving binaries produced using different compilation configurations, as explained in Section~\ref{sub:training}.
The goal of this evaluation is to provide answers to the following research questions:

\begin{enumerate}[label=(RQ\arabic*)]
    \item \labelText{}{RQ1} How well does our approach perform on the binary diffing task, withstanding the differences arising out of different architectures, compilers, optimizations, and obfuscations? (Section~\ref{subsec:exp-diffing})
    \item \labelText{}{RQ2} How effective is our approach in searching and retrieving a similar function from a large pool of functions varying in compilation configurations? (Section~\ref{subsec:exp-searching})
    \item \labelText{}{RQ3} How well does the vocabulary capture the semantic information from the input \vexir? (Section~\ref{sec:vocab-eval})
    \item \labelText{}{RQ4} How scalable is our proposed tool in comparison with the other approaches? (Section~\ref{sub:scalability})
\end{enumerate}

\input{5a-experiment-diffing}
\input{5b-experiment-searching}
\input{5c-vocabulary-eval}
\input{5d-scalability}

%% file: 5a-experiment-diffing.tex
\subsection{RQ1: Binary Diffing}
\label{subsec:exp-diffing}

We evaluate our approach on the binary diffing task at the function level, as described in Section~\ref{def:binsim-tasks} and compare it with different baselines to answer \hyperref[RQ1]{RQ1}.

\subsubsection{Ground truth}
\label{sss:gt_diffing}

The ground truth for the binary diffing experiment consists of function pairs from stripped source and target binaries, corresponding to the same source code function compiled using different configurations. To identify functions from the same source code, we first match functions from stripped binaries with their unstripped equivalents using their function addresses. Then, we obtain the debugging and symbol information from the unstripped functions. A pair of functions from the stripped source and target binaries with the same source filename and function name are considered identical and added to the ground truth.

\begin{example}
To create a ground truth function pair between the \texttt{find} binaries compiled with x86-Clang12-O0 and x86-GCC10-O2, we first match the function addresses in the stripped and unstripped binaries to associate the source filename and function name from the unstripped function with the stripped function. This process is performed separately for each of the \texttt{find} binaries compiled with x86-Clang12-O0 and x86-GCC10-O2. Next, we pair functions from the stripped x86-Clang12-O0 binary with those from the x86-GCC10-O2 binary by checking for matching source filenames and function names. When a match is found, we identify two functions derived from the same source code and add this pair to the ground truth.
\end{example}

\subsubsection{Metrics}
\label{sss:metrics_diffing}

We use precision, recall, and F1 score as the primary evaluation metrics for the binary diffing task.
\begin{itemize}
    \item True Positive ($TP$): number of pairs of functions correctly predicted in the ground truth.
    \item False Positive ($FP$): number of pairs of functions predicted as similar but not present in the ground truth.
    \item False Negative ($FN$): number of pairs of functions in the ground truth that are not predicted to be similar.
    \item Precision ($\frac{TP}{TP+FP}$): the proportion of true positives among all positive results predicted by the model.
    \item Recall ($\frac{TP}{TP+FN}$):  the proportion of true positives among all actual positive cases in the dataset.
    \item F1 score: the harmonic mean of precision and recall, computed as $\frac{2*Precision*Recall}{Precision+Recall}$.
\end{itemize}

DeepBinDiff and BinDiff directly generate matching function pairs for two binaries with a similarity score.
For other approaches, we use the KDTree implementation from scikit-learn~\cite{scikit-learn} to compute the nearest neighbors for a source function among the set of target functions. 
KDTrees are binary search trees where data in each node is a K-Dimensional point. KDTrees provide an efficient means\footnote{They have a linear-space complexity, an average logarithmic-time complexity, and a worst-case linear-time complexity.} of computation of nearest neighbors~\cite{bentley1975kdtrees}. 
We consider $10$ neighbors in our experiments. 
A prediction is counted as a true positive ($TP$) if the target function is among the computed neighbors; otherwise, it is marked as a false positive ($FP$). 
Additionally, we plot the Cumulative Distribution Function (CDF) of F1-Score curves for each experiment.

Some functions identified in the ground truth are not present in the baselines.
Such function pairs are considered false negatives according to our definitions. Experiments use a timeout of 7,200 seconds, beyond which we terminate the process for the pair of binaries under comparison. We skip evaluating DeepBinDiff on binaries from Lua, cURL, and PuTTY, and in experiments involving obfuscated binaries, as generating the embeddings and matching exceeded the 7,200-second timeout for almost 90\% of the binary pairs. Additionally, we do not evaluate DeepBinDiff on cross-architecture experiments since the approach is architecture-specific and supports only x86.

\subsubsection{Cross-Optimization Binary Diffing}
\label{sec:cross-opt-diffing}

\input{tables/Exp1-tables}

We perform cross-optimization level binary diffing involving binaries compiled with Clang12 and GCC8 targeting x86 and ARM architectures.
 Table~\ref{tab:crossopt-O0-O2-O3}\footnote{Henceforth, we add to the caption of each table the expected score of a random guesser (the null hypothesis). This score, being so low, demonstrates that the different binary similarity tools are able to learn non-trivial information.} shows the precision and recall values corresponding to the diffing between O0-O3 and O0-O2 optimization levels. Additionally, we studied the performance of diffing between O1-O3 optimization levels (results are shown in Table~\ref{tab:crossopt-O1-O3} of the Appendix). Our approach consistently achieves higher precision and recall values across all configurations compared to the baselines. It attains the best average F1-Score across all three experiments, at 0.65, while the nearest baselines, BinFinder and Histograms of Opcodes (OPC), achieve F1-Scores of 0.47 and 0.41, respectively. On average, SAFE, DeepBinDiff, and BinDiff achieve F1 scores of 0.28, 0.27, and 0.26, respectively.

\subsubsection{Cross-Compiler Binary Diffing}
\label{sec:cross-compiler-diffing}

To perform cross-compiler diffing, we consider binaries generated using Clang12 and GCC8 targeting x86 and ARM. Diffing is performed on binaries that target the same architecture but that were produced with different compilers or different versions of the same compiler. Table~\ref{tab:crosscompiler} shows the precision and recall for our approach and the baselines. Our approach achieves the highest precision, recall, and F1 scores, outperforming the baselines in all configurations.

\input{tables/Exp2a-tables}

We also carry out diffing by varying the compiler version. Table~\ref{tab:crossversion} shows the results of the comparisons between Clang6 and Clang12, Clang8 and Clang12, GCC6 and GCC10, GCC8 and GCC10. In our dataset, we observe that the cross-compiler version diffing is easier than other diffing experiments. Our approach results in the highest average F1 score of $0.93$, while BinFinder, OPC, and SAFE result in $0.87$, $0.8$, and $0.53$, respectively.

\input{tables/Exp2b-tables}

\subsubsection{Cross-Architecture Binary Diffing}
\label{sec:cross-arch-diffing}
\input{tables/Exp3-tables}

This experiment uses binaries generated by Clang (V12.0.1) and GCC (V10.0) at \texttt{-O0} and \texttt{-O3} optimization levels for x86 and ARM. We perform diffing between binaries generated by the same compiler and optimization level but for different architectures. Table~\ref{tab:crossarchitecture} shows the results. \vexirtovec achieves the highest precision and recall values in most configurations.
Notice that in this experiment, \vexirtovec is not consistently better than all the other baselines. For instance, in the clang12 \texttt{-O0} setup, \vexirtovec trails BinDiff on the Lua binary, although it outperforms it in all the other six binaries on the same setting.
On average, \vexirtovec achieves the highest F1 score of $0.82$, surpassing the closest baseline, BinFinder, which obtains an average F1 score of $0.67$.

\subsubsection{Mixture-of-all Diffing}
\label{sec:all-diffing}

\input{tables/Exp4a-tables}
\input{tables/Exp4b-tables}

As an extreme case of adversarial diffing, we conducted experiments using a diverse set of binaries compiled with different compilers and varying optimization levels for both x86 and ARM.
This setup poses a significant challenge to the binary diffing tools due to its mix of compilation configurations.
Tables~\ref{tab:cross-comp-arch-opt} and~\ref{tab:cross-comp-opt} present the compilation configurations, precision, and recall results for this experiment.
Cross-architecture experiments are excluded for DeepBinDiff for the reasons mentioned in Section~\ref{sss:metrics_diffing}.
\vexirtovec outperforms all baselines across these configurations in terms of both precision and recall.
On average, \vexirtovec achieves approximately a $45\%$ higher average F1-score compared to the best-performing baseline, BinFinder.

Figure~\ref{fig:f1-curves} shows Cumulative Distributive Function plots for F1-scores.
Ideal curves have a maximum increase in the fraction of data close to 1.0.
Notice that \vexirtovec{} achieves superior results than the baselines in all the experiments that Figure~\ref{fig:f1-curves} summarizes.

\begin{figure*}[ht]
    \centering
    \includegraphics[width=\textwidth,  height=0.30\textheight] {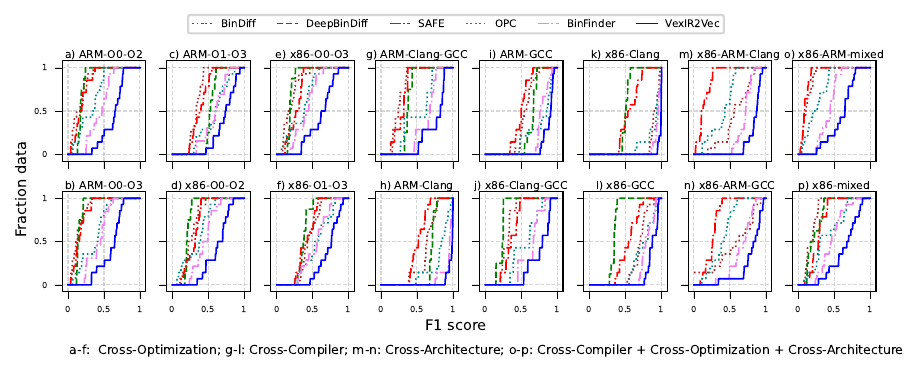}
    \caption{The cumulative distribution function (CDF) of F1-Scores in different adversarial settings.}
    \Description{F1-Score CDF plots}
    \label{fig:f1-curves}
    \vspace{-0.75\baselineskip}
\end{figure*}

\subsubsection{Diffing on Obfuscated Binaries}
\label{sec:obf-diffing}

\input{tables/Obf-tables}

Code obfuscation is a deliberate technique that alters code to make it difficult for humans to understand or interpret its underlying functionality.
This section evaluates the ability of the different diffing tools to resist the effects of code obfuscation on the target binaries.
To this end, we use O-LLVM to modify binaries before passing them to the diffing software.
O-LLVM is a widely used obfuscation tool integrated with the Clang and LLVM toolchain (specifically, it runs on Clang version 4.0)~\cite{ollvm-JunodRWM}.
O-LLVM offers three primary obfuscation schemes:

\begin{itemize}

\item \textbf{Bogus Control Flow (BCF)}: This scheme adds unnecessary branches guarded by conditions (opaque predicates) that always evaluate to the same outcome. These branches cannot be eliminated by compiler optimizations, making the code's control flow harder to follow.

\item \textbf{Control-Flow Flattening (FLA)}: This scheme flattens the Control Flow Graph by modifying the conditions within the graph and inserting additional code blocks that do not affect the program's functionality.
In practice, the program becomes a large switch within a loop~\cite{laszlo2009obfuscating}.

\item \textbf{Instruction Substitution (SUB)}: This scheme replaces a single instruction with a sequence of multiple instructions that achieve the same result. These replacements increase the complexity of the code without changing its behavior.
\end{itemize}

We generated obfuscated binaries using O-LLVM (Clang V4.0) on x86 at the \texttt{-O3} optimization level.  In addition to applying the three individual obfuscation passes mentioned above, we also created a variant by applying all three passes together (denoted as ``ALL'').  These obfuscated binaries were then compared to their non-obfuscated counterparts generated using the GCC compiler (version 10.0) with the optimization level set to \texttt{-O0}.  The experimental setup was identical to that described in Section~\ref{subsec:exp-diffing}.
Table~\ref{tab:obfuscation} shows the results of our approach and the baselines\footnote{We do not consider the binaries from Lua in this experiment as the \angr (V9.2.96) was not able to retrieve CFGs of the functions after obfuscation using Clang 4.}.
As shown in that table, \vexirtovec achieves the highest precision and recall values in almost all cases.  \vexirtovec obtains the highest average F1 score of $0.65$, compared to BinFinder's $0.41$. SAFE outperforms OPC, with an average F1 score of $0.29$ versus $0.22$. In general, it can be seen that the baselines do not perform well under obfuscation.

%% file: tables/Exp1-tables.tex
\begin{table}[ht]
\centering
\caption{Precision and Recall Scores for Cross-Optimization Binary Diffing. The first number of each pair corresponds to the O0 Vs. O3 setting. The second number shows the O0 Vs. O2 setting. Missing pairs occur due to timeouts.
The null hypothesis' expected score is 0.005.}
\label{tab:crossopt-O0-O2-O3}
\resizebox{\textwidth}{!}{
\begin{tabular}{l|cccccc|cccccc}
\toprule
                                              & \multicolumn{6}{|c}{\textbf{\textbf{Precision}}} & \multicolumn{6}{|c}{\textbf{\textbf{Recall}}} \\
                                              &   \textbf{\textbf{BinDiff}} & \textbf{\textbf{DBD}} & \textbf{\textbf{SAFE}} & \textbf{\textbf{OPC}} & \textbf{\textbf{BinFinder}} &    \textsc{\textbf{VexIR2Vec}} & \textbf{\textbf{BinDiff}} & \textbf{\textbf{DBD}} & \textbf{\textbf{SAFE}} & \textbf{\textbf{OPC}} & \textbf{\textbf{BinFinder}} &    \textsc{\textbf{VexIR2Vec}} \\
\midrule
& \multicolumn{12}{c}{\textbf{\underline{ARM - Clang12}}} \\
\textbf{\textbf{Coreutils}} &                 0.08 / 0.08 &           0.17 / 0.18 &            0.11 / 0.11 &                0.36 / 0.34 &                 0.35 / 0.35 &  \textbf{0.57} / \textbf{0.58} &               0.05 / 0.05 &           0.09 / 0.10 &            0.09 / 0.09 &                0.47 / 0.45 &                 0.46 / 0.45 &  \textbf{0.75} / \textbf{0.76} \\
                       \textbf{\textbf{Diffutils}} &                 0.37 / 0.41 &           0.24 / 0.27 &            0.17 / 0.19 &                0.44 / 0.44 &                 0.49 / 0.48 &  \textbf{0.68} / \textbf{0.67} &               0.27 / 0.30 &           0.19 / 0.22 &            0.13 / 0.14 &                0.57 / 0.58 &                 0.65 / 0.63 &  \textbf{0.89} / \textbf{0.89} \\
                       \textbf{\textbf{Findutils}} &                 0.11 / 0.34 &           0.31 / 0.32 &            0.12 / 0.13 &                0.40 / 0.40 &                 0.47 / 0.52 &  \textbf{0.68} / \textbf{0.69} &               0.06 / 0.22 &           0.15 / 0.17 &            0.10 / 0.11 &                0.48 / 0.48 &                 0.56 / 0.62 &  \textbf{0.82} / \textbf{0.82} \\
                       \textbf{\textbf{Gzip}} &                 0.07 / 0.05 &           0.15 / 0.17 &            0.16 / 0.20 &                0.32 / 0.33 &                 0.41 / 0.43 &   \textbf{0.60} / \textbf{0.61} &               0.05 / 0.04 &           0.16 / 0.18 &            0.11 / 0.13 &                0.42 / 0.43 &                 0.54 / 0.56 &   \textbf{0.78} / \textbf{0.80} \\
                       \textbf{\textbf{Lua}} &                 0.20 / 0.17 &                       &            0.04 / 0.04 &                0.13 / 0.14 &                 0.19 / 0.21 &  \textbf{0.41} / \textbf{0.41} &               0.19 / 0.16 &                       &            0.03 / 0.02 &                0.21 / 0.21 &                 0.31 / 0.33 &  \textbf{0.65} / \textbf{0.65} \\
                       \textbf{\textbf{PuTTY}} &                 0.07 / 0.07 &                       &            0.04 / 0.04 &                0.06 / 0.06 &                 0.21 / 0.23 &  \textbf{0.28} / \textbf{0.29} &               0.04 / 0.05 &                       &            0.04 / 0.04 &                0.08 / 0.07 &                 0.29 / 0.30 &  \textbf{0.39} / \textbf{0.39} \\
                       \textbf{\textbf{cURL}} &                 0.14 / 0.07 &                       &            0.12 / 0.14 &                0.50 / 0.46 &                 0.58 / 0.58 &  \textbf{0.67} / \textbf{0.67} &               0.12 / 0.06 &                       &            0.08 / 0.10 &                0.57 / 0.53 &                 0.67 / 0.67 &  \textbf{0.77} / \textbf{0.77} \\
\cline{1-13}
& \multicolumn{12}{c}{\textbf{\underline{ARM - GCC8}}} \\
\textbf{\textbf{Coreutils}} &                 0.13 / 0.14 &           0.16 / 0.18 &            0.18 / 0.20 &                0.36 / 0.29 &                 0.34 / 0.40 &  \textbf{0.51} / \textbf{0.55} &               0.09 / 0.10 &           0.11 / 0.12 &            0.13 / 0.15 &                0.49 / 0.41 &                 0.47 / 0.56 &   \textbf{0.70} / \textbf{0.76} \\
                                               \textbf{\textbf{Diffutils}} &                 0.29 / 0.25 &           0.19 / 0.22 &            0.32 / 0.34 &                0.36 / 0.40 &                 0.43 / 0.41 &   \textbf{0.60} / \textbf{0.65} &               0.24 / 0.18 &           0.09 / 0.14 &            0.21 / 0.23 &                0.52 / 0.55 &                 0.62 / 0.63 &   \textbf{0.86} / \textbf{0.90} \\
                                               \textbf{\textbf{Findutils}} &                 0.13 / 0.10 &           0.18 / 0.28 &            0.22 / 0.26 &                0.41 / 0.33 &                 0.42 / 0.50 &  \textbf{0.59} / \textbf{0.65} &               0.09 / 0.07 &           0.12 / 0.12 &            0.18 / 0.20 &                0.52 / 0.42 &                 0.55 / 0.64 &  \textbf{0.77} / \textbf{0.83} \\
                                               \textbf{\textbf{Gzip}} &                 0.23 / 0.23 &           0.14 / 0.21 &            0.33 / 0.37 &                0.17 / 0.18 &                 0.38 / 0.48 &   \textbf{0.50} / \textbf{0.55} &               0.21 / 0.22 &           0.13 / 0.16 &            0.21 / 0.25 &                0.25 / 0.24 &                 0.55 / 0.64 &  \textbf{0.73} / \textbf{0.72} \\
                                               \textbf{\textbf{Lua}} &                 0.13 / 0.23 &                       &            0.10 / 0.16 &                0.15 / 0.13 &                 0.20 / 0.30 &   \textbf{0.32} / \textbf{0.40} &               0.13 / 0.23 &                       &            0.07 / 0.12 &                0.25 / 0.19 &                 0.33 / 0.44 &  \textbf{0.52} / \textbf{0.58} \\
                                               \textbf{\textbf{PuTTY}} &                 0.05 / 0.10 &                       &            0.08 / 0.10 &                0.09 / 0.06 &                    0.19 / 0.27          &  \textbf{0.28} / \textbf{0.36} &               0.03 / 0.07 &                     &            0.08 / 0.10 &                0.12 / 0.08 &                 0.23 / 0.36            &  \textbf{0.39} / \textbf{0.48} \\
                                               \textbf{\textbf{cURL}} &                 0.19 / 0.23 &                       &            0.21 / 0.25 &                0.53 / 0.12 &                 0.14 / 0.15 &  \textbf{0.57} / \textbf{0.57} &               0.15 / 0.19 &                       &            0.16 / 0.19 &                0.69 / 0.58 &                 0.58 / 0.61 &  \textbf{0.76} / \textbf{0.75} \\
\cline{1-13}
& \multicolumn{12}{c}{\textbf{\underline{x86 - Clang12}}} \\
\textbf{\textbf{Coreutils}} &                 0.12 / 0.16 &           0.12 / 0.12 &            0.24 / 0.24 &                0.33 / 0.34 &                 0.34 / 0.34 &  \textbf{0.53} / \textbf{0.53} &               0.09 / 0.12 &           0.36 / 0.34 &            0.20 / 0.19 &                0.44 / 0.46 &                 0.44 / 0.44 &   \textbf{0.71} / \textbf{0.70} \\
                                               \textbf{\textbf{Diffutils}} &                 0.48 / 0.48 &           0.17 / 0.17 &            0.31 / 0.33 &                0.45 / 0.46 &                 0.43 / 0.45 &  \textbf{0.68} / \textbf{0.68} &               0.44 / 0.45 &           0.55 / 0.60 &            0.22 / 0.23 &                0.60 / 0.61 &                 0.57 / 0.59 &   \textbf{0.91} / \textbf{0.90} \\
                                               \textbf{\textbf{Findutils}} &                 0.27 / 0.32 &           0.12 / 0.17 &            0.32 / 0.32 &                0.38 / 0.37 &                 0.50 / 0.45 &  \textbf{0.66} / \textbf{0.67} &               0.21 / 0.28 &           0.47 / 0.53 &            0.26 / 0.26 &                0.47 / 0.46 &                 0.61 / 0.55 &  \textbf{0.82} / \textbf{0.82} \\
                                               \textbf{\textbf{Gzip}} &                 0.22 / 0.24 &           0.11 / 0.11 &            0.35 / 0.35 &                0.38 / 0.37 &                 0.42 / 0.36 &   \textbf{0.60} / \textbf{0.62} &               0.18 / 0.20 &           0.66 / 0.76 &            0.25 / 0.25 &                0.49 / 0.48 &                 0.54 / 0.48 &  \textbf{0.79} / \textbf{0.81} \\
                                               \textbf{\textbf{Lua}} &                 0.26 / 0.23 &                       &            0.17 / 0.17 &                0.11 / 0.10 &                 0.21 / 0.22 &  \textbf{0.39} / \textbf{0.39} &               0.24 / 0.21 &                       &            0.12 / 0.12 &                0.17 / 0.15 &                 0.32 / 0.34 &  \textbf{0.59} / \textbf{0.59} \\
                                               \textbf{\textbf{PuTTY}} &                 0.11 / 0.11 &                       &            0.12 / 0.12 &                0.06 / 0.06 &                 0.20 / 0.21 &    \textbf{0.30} / \textbf{0.30} &               0.10 / 0.10 &                       &            0.12 / 0.12 &                0.08 / 0.07 &                 0.27 / 0.28 &   \textbf{0.40} / \textbf{0.41} \\
                                               \textbf{\textbf{cURL}} &                 0.25 / 0.30 &                       &            0.35 / 0.37 &                0.42 / 0.41 &                 0.51 / 0.57 &  \textbf{0.74} / \textbf{0.73} &               0.20 / 0.25 &                       &            0.25 / 0.27 &                0.48 / 0.47 &                 0.59 / 0.66 &  \textbf{0.84} / \textbf{0.84} \\
\cline{1-13}
& \multicolumn{12}{c}{\textbf{\underline{x86 - GCC8}}} \\
\textbf{\textbf{Coreutils}} &                 0.11 / 0.21 &           0.12 / 0.14 &            0.23 / 0.30 &                0.31 / 0.32 &                 0.32 / 0.36 &   \textbf{0.50} / \textbf{0.51} &               0.08 / 0.19 &           0.29 / 0.31 &            0.18 / 0.23 &                0.44 / 0.45 &                 0.44 / 0.51 &  \textbf{0.69} / \textbf{0.72} \\
                                               \textbf{\textbf{Diffutils}} &                 0.27 / 0.29 &           0.14 / 0.14 &            0.37 / 0.39 &                0.42 / 0.43 &                 0.38 / 0.42 &    \textbf{0.60} / \textbf{0.60} &               0.25 / 0.31 &           0.42 / 0.42 &            0.25 / 0.26 &                0.60 / 0.61 &                 0.55 / 0.60 &  \textbf{0.86} / \textbf{0.86} \\
                                               \textbf{\textbf{Findutils}} &                 0.10 / 0.36 &           0.10 / 0.13 &            0.26 / 0.32 &                0.34 / 0.32 &                 0.40 / 0.48 &   \textbf{0.60} / \textbf{0.62} &               0.08 / 0.32 &           0.31 / 0.34 &            0.21 / 0.25 &                0.44 / 0.41 &                 0.52 / 0.62 &  \textbf{0.78} / \textbf{0.81} \\
                                               \textbf{\textbf{Gzip}} &                 0.30 / 0.35 &           0.11 / 0.14 &            0.34 / 0.38 &                0.18 / 0.22 &                 0.34 / 0.44 &  \textbf{0.52} / \textbf{0.55} &               0.32 / 0.36 &           0.53 / 0.53 &            0.22 / 0.27 &                0.25 / 0.30 &                 0.48 / 0.59 &  \textbf{0.75} / \textbf{0.73} \\
                                               \textbf{\textbf{Lua}} &                 0.14 / 0.35 &                       &            0.14 / 0.19 &                0.15 / 0.12 &                 0.22 / 0.30 &  \textbf{0.31} / \textbf{0.39} &               0.14 / 0.37 &                       &            0.10 / 0.14 &                0.24 / 0.18 &                 0.35 / 0.44 &   \textbf{0.50} / \textbf{0.57} \\
                                               \textbf{\textbf{PuTTY}} &                 0.09 / 0.23 &                       &            0.12 / 0.15 &                0.08 / 0.07 &                 0.16 / 0.25 &  \textbf{0.29} / \textbf{0.36} &               0.08 / 0.21 &                       &            0.12 / 0.15 &                0.11 / 0.10 &                 0.23 / 0.33 &  \textbf{0.41} / \textbf{0.49} \\
                                               \textbf{\textbf{cURL}} &                 0.16 / 0.20 &                       &            0.34 / 0.40 &                0.55 / 0.56 &                 0.58 / 0.63 &  \textbf{0.79} / \textbf{0.79} &               0.15 / 0.19 &                       &            0.24 / 0.29 &                0.64 / 0.64 &                 0.69 / 0.72 &  \textbf{0.93} / \textbf{0.91} \\
\bottomrule
\end{tabular}
}
\end{table}

%% file: tables/Exp2a-tables.tex
\begin{table}[ht]
\centering
\caption{Cross-Compiler Binary Diffing - Clang12 Vs. GCC8. Null hypothesis' expected score: 0.006}
\label{tab:crosscompiler}
\resizebox{0.9\textwidth}{!}{
\begin{tabular}{l|rrrrrr|rrrrrrr}
\toprule
                             & \multicolumn{6}{|c}{\textbf{Precision}} & \multicolumn{6}{|c}{\textbf{Recall}} \\
                             &   \textbf{BinDiff} & \textbf{DBD} & \textbf{SAFE} & \textbf{OPC} & \textbf{BinFinder} & \textbf{\textsc{\textbf{VexIR2Vec}}} & \textbf{BinDiff} & \textbf{DBD} & \textbf{SAFE} & \textbf{OPC} & \textbf{BinFinder} & \textbf{\textsc{\textbf{VexIR2Vec}}} \\
\midrule

& \multicolumn{12}{c}{\textbf{\underline{ARM}}} \\

\textbf{Coreutils} &               0.23 &         0.34 &          0.27 &              0.60 &               0.66 &      \textbf{0.78} &             0.13 &         0.32 &          0.22 &              0.68 &               0.75 &      \textbf{0.88} \\
              \textbf{cURL} &               0.20 &              &          0.32 &              0.65 &               0.69 &      \textbf{0.79} &             0.19 &              &          0.22 &              0.79 &               0.84 &      \textbf{0.96} \\
              \textbf{Diffutils} &               0.28 &         0.37 &          0.37 &              0.64 &               0.59 &      \textbf{0.78} &             0.21 &         0.39 &          0.26 &              0.78 &               0.79 &      \textbf{0.94} \\
              \textbf{Findutils} &               0.41 &         0.53 &          0.32 &              0.60 &               0.69 &      \textbf{0.83} &             0.29 &         0.36 &          0.27 &              0.67 &               0.78 &      \textbf{0.93} \\
              \textbf{Gzip} &               0.37 &         0.27 &          0.37 &              0.29 &               0.61 &      \textbf{0.71} &             0.34 &         0.56 &          0.26 &              0.35 &               0.72 &      \textbf{0.84} \\
              \textbf{Lua} &               0.27 &              &          0.19 &              0.29 &               0.44 &      \textbf{0.52} &             0.26 &              &          0.14 &              0.36 &               0.54 &      \textbf{0.64} \\
              \textbf{PuTTY} &               0.19 &              &          0.13 &              0.26 &               0.44 &      \textbf{0.49} &             0.13 &              &          0.15 &              0.28 &               0.48 &      \textbf{0.55} \\
\cline{1-13}
 
 & \multicolumn{12}{c}{\textbf{\underline{x86}}} \\

\textbf{Coreutils} &               0.25 &         0.16 &          0.32 &              0.58 &               0.63 &      \textbf{0.75} &             0.22 &         0.53 &          0.27 &              0.66 &               0.71 &      \textbf{0.84} \\
              \textbf{cURL} &               0.30 &              &          0.48 &              0.76 &               0.82 &      \textbf{0.87} &             0.22 &              &          0.35 &              0.81 &               0.88 &      \textbf{0.93} \\
              \textbf{Diffutils} &               0.42 &         0.18 &          0.44 &              0.65 &               0.62 &      \textbf{0.81} &             0.43 &         0.32 &          0.31 &              0.80 &               0.76 &      \textbf{0.98} \\
              \textbf{Findutils} &               0.34 &         0.16 &          0.36 &              0.57 &               0.68 &      \textbf{0.82} &             0.30 &         0.59 &          0.30 &              0.65 &               0.76 &      \textbf{0.92} \\
              \textbf{Gzip} &               0.34 &         0.08 &          0.43 &              0.32 &               0.62 &       \textbf{0.7} &             0.33 &         0.72 &          0.32 &              0.37 &               0.70 &      \textbf{0.81} \\
              \textbf{Lua} &               0.34 &              &          0.29 &              0.35 &               0.42 &      \textbf{0.54} &             0.33 &              &          0.23 &              0.43 &               0.51 &      \textbf{0.66} \\
              \textbf{PuTTY} &               0.22 &              &          0.21 &              0.38 &               0.43 &      \textbf{0.51} &             0.20 &              &          0.24 &              0.41 &               0.47 &      \textbf{0.55} \\
\bottomrule
\end{tabular}
}
\end{table}

%% file: tables/Exp2b-tables.tex
\begin{table}[ht]
\centering
\caption{Compiler Cross-Version Binary Diffing. Null hypothesis' expected score: 0.006}
\label{tab:crossversion}
\resizebox{0.9\textwidth}{!}{
\begin{tabular}{l|rrrrrr|rrrrrr}
\toprule
                                      & \multicolumn{6}{|c}{\textbf{Precision}} & \multicolumn{6}{|c}{\textbf{Recall}} \\
                                      &   \textbf{BinDiff} & \textbf{DBD} & \textbf{SAFE} & \textbf{OPC} & \textbf{BinFinder} & \textbf{\textsc{\textbf{VexIR2Vec}}} & \textbf{BinDiff} & \textbf{DBD} & \textbf{SAFE} & \textbf{OPC} & \textbf{BinFinder} & \textbf{\textsc{\textbf{VexIR2Vec}}} \\
\midrule

& \multicolumn{12}{c}{\textbf{\underline{Clang6 Vs. Clang12}}} \\

\textbf{ARM} &               0.72 &         0.59 &          0.51 &              0.80 &               0.93 &      \textbf{0.97} &             0.57 &         0.86 &          0.43 &              0.80 &               0.93 &      \textbf{0.98} \\
                         \textbf{x86} &               0.86 &         0.34 &          0.55 &              0.88 &               0.94 &      \textbf{0.98} &             0.94 &         0.89 &          0.48 &              0.89 &               0.94 &      \textbf{0.99} \\
\cline{1-13}

& \multicolumn{12}{c}{\textbf{\underline{Clang8 Vs. Clang12}}} \\

\textbf{ARM} &               0.72 &         0.62 &          0.53 &              0.84 &               0.94 &      \textbf{0.97} &             0.58 &         0.88 &          0.45 &              0.84 &               0.94 &      \textbf{0.98} \\
                         \textbf{x86} &               0.88 &         0.36 &          0.56 &              0.89 &               0.95 &      \textbf{0.99} &             0.96 &         0.92 &          0.50 &              0.89 &               0.95 &      \textbf{0.99} \\
\cline{1-13}

& \multicolumn{12}{c}{\textbf{\underline{GCC6 Vs. GCC10}}} \\

\textbf{ARM} &               0.58 &         0.52 &          0.49 &              0.72 &               0.75 &      \textbf{0.84} &             0.49 &         0.82 &          0.40 &              0.79 &               0.83 &      \textbf{0.92} \\
                     \textbf{x86} &               0.62 &         0.21 &          0.53 &              0.69 &               0.76 &      \textbf{0.83} &             0.66 &         0.84 &          0.44 &              0.76 &               0.83 &      \textbf{0.91} \\
\cline{1-13}

& \multicolumn{12}{c}{\textbf{\underline{GCC8 Vs. GCC10}}} \\

\textbf{ARM} &               0.66 &         0.56 &          0.51 &              0.75 &               0.75 &      \textbf{0.85} &             0.56 &         0.84 &          0.41 &              0.81 &               0.84 &      \textbf{0.92} \\
                         \textbf{x86} &               0.71 &         0.21 &          0.55 &              0.74 &               0.77 &      \textbf{0.84} &             0.81 &         0.86 &          0.46 &              0.82 &               0.85 &      \textbf{0.93} \\
\bottomrule
\end{tabular}
}
\end{table}

%% file: tables/Exp3-tables.tex
\begin{table}[ht]
    \centering
    \caption{Cross-Architecture Binary Diffing. Null hypothesis' expected score: 0.005}
    \label{tab:crossarchitecture}
    \resizebox{0.9\textwidth}{!}{
    \begin{tabular}{l|ccccc|ccccc}
    \toprule
                                        & \multicolumn{5}{|c}{\textbf{Precision}} & \multicolumn{5}{|c}{\textbf{Recall}} \\
                                        &   \textbf{BinDiff} & \textbf{SAFE} & \textbf{OPC} & \textbf{BinFinder} & \textbf{\textsc{\textbf{VexIR2Vec}}} & \textbf{BinDiff} & \textbf{SAFE} & \textbf{OPC} & \textbf{BinFinder} & \textbf{\textsc{\textbf{VexIR2Vec}}} \\
    \midrule
& \multicolumn{10}{c}{\textbf{\underline{Clang12 - O0}}} \\
    \textbf{Coreutils} &               0.53 &          0.10 &              0.41 &               0.64 &      \textbf{0.79} &             0.43 &          0.09 &              0.45 &               0.67 &      \textbf{0.85} \\
                         \textbf{cURL} &               0.83 &          0.10 &              0.41 &               0.83 &      \textbf{0.94} &             0.72 &          0.08 &              0.41 &               0.83 &      \textbf{0.95} \\
                         \textbf{Diffutils} &               0.71 &          0.14 &              0.41 &               0.70 &      \textbf{0.84} &             0.60 &          0.12 &              0.44 &               0.74 &       \textbf{0.9} \\
                         \textbf{Findutils} &               0.78 &          0.10 &              0.37 &               0.70 &      \textbf{0.85} &             0.66 &          0.09 &              0.39 &               0.73 &       \textbf{0.9} \\
                         \textbf{Gzip} &               0.81 &          0.08 &              0.50 &               0.71 &      \textbf{0.89} &             0.78 &          0.07 &              0.53 &               0.76 &      \textbf{0.94} \\
                         \textbf{Lua} &      \textbf{0.87} &          0.02 &              0.12 &               0.61 &               0.77 &    \textbf{0.82} &          0.02 &              0.12 &               0.61 &               0.77 \\
                         \textbf{PuTTY} &               0.56 &          0.03 &              0.10 &               0.50 &      \textbf{0.65} &             0.55 &          0.03 &              0.12 &               0.55 &      \textbf{0.72} \\
    \cline{1-11}
& \multicolumn{10}{c}{\textbf{\underline{Clang12 - O3}}} \\
    \textbf{Coreutils} &               0.30 &          0.17 &              0.48 &               0.67 &      \textbf{0.84} &             0.14 &          0.15 &              0.52 &               0.73 &      \textbf{0.92} \\
                         \textbf{cURL} &               0.63 &          0.26 &              0.54 &               0.77 &      \textbf{0.82} &             0.51 &          0.19 &              0.55 &               0.78 &      \textbf{0.83} \\
                         \textbf{Diffutils} &               0.64 &          0.22 &              0.55 &               0.76 &      \textbf{0.88} &             0.43 &          0.17 &              0.61 &               0.83 &      \textbf{0.96} \\
                         \textbf{Findutils} &               0.45 &          0.23 &              0.52 &               0.69 &      \textbf{0.86} &             0.21 &          0.20 &              0.56 &               0.74 &      \textbf{0.92} \\
                         \textbf{Gzip} &               0.55 &          0.26 &              0.49 &               0.75 &      \textbf{0.86} &             0.45 &          0.19 &              0.53 &               0.82 &      \textbf{0.94} \\
                         \textbf{Lua} &               0.80 &          0.13 &              0.24 &               0.70 &      \textbf{0.82} &             0.71 &          0.11 &              0.25 &               0.73 &      \textbf{0.85} \\
                         \textbf{PuTTY} &               0.56 &          0.10 &              0.18 &               0.54 &      \textbf{0.67} &             0.41 &          0.12 &              0.19 &               0.57 &      \textbf{0.71} \\
    \cline{1-11}
& \multicolumn{10}{c}{\textbf{\underline{GCC10 - O0}}} \\
    \textbf{Coreutils} &               0.78 &          0.29 &              0.43 &               0.71 &      \textbf{0.89} &             0.65 &          0.26 &              0.46 &               0.76 &      \textbf{0.94} \\
                         \textbf{cURL} &      \textbf{0.87} &          0.36 &              0.10 &               0.18 &               0.21 &             0.78 &          0.28 &              0.47 &               0.82 &      \textbf{0.94} \\
                         \textbf{Diffutils} &               0.88 &          0.39 &              0.41 &               0.79 &      \textbf{0.93} &             0.76 &          0.33 &              0.43 &               0.82 &      \textbf{0.98} \\
                         \textbf{Findutils} &               0.86 &          0.29 &              0.45 &               0.73 &      \textbf{0.95} &             0.71 &          0.27 &              0.47 &               0.76 &      \textbf{0.98} \\
                         \textbf{Gzip} &               0.84 &          0.31 &              0.49 &               0.84 &      \textbf{0.88} &             0.82 &          0.25 &              0.52 &               0.89 &      \textbf{0.93} \\
                         \textbf{Lua} &      \textbf{0.95} &          0.17 &              0.34 &               0.59 &               0.88 &    \textbf{0.91} &          0.16 &              0.34 &               0.59 &               0.88 \\
                         \textbf{PuTTY} &               0.59 &          0.12 &              0.10 &               0.47 &      \textbf{0.72} &             0.57 &          0.13 &              0.12 &               0.52 &      \textbf{0.78} \\
    \cline{1-11}
& \multicolumn{10}{c}{\textbf{\underline{GCC10 - O3}}} \\
    \textbf{Coreutils} &               0.34 &          0.26 &              0.33 &               0.51 &      \textbf{0.61} &             0.24 &          0.20 &              0.44 &               0.68 &      \textbf{0.81} \\
                         \textbf{cURL} &               0.56 &          0.30 &              0.52 &               0.64 &      \textbf{0.75} &             0.47 &          0.21 &              0.59 &               0.73 &      \textbf{0.86} \\
                         \textbf{Diffutils} &               0.48 &          0.25 &              0.44 &               0.71 &      \textbf{0.79} &             0.33 &          0.17 &              0.52 &               0.84 &      \textbf{0.93} \\
                         \textbf{Findutils} &               0.45 &          0.26 &              0.34 &               0.65 &      \textbf{0.68} &             0.31 &          0.22 &              0.42 &               0.80 &      \textbf{0.83} \\
                         \textbf{Gzip} &       0.61             &          0.24 &              0.56 &               0.75 &      \textbf{0.87} &       0.52           &          0.18 &              0.64 &               0.84 &      \textbf{0.97} \\
                         \textbf{Lua} &               0.67 &          0.16 &              0.25 &               0.62 &      \textbf{0.81} &             0.59 &          0.13 &              0.26 &               0.64 &      \textbf{0.83} \\
                         \textbf{PuTTY} &        0.46          &          0.09 &              0.21 &               0.48 &      \textbf{0.67} &          0.29        &          0.10 &              0.23 &               0.50 &       \textbf{0.7} \\
    \bottomrule
    \end{tabular}
    }
    \end{table}

%% file: tables/Exp4a-tables.tex
\begin{table}[ht]
    \centering
    \caption{Cross-Compiler + Cross-Optimization + Cross-Architecture Binary Diffing. Null hypothesis' expected score: 0.005}
    \label{tab:cross-comp-arch-opt}
    \resizebox{0.9\textwidth}{!}{
    \begin{tabular}{l|rrrrr|rrrrr}
    \toprule
                                                           & \multicolumn{5}{|c}{\textbf{Precision}} & \multicolumn{5}{|c}{\textbf{Recall}} \\
                                                           &   \textbf{BinDiff} & \textbf{SAFE} & \textbf{OPC} & \textbf{BinFinder} & \textbf{\textsc{\textbf{VexIR2Vec}}} & \textbf{BinDiff} & \textbf{SAFE} & \textbf{OPC} & \textbf{BinFinder} & \textbf{\textsc{\textbf{VexIR2Vec}}} \\
    \midrule
    
& \multicolumn{10}{c}{\textbf{\underline{x86-Clang12-O0 Vs. ARM-GCC10-O2}}} \\

     \textbf{Coreutils} &               0.10 &          0.08 &              0.21 &               0.28 &      \textbf{0.45} &             0.05 &          0.05 &              0.34 &               0.43 &       \textbf{0.7} \\
                                            \textbf{cURL} &               0.10 &          0.09 &              0.25 &               0.42 &      \textbf{0.57} &             0.06 &          0.06 &              0.33 &               0.56 &      \textbf{0.76} \\
                                            \textbf{Diffutils} &               0.19 &          0.09 &              0.25 &               0.31 &      \textbf{0.58} &             0.11 &          0.06 &              0.37 &               0.51 &      \textbf{0.86} \\
                                            \textbf{Findutils} &               0.14 &          0.07 &              0.17 &               0.39 &      \textbf{0.56} &             0.08 &          0.05 &              0.24 &               0.53 &      \textbf{0.76} \\
                                            \textbf{Gzip} &               0.13 &          0.08 &              0.12 &               0.39 &       \textbf{0.5} &             0.10 &          0.05 &              0.17 &               0.54 &       \textbf{0.7} \\
                                            \textbf{Lua} &               0.19 &          0.04 &              0.06 &               0.20 &      \textbf{0.34} &             0.16 &          0.03 &              0.08 &               0.30 &      \textbf{0.51} \\
                                            \textbf{PuTTY} &               0.06 &          0.03 &              0.04 &          0.19          &      \textbf{0.22} &             0.04 &          0.03 &              0.05 &          0.24          &       \textbf{0.3} \\
    \cline{1-11}
    
& \multicolumn{10}{c}{\textbf{\underline{x86-GCC8-O1 Vs. ARM-Clang6-O3}}} \\

     \textbf{Coreutils} &               0.10 &          0.13 &              0.38 &               0.41 &      \textbf{0.64} &             0.05 &          0.10 &              0.45 &               0.49 &      \textbf{0.76} \\
                                            \textbf{cURL} &               0.07 &          0.18 &              0.41 &               0.57 &      \textbf{0.74} &             0.04 &          0.13 &              0.43 &               0.59 &      \textbf{0.77} \\
                                            \textbf{Diffutils} &               0.24 &          0.17 &              0.36 &               0.45 &      \textbf{0.72} &             0.17 &          0.12 &              0.44 &               0.55 &      \textbf{0.89} \\
                                            \textbf{Findutils} &               0.19 &          0.14 &              0.39 &               0.47 &      \textbf{0.73} &             0.09 &          0.12 &              0.45 &               0.54 &      \textbf{0.83} \\
                                            \textbf{Gzip} &               0.28 &          0.13 &              0.17 &               0.47 &      \textbf{0.59} &             0.22 &          0.09 &              0.20 &               0.57 &      \textbf{0.72} \\
                                            \textbf{Lua} &               0.22 &          0.06 &              0.12 &               0.27 &      \textbf{0.46} &             0.19 &          0.04 &              0.15 &               0.34 &      \textbf{0.57} \\
                                            \textbf{PuTTY} &               0.08 &          0.06 &              0.09 &               0.25 &      \textbf{0.29} &             0.05 &          0.07 &              0.10 &               0.28 &      \textbf{0.32} \\
    \bottomrule
    \end{tabular}
    }
    \end{table}
    

%% file: tables/Exp4b-tables.tex
\begin{table}[ht]
    \centering
    \caption{Cross-Compiler + Cross-Optimization Binary Diffing. Null hypothesis' expected score: 0.006}
    \label{tab:cross-comp-opt}
    \resizebox{0.9\textwidth}{!}{
    \begin{tabular}{l|rrrrrr|rrrrrr}
    \toprule
                                                           & \multicolumn{6}{|c}{\textbf{Precision}} & \multicolumn{6}{|c}{\textbf{Recall}} \\
                                                           &   \textbf{BinDiff} & \textbf{DBD} & \textbf{SAFE} & \textbf{OPC} & \textbf{BinFinder} & \textbf{\textsc{\textbf{VexIR2Vec}}} & \textbf{BinDiff} & \textbf{DBD} & \textbf{SAFE} & \textbf{OPC} & \textbf{BinFinder} & \textbf{\textsc{\textbf{VexIR2Vec}}} \\
    \midrule
    
& \multicolumn{12}{c}{\textbf{\underline{x86-Clang12-O0 Vs. x86-GCC10-O2}}} \\

     \textbf{Coreutils} &               0.11 &         0.08 &          0.20 &              0.24 &               0.31 &      \textbf{0.48} &             0.06 &         0.23 &          0.16 &              0.34 &               0.43 &      \textbf{0.67} \\
                                            \textbf{cURL} &               0.14 &              &          0.27 &              0.41 &               0.50 &      \textbf{0.73} &             0.10 &              &          0.19 &              0.49 &               0.60 &      \textbf{0.86} \\
                                            \textbf{Diffutils} &               0.24 &         0.11 &          0.28 &              0.44 &               0.39 &      \textbf{0.63} &             0.17 &         0.36 &          0.19 &              0.60 &               0.53 &      \textbf{0.86} \\
                                            \textbf{Findutils} &               0.27 &         0.08 &          0.21 &              0.26 &               0.38 &       \textbf{0.6} &             0.16 &         0.26 &          0.17 &              0.34 &               0.49 &      \textbf{0.77} \\
                                            \textbf{Gzip} &               0.32 &         0.09 &          0.30 &              0.21 &               0.39 &      \textbf{0.55} &             0.25 &         0.37 &          0.21 &              0.27 &               0.52 &      \textbf{0.71} \\
                                            \textbf{Lua} &               0.24 &              &          0.16 &              0.11 &               0.21 &      \textbf{0.36} &             0.23 &              &          0.12 &              0.16 &               0.32 &      \textbf{0.53} \\
                                            \textbf{PuTTY} &               0.10 &              &          0.09 &              0.07 &               0.21 &      \textbf{0.24} &             0.09 &              &          0.09 &              0.09 &               0.28 &      \textbf{0.33} \\
    \cline{1-13}
    
& \multicolumn{12}{c}{\textbf{\underline{x86-GCC8-O1 Vs. x86-Clang6-O3}}} \\

\textbf{Coreutils} &               0.15 &         0.22 &          0.29 &              0.61 &               0.49 &      \textbf{0.72} &             0.11 &         0.46 &          0.25 &              0.67 &               0.54 &      \textbf{0.79} \\
                                            \textbf{cURL} &               0.13 &              &          0.39 &              0.66 &               0.73 &      \textbf{0.87} &             0.09 &              &          0.29 &              0.70 &               0.77 &      \textbf{0.91} \\
                                            \textbf{Diffutils} &               0.29 &         0.25 &          0.41 &              0.63 &               0.62 &      \textbf{0.81} &             0.25 &         0.60 &          0.29 &              0.71 &               0.70 &      \textbf{0.92} \\
                                            \textbf{Findutils} &               0.20 &         0.24 &          0.32 &              0.60 &               0.51 &      \textbf{0.81} &             0.15 &         0.57 &          0.27 &              0.66 &               0.56 &      \textbf{0.89} \\
                                            \textbf{Gzip} &               0.20 &         0.11 &          0.26 &              0.28 &               0.41 &      \textbf{0.54} &             0.17 &         0.54 &          0.21 &              0.33 &               0.49 &      \textbf{0.64} \\
                                            \textbf{Lua} &               0.29 &              &          0.28 &              0.28 &               0.34 &      \textbf{0.51} &             0.27 &              &          0.22 &              0.34 &               0.42 &      \textbf{0.62} \\
                                            \textbf{PuTTY} &               0.14 &              &          0.20 &              0.24 &               0.32 &      \textbf{0.37} &             0.11 &              &          0.23 &              0.26 &               0.35 &       \textbf{0.4} \\
    \bottomrule
    \end{tabular}
    }
    \end{table}

%% file: tables/Obf-tables.tex
\begin{table}[ht]
\centering
\caption{Diffing between binaries generated by GCC-10 (with -O0) and the obfuscated version of binaries generated by Clang-4 (with -O3). Null hypothesis' expected score: 0.006}
\label{tab:obfuscation}
\resizebox{0.9\textwidth}{!}{
\begin{tabular}{l|rrrrrr|rrrrrr}
\toprule
                                                            & \multicolumn{6}{|c}{\textbf{Precision}} & \multicolumn{6}{|c}{\textbf{Recall}} \\
                                                            &   \textbf{BinDiff} & \textbf{DBD} & \textbf{SAFE} & \textbf{OPC} & \textbf{BinFinder} & \textbf{\textsc{\textbf{VexIR2Vec}}} & \textbf{BinDiff} & \textbf{DBD} &  \textbf{SAFE} & \textbf{OPC} & \textbf{BinFinder} & \textbf{\textsc{\textbf{VexIR2Vec}}} \\
\midrule

& \multicolumn{12}{c}{\textbf{\underline{BCF}}} \\

\textbf{Coreutils} &               0.07 &         0.15 &          0.24 &         0.22 &               0.31 &      \textbf{0.56} &             0.04 &         0.16 &           0.35 &         0.28 &               0.39 &      \textbf{0.72} \\
                                             \textbf{cURL} &               0.05 &              &          0.34 &         0.29 &               0.68 &      \textbf{0.83} &             0.02 &              &           0.56 &         0.31 &               0.74 &       \textbf{0.9} \\
                                             \textbf{Diffutils} &               0.23 &         0.14 &          0.27 &         0.20 &               0.39 &      \textbf{0.64} &             0.13 &         0.32 &           0.49 &         0.27 &               0.52 &      \textbf{0.85} \\
                                             \textbf{Findutils} &               0.12 &         0.13 &          0.32 &         0.21 &               0.40 &      \textbf{0.68} &             0.07 &         0.25 &           0.42 &         0.25 &               0.47 &      \textbf{0.81} \\
                                             \textbf{Gzip} &               0.24 &              &          0.32 &         0.17 &      \textbf{0.37} &               0.27 &             0.16 &              &  \textbf{0.55} &         0.21 &               0.46 &               0.34 \\
                                             \textbf{PuTTY} &               0.04 &              &          0.23 &         0.06 &               0.17 &      \textbf{0.63} &             0.02 &              &           0.21 &         0.08 &               0.22 &      \textbf{0.79} \\
\cline{1-13}

& \multicolumn{12}{c}{\textbf{\underline{FLA}}} \\

 \textbf{Coreutils} &               0.09 &         0.14 &          0.20 &         0.15 &               0.35 &      \textbf{0.55} &             0.06 &         0.07 &           0.30 &         0.20 &               0.46 &      \textbf{0.72} \\
                                             \textbf{cURL} &               0.16 &              &          0.22 &         0.22 &               0.70 &       \textbf{0.8} &             0.08 &              &           0.39 &         0.25 &               0.80 &       \textbf{0.9} \\
                                             \textbf{Diffutils} &               0.22 &         0.14 &          0.20 &         0.10 &               0.40 &      \textbf{0.64} &             0.16 &         0.23 &           0.36 &         0.14 &               0.54 &      \textbf{0.86} \\
                                             \textbf{Findutils} &               0.17 &         0.13 &          0.24 &         0.13 &               0.37 &      \textbf{0.63} &             0.12 &         0.20 &           0.34 &         0.16 &               0.45 &      \textbf{0.77} \\
                                             \textbf{Gzip} &               0.25 &              &          0.22 &         0.20 &      \textbf{0.32} &               0.28 &             0.18 &              &           0.41 &         0.27 &      \textbf{0.42} &               0.37 \\
                                             \textbf{PuTTY} &               0.08 &              &          0.17 &         0.06 &               0.22 &      \textbf{0.61} &             0.07 &              &           0.16 &         0.08 &               0.30 &       \textbf{0.8} \\
\cline{1-13}

& \multicolumn{12}{c}{\textbf{\underline{SUB}}} \\

\textbf{Coreutils} &               0.11 &         0.15 &          0.25 &         0.34 &               0.31 &      \textbf{0.54} &             0.06 &         0.08 &           0.38 &         0.45 &               0.42 &      \textbf{0.72} \\
                                             \textbf{cURL} &               0.07 &              &          0.35 &         0.56 &               0.66 &      \textbf{0.79} &             0.05 &              &           0.63 &         0.65 &               0.76 &      \textbf{0.91} \\
                                             \textbf{Diffutils} &               0.28 &         0.16 &          0.27 &         0.44 &               0.39 &      \textbf{0.64} &             0.18 &         0.31 &           0.51 &         0.60 &               0.54 &      \textbf{0.88} \\
                                             \textbf{Findutils} &               0.12 &         0.14 &          0.30 &         0.38 &               0.34 &      \textbf{0.65} &             0.07 &         0.26 &           0.44 &         0.48 &               0.43 &      \textbf{0.82} \\
                                             \textbf{Gzip} &               0.25 &              &          0.32 &         0.32 &      \textbf{0.38} &               0.25 &             0.17 &              &  \textbf{0.61} &         0.43 &               0.51 &               0.34 \\
                                             \textbf{PuTTY} &               0.07 &              &          0.30 &         0.08 &               0.18 &      \textbf{0.62} &             0.05 &              &           0.29 &         0.11 &               0.24 &      \textbf{0.84} \\

\cline{1-13}

& \multicolumn{12}{c}{\textbf{\underline{ALL}}} \\

\textbf{Coreutils} &               0.06 &         0.14 &          0.16 &         0.08 &               0.29 &      \textbf{0.58} &             0.04 &         0.03 &           0.24 &         0.11 &               0.36 &      \textbf{0.74} \\
                                             \textbf{cURL} &               0.00 &              &          0.16 &         0.15 &               0.48 &      \textbf{0.78} &             0.00 &              &           0.26 &         0.16 &               0.52 &      \textbf{0.84} \\
                                             \textbf{Diffutils} &               0.29 &         0.11 &          0.15 &         0.11 &               0.35 &      \textbf{0.65} &             0.20 &         0.24 &           0.26 &         0.14 &               0.45 &      \textbf{0.81} \\
                                             \textbf{Findutils} &               0.08 &         0.11 &          0.16 &         0.08 &               0.31 &      \textbf{0.61} &             0.04 &         0.23 &           0.22 &         0.10 &               0.37 &      \textbf{0.74} \\
                                             \textbf{Gzip} &               0.10 &              &          0.17 &         0.10 &      \textbf{0.34} &               0.28 &             0.07 &              &           0.28 &         0.12 &      \textbf{0.41} &               0.35 \\
                                             \textbf{PuTTY} &               0.04 &              &          0.11 &         0.03 &               0.14 &      \textbf{0.65} &             0.02 &              &           0.10 &         0.04 &               0.17 &      \textbf{0.77} \\

\bottomrule
\end{tabular}
}
\end{table}

%% file: 5b-experiment-searching.tex
\subsection{RQ2: Searching and Retrieval}
\label{subsec:exp-searching}

In this section, we evaluate our approach on the searching task described in Section~\ref{def:binsim-tasks} to provide an answer to \hyperref[RQ2]{RQ2}.

\subsubsection{Ground Truth}
\label{sss:gt_searching}

We create a pool of functions from the test set of x86 binaries generated with Clang (V8.0.1 and V12.0.1) and GCC (V8.0 and V10.0), at four optimization levels: \texttt{-O0}, \texttt{-O1}, \texttt{-O2} and \texttt{-O3}.
In each search, functions from one compilation configuration are used as the query set.
All the other functions, excluding the query set, form the set of search candidates.
We obtain the function identifiers by following the same process described in Section~\ref{subsec:exp-diffing} to form the ground truth. 
In any search set, there are $15$ candidates that can match a query ($2 \ \mbox{compilers} \times 2 \ \mbox{versions} \times 4 \ \mbox{optimizations} - \mbox{input query}=15$).
On average, a query set has about $2K$ functions, and a search set has about $323K$ functions; thus, the chance of correctly retrieving a match for a query with a random guess is $\approx 15/323K$.

\begin{figure}[ht]
    \centering
    \includegraphics[width=\linewidth]{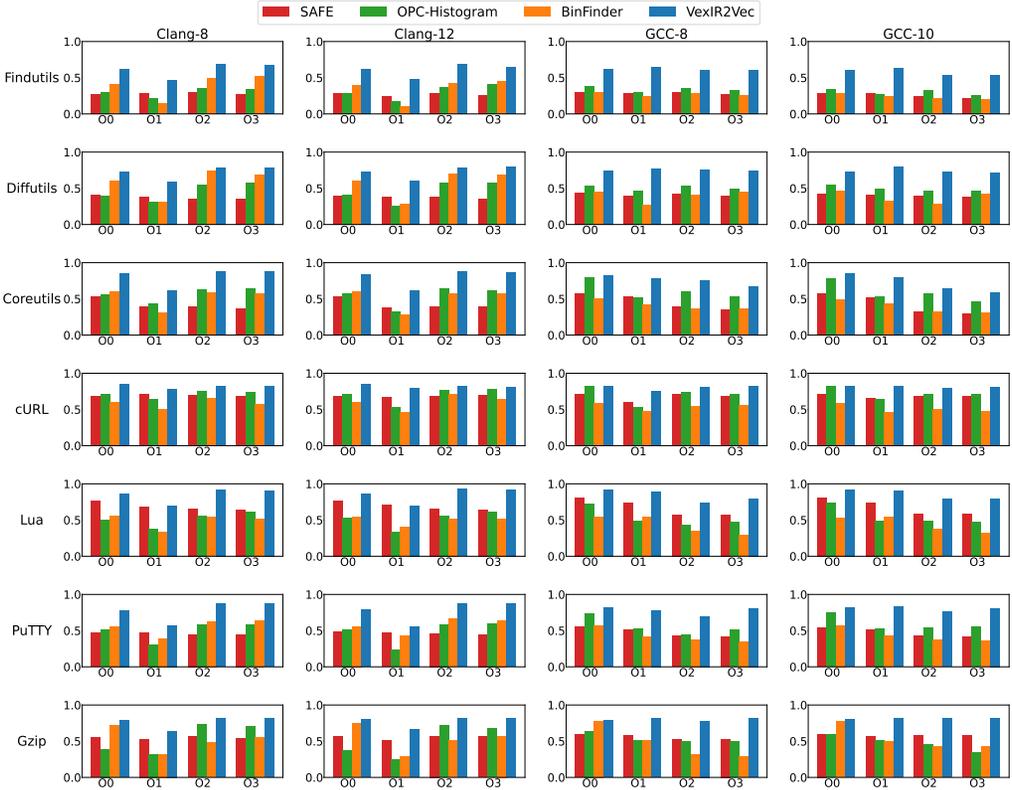}
    \caption{Mean Average Precision (MAP) scores for searching experiments.}
    \Description{MAP scores for searching experiment}
    \label{fig:MAP-Scores-searching}
    \vspace{-\baselineskip}
\end{figure}

\subsubsection{Metric}
We use the Mean Average Precision (MAP) as the evaluation metric for the searching task.
MAP is computed as $\frac{1}{Q}\sum_{q=1}^{Q}AP(q)$, where $AP(q)$ is the average precision for query $q$. $Q$ is the total number of queries. $AP(q)$ is defined as $\frac{1}{N}\sum_{i=1}^{N}Prec(i)\times Rel(i)$. 
$Prec(i)$ is precision at $i$; 
$Rel(i)$ denotes the relevance at $i$. $Rel(i)$ is set to $1$ when the retrieved function at position $i$ matches the query function, and is set to $0$ otherwise. $N$ is the total number of retrieved functions that match the query function.

\subsubsection{Mixture of all - searching}
\label{sec:all-searching}

Each query function is searched against the candidate set of functions. Similar to Section~\ref{subsec:exp-diffing}, we obtain a list of $10$ nearest functions to the query function among the set of candidates ranked by the Euclidean distance. 
A function from the list is considered to match the query function if both of them have the same source file and function name.
We do not consider BinDiff and DeepBinDiff for the searching experiments as they are designed only for binary diffing.
Figure~\ref{fig:MAP-Scores-searching} shows the MAP scores produced by SAFE, BinFinder, OPC, and \vexirtovec{} for this task.
\vexirtovec{} consistently obtains higher MAP scores than the other baselines across all compiler configurations.
The MAP score of \vexirtovec is $0.76$ on average, while SAFE's, BinFinder's, and OPC's are $0.5$, $0.47$, and $0.52$, respectively.

%% file: 5c-vocabulary-eval.tex
\subsection{RQ3: Evaluation of Vocabulary}
\label{sec:vocab-eval}

Our vocabulary is the set of 128-dimensional vectors that represent the entities that form the program's intermediate format: opcodes, types, and operands.
As explained in Section~\ref{sub:lookup}, this vocabulary is encoded as a function $\mathcal{V}_{lookup}$, which maps entities to vectors.
This section presents different experiments---essentially of a qualitative nature---that we have engineered to answer (\hyperref[RQ3]{RQ3}), which deals with the capacity of $\mathcal{V}_{lookup}$ to encode meaningful semantic information.

\subsubsection{Clustering}
\label{sss:vocab_clustering}

$\mathcal{V}_{lookup}$ encodes semantic information as points within an Euclidean Space.
Thus, entities that are semantically related tend to be mapped to points that are close in this space, as illustrated in Figure~\ref{fig:derivativeBasedProgramming} (Page~\pageref{fig:derivativeBasedProgramming}).
To provide some evidence that this semantic information is correctly encoded in $\mathcal{V}_{lookup}$, we use t-SNE~\cite{van2008tsne} to project the $128$-D entity vectors onto a bi-dimensional surface.
Figure~\ref{fig:vocab-clusters} shows the resulting image.
For easier interpretation, we mark entity vectors into nine logical groups.
These groups are either based on the type of the data they operate on (integers, floating-point numbers, vectors), or on the kind of operations (logic operations, loads, stores, comparisons, max/min).
We also have a group for entities to denote user-defined symbols: variables, functions, constants, etc.

\begin{table}[ht]
    \centering
    \begin{minipage}{0.53\textwidth}
        \centering
        \includegraphics[width=\textwidth]{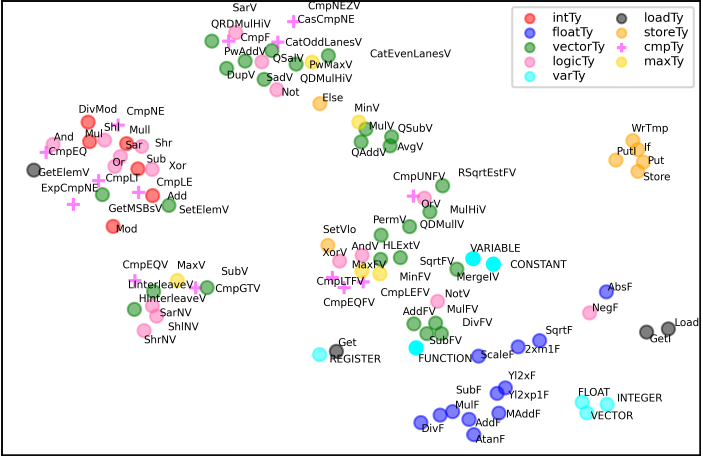}
        \captionof{figure}{Vocabulary Clusters}
        \Description{Vocabulary Clusters}
        \label{fig:vocab-clusters}
    \end{minipage}
    \qquad
    \begin{minipage}{0.26\textwidth}
        \centering
        \captionof{table}{Example Analogies}
        \label{tab:some-analogies}
        \resizebox{\textwidth}{!}{
        \begin{tabular}{l}
        \toprule
        \multicolumn{1}{c}{\textbf{Syntactic Analogies}}\\ 
        \midrule
        getI : putI    :: get   : put                 \\ 
        subf : addf    :: subfv : addfv   \\ 
        add  : addf    :: sub   : subf       \\ 
        ext  : integer :: extf  : float  \\ 
        orv  : vector  :: or    : integer \\
        \hline
        \multicolumn{1}{c}{\textbf{Semantic Analogies}} \\ 
        \hline
        shl  : mul     :: shr : divmod             \\
        mul  : divmod  :: and : or              \\
        get  : put :: load  : store \\
        maxfv : minfv :: addfv : subfv \\
        get  : register :: geti : constant \\
        \bottomrule
        \end{tabular}
        }
    \end{minipage}
\end{table}

Figure~\ref{fig:vocab-clusters} shows the distinct formation of smaller clusters in the case of integer types, float types, and store types.
The cluster of logical types containing operations such as \texttt{xor} are closer to the integer types.
This proximity comes from the fact that logical operations actuate on integer-type variables.
Similar observations can also be made for some of the vector types lying in closer proximity to some of the Ext types, such as \texttt{truncv}.

\subsubsection{Analogies}
\label{sss:vocab_analogies}

An {\it analogy} in a machine-learning embedding is a relationship between vectors that reflects a similar relationship between the entities they represent, often expressed in the form: "{\it a is to b as c is to d}."
\vexirtovec{} lets us explore several kinds of analogies.
For instance, an analogy such as ``\texttt{get} is to \texttt{put} as \texttt{load} is to \texttt{store}'' explores the similarity between different semantics of data movement instructions related to register and memory accesses.
In order to probe \vexirtovec's capacity to build meaningful analogies, we have designed $90$ analogy queries to cover relations among operators, types, arguments, and their semantics.
We represent an analogy query as \textit{a\,:\,b\,::\,c\,:\,?}, meaning ``{\it if a is to b, then c is to which entity?}''
To answer the query, the missing value is computed as the entity the closest to $b - a + c$ by Euclidean distance.
Appendix~\ref{appendix:clusters-analogies} gives the full list of 90 queries.
Table~\ref{tab:some-analogies} lists some of the analogies that were correctly captured by $\mathcal{V}_{lookup}$.
Notably, the vocabulary is able to capture intrinsic semantic information to capture that multiplication and division operations can be achieved using left and right shift operations, respectively. Other relations that characterize the type of operand for each operation are also captured.
In this exercise, we observed the expected query result for 63 out of 90 queries, totaling an accuracy of 70\%.

\subsubsection{Number of Dimensions}
\label{sss:vocab_dimensions}

The queries explored in Section~\ref{sss:vocab_analogies}
gave us enough data to fine-tune hyperparameters (learning rate, batch size, and margin) besides the number of dimensions used in the representation of entities.
As shown in Figure~\ref{fig:dim-oov}(a), we observe an increase in the accuracy in correctly answering the number of analogies with the increase in the number of dimensions until $128$.
Further increase in the number of dimensions resulted in a fall in accuracy.

\begin{figure}[ht]
    \vspace{-\baselineskip}
    \centering
    \subfloat[Accuracy vs Dimensions]{\includegraphics[width=0.48\linewidth]{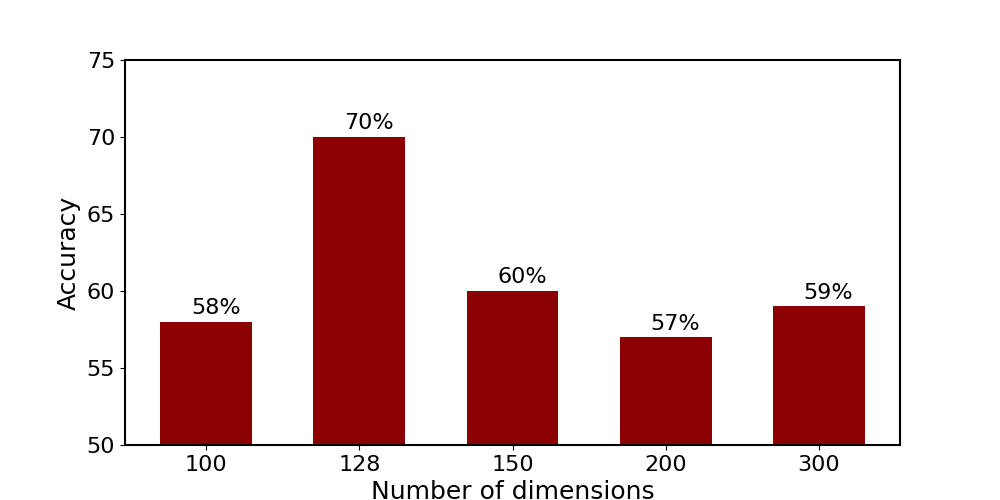}} \quad
    \subfloat[\#OOV Occurrences]{{\includegraphics[width=0.4\linewidth]{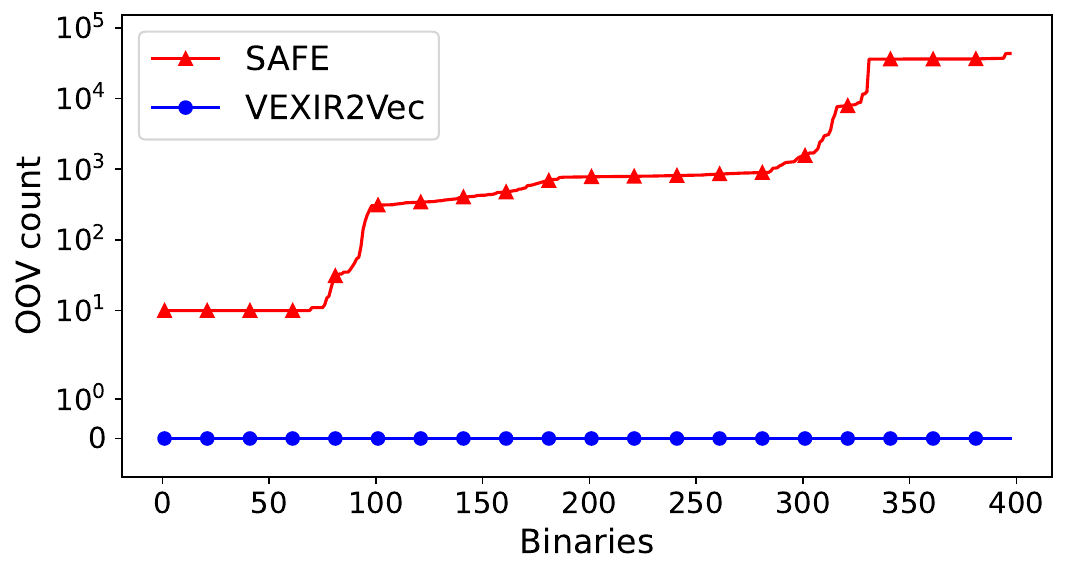}} }%
    \caption{\#Dimensions of Embeddings and OOV Study}
    \Description{\#Dimensions of Embeddings and OOV Study}
    \label{fig:dim-oov}
    \vspace{-\baselineskip}
\end{figure}

\subsubsection{OOV study}
\label{sss:oov_study}
The performance of supervised learning approaches typically depends on the quality and comprehensiveness of the training data. When the training data does not cover the entire vocabulary of the language, unseen words can occur at test time, leading to Out-Of-Vocabulary (OOV) issues. These OOV words, which are not seen by the model during training, create challenges in accurately predicting and understanding the input binaries during inference.

The approaches that do not model input properly suffer from serious OOV issues. Approaches like SAFE~\cite{zuo2018innereye} and InnerEye~\cite{zuo2018innereye} learn the representations at the instruction level. And, typically there could be a very high number of combinations of opcodes ($O$), types ($T$), and the number of arguments ($A$). To avoid OOV issues, such approaches should use a training set comprising all possible combinations of opcodes, types, and arguments, leading to an enormous training space of $|O| \times |T| \times |A|$. Covering this huge space is highly infeasible.  

In contrast, \vexirtovec learns the representations at the entity level, reducing the training space to  $|O| + |T| + |A|$. 
These entities exist in limited numbers and can be quickly learned. By covering all possible opcodes, types, and arguments individually, it is easier to avoid OOV issues with a limited training set in \vexirtovec.

Figure~\ref{fig:dim-oov}(b) shows the number of OOV occurrences encountered by SAFE and \vexirtovec during inference on a collection of $400$ randomly chosen binaries.
It can be observed that \vexirtovec did not face any OOV issues.
SAFE, in contrast, encountered a large number of ($\approx 10^5$) OOV words.
This study does not consider the implementations of OPC and BinFinder, because they are feature-based, not relying on a vocabulary of embeddings.

%% file: 5d-scalability.tex
\subsection{RQ4: Scalability}
\label{sub:scalability}

In this section, we evaluate the scalability of \vexirtovec{} to answer \hyperref[RQ4]{RQ4}.

\subsubsection{Time for Embedding Generation}
\label{sss:embedding_time}

The implementation of \vexirtovec{} features two levels of parallelism:
\begin{description}
    \item[Thread level:] Each function of the binary is processed in parallel by different threads to obtain the function-level embedding.
    \item[Task level:] Each binary is processed in parallel via a new process.
\end{description}
Figure~\ref{fig:opttime-oov} (a) provides data that demonstrates that this parallelization yields benefits.
The figure reports running times to generate embeddings for a randomly chosen set of $100$ binaries from our dataset, whose sizes range between $15KB$ and $5MB$.
Lines show the cumulative time observed for the different binaries.
These binaries are sorted by size along the X-axis.
The implementation of \vexirtovec{} with one, two, four and eight threads takes $490$, $335$, $264$, $221$ seconds,
respectively, to process the $100$ binary files.
Thus, the eight-threaded implementation of \vexirtovec is $2.2 \times$ faster than its sequential version in this experiment.

\begin{figure}[ht]
    \vspace{-\baselineskip}
    \centering
    \subfloat[Embedding Generation Time]{\includegraphics[width=0.5\linewidth]{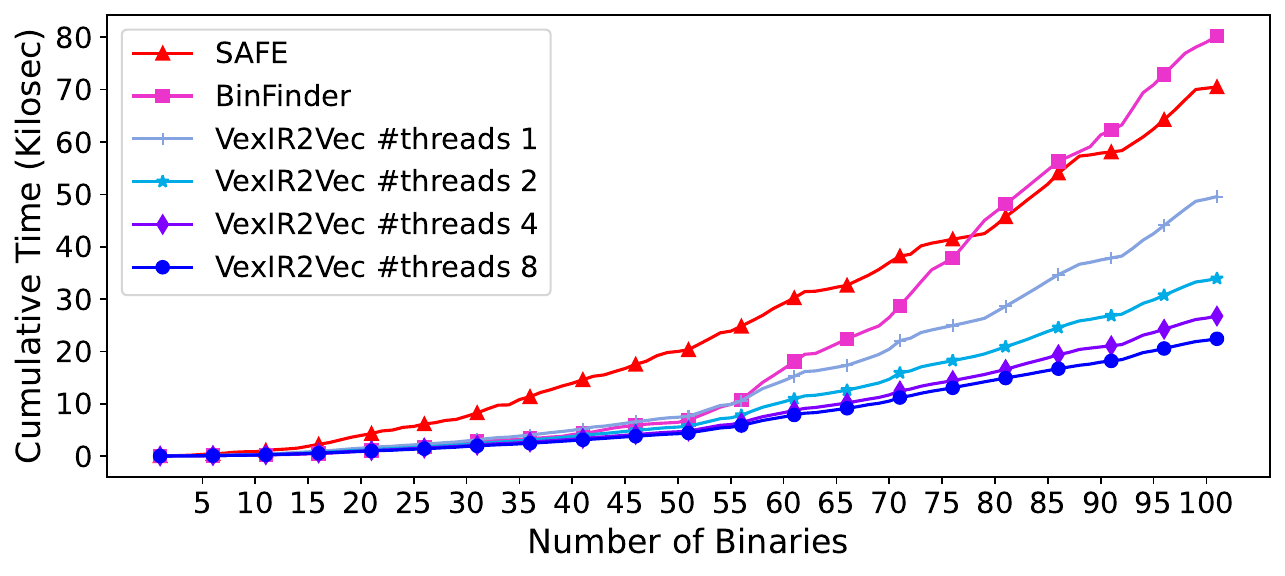}} \quad
    \subfloat[Normalization Time]{{\includegraphics[width=0.4\linewidth]{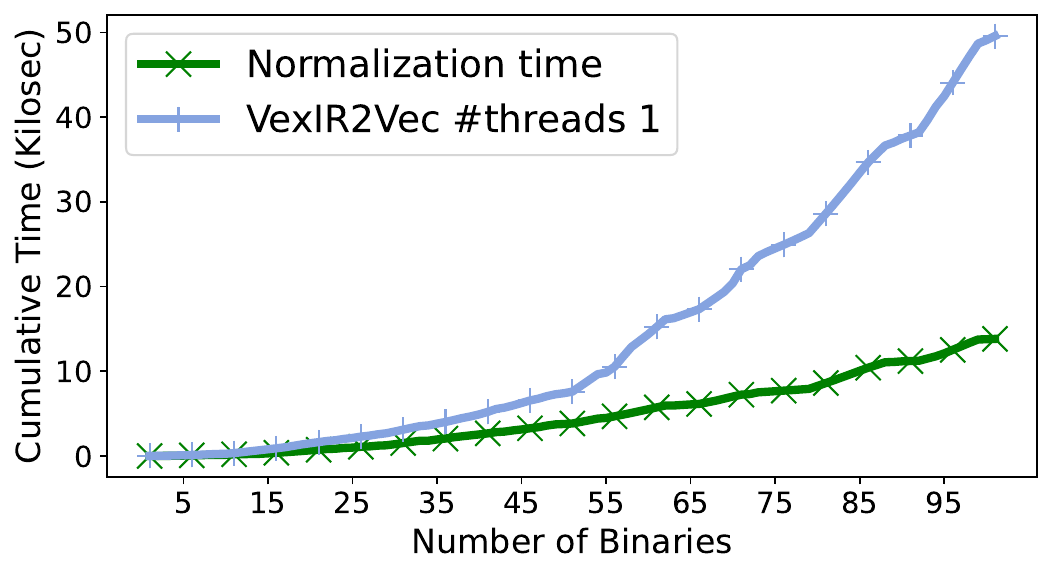}}}
    \caption{Time Taken to generate and normalize embeddings.}
    \Description{Time Taken by \optengine and OOV Study}
    \label{fig:opttime-oov}
    \vspace{-\baselineskip}
\end{figure}

Figure~\ref{fig:opttime-oov} (a) also reports running times to generate embeddings via SAFE and BinFinder.
We omit BinDiff and OPC as they do not generate embeddings.
We also omit DeepBinDiff because, in this experiment, this tool takes about $6,800$ seconds on average, timing out for binaries greater than $300KB$.
SAFE and Binfinder take $700$ and $795$ seconds, respectively, to process the 100 binaries; hence, these tools tend to be much slower than \vexirtovec.
\vexirtovec{} running with one thread achieves a speedup of $1.62 \times$ and $1.42 \times$ over BinFinder and SAFE, respectively.
With eight threads, this speedup grows to $3.5 \times$ and to $3.11 \times$, respectively.
BinFinder and SAFE also seem to be impacted by worst asymptotic behavior: the time these tools need to produce embeddings grows orders of magnitude faster with the binary size compared to \vexirtovec.

\subsubsection{Time taken by \optengine}
\label{sss:time_optengine}

Figure~\ref{fig:opttime-oov} (b) shows the time that \optengine takes to normalize the same 100 binaries used in Section~\ref{sss:embedding_time}.
All the normalizations described in Section~\ref{sub:normalization} can be implemented to run in linear time on the number of instructions that constitute a peephole.
Linearity comes from the fact that peepholes are straight-line sequences of code; hence, each optimization visits each instruction only once.
Figure~\ref{fig:opttime-oov} (b) indicates that the time required to normalize peepholes grows more slowly than the overall time taken by \vexirtovec{} to produce embeddings.
The cumulative time taken by \optengine is about $25\%$ of the time that \vexirtovec{} takes to produce embeddings.

%% file: 6.ablation.tex
\section{Ablation Study}
\label{sec:ablation}

\vexirtovec{} is customizable.
For instance, its scalability and precision can be controlled by parameters of Algorithm~\ref{algorithm:random-walk}, such as the length $k$ of peepholes and the number $c$ of times each basic block is visited to produce a peephole.
Similarly, precision can be enhanced with the addition of phases such as the normalization proposed in Section~\ref{sub:normalization}, which is entirely optional.
In other words, \vexirtovec{} can still be extracted from non-normalized peepholes.
To evaluate the impact of different customizations on the precision and on efficiency of \vexirtovec, this section provides answers to the following research questions:

\begin{enumerate}[label=(RQ\arabic*)]
    \setcounter{enumi}{4}
    \item \labelText{}{RQ5} How the parameters of Algorithm~\ref{algorithm:random-walk} ($k$ and $c$) impact the precision and the scalability of \vexirtovec? (Section~\ref{sub:random_walk})
    \item \labelText{}{RQ6} What is the impact of the different types of normalization implemented by \optengine on the precision of \vexirtovec? (Section~\ref{sub:eval_normalization})
    \item \labelText{}{RQ7} Can the proposed normalizations be adapted to work with the baseline classifiers used in this work, and in case such is possible, how would they change the precision of these tools? (Section~\ref{sub:eval_norm_baselines})
    \item \labelText{}{RQ8} What is the contribution of strings and library calls to the precision of \vexirtovec? (Section~\ref{sub:eval_strings})
    \item \labelText{}{RQ9} What is the relative importance of the different program entities in the precision of \vexirtovec? Or, in other words, how much attention does \vexnet{} put on each kind of entity? (Section~\ref{sub:attention_eval})
\end{enumerate}

\subsection{RQ5: Impact of the Parameters of Random Walk}
\label{sub:random_walk}

The random walk algorithm to generate peepholes described in Algorithm~\ref{algorithm:random-walk} is parameterized by the maximum length of the peephole $k$ and the minimum number of visits per basic block $c$.
The experiments discussed in Section~\ref{sec:performance-evaluation} use $k = 72$ and $c = 2$.
We have arrived at this configuration after testing different values.
This section describes this search, using, to this end, the setting involving cross-architecture and mixture-of-all diffing experiments.

Figure~\ref{fig:ablation1}(a) shows the variation in F1 score for different values of $k$.
The F1 score increases with $k$ from $1$ to $72$.
This growth indicates the benefit of additional contextual information that comes with larger peepholes.
However, past $k = 72$, we start observing a decrease in F1 scores.
Indeed, $k=144$ and $k=32$ deliver almost identical results in the cross-architecture/mixture-of-all setting.
This behavior suggests that context improves the precision of \vexirtovec; however, only up to a certain limit.

\begin{figure}[ht]
    \centering
    \vspace{-\baselineskip}
    \subfloat[Impact of $k$]{{\includegraphics[width=0.37\linewidth]{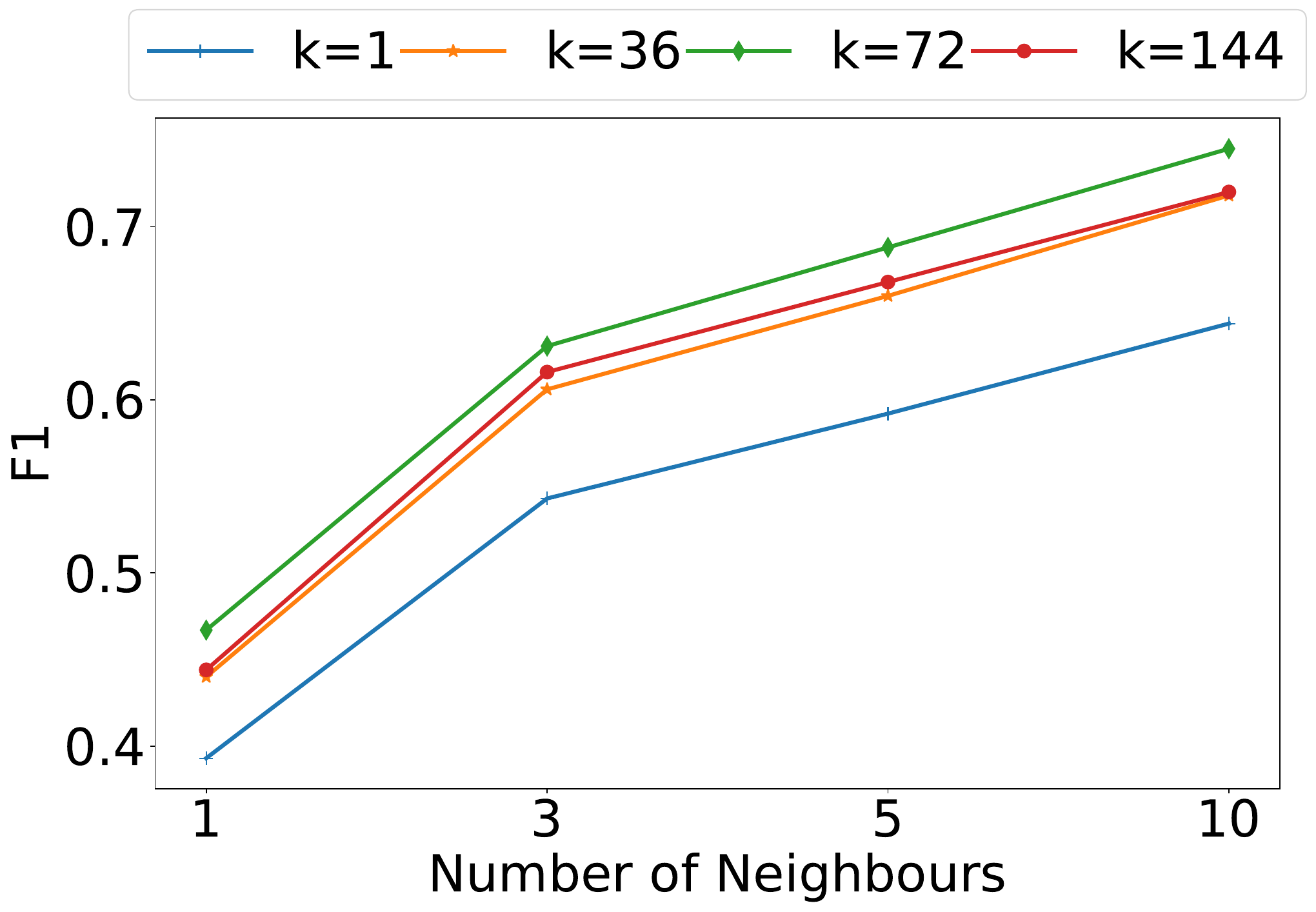}}} \quad
    \subfloat[Impact of $c$]{{\includegraphics[width=0.37\linewidth]{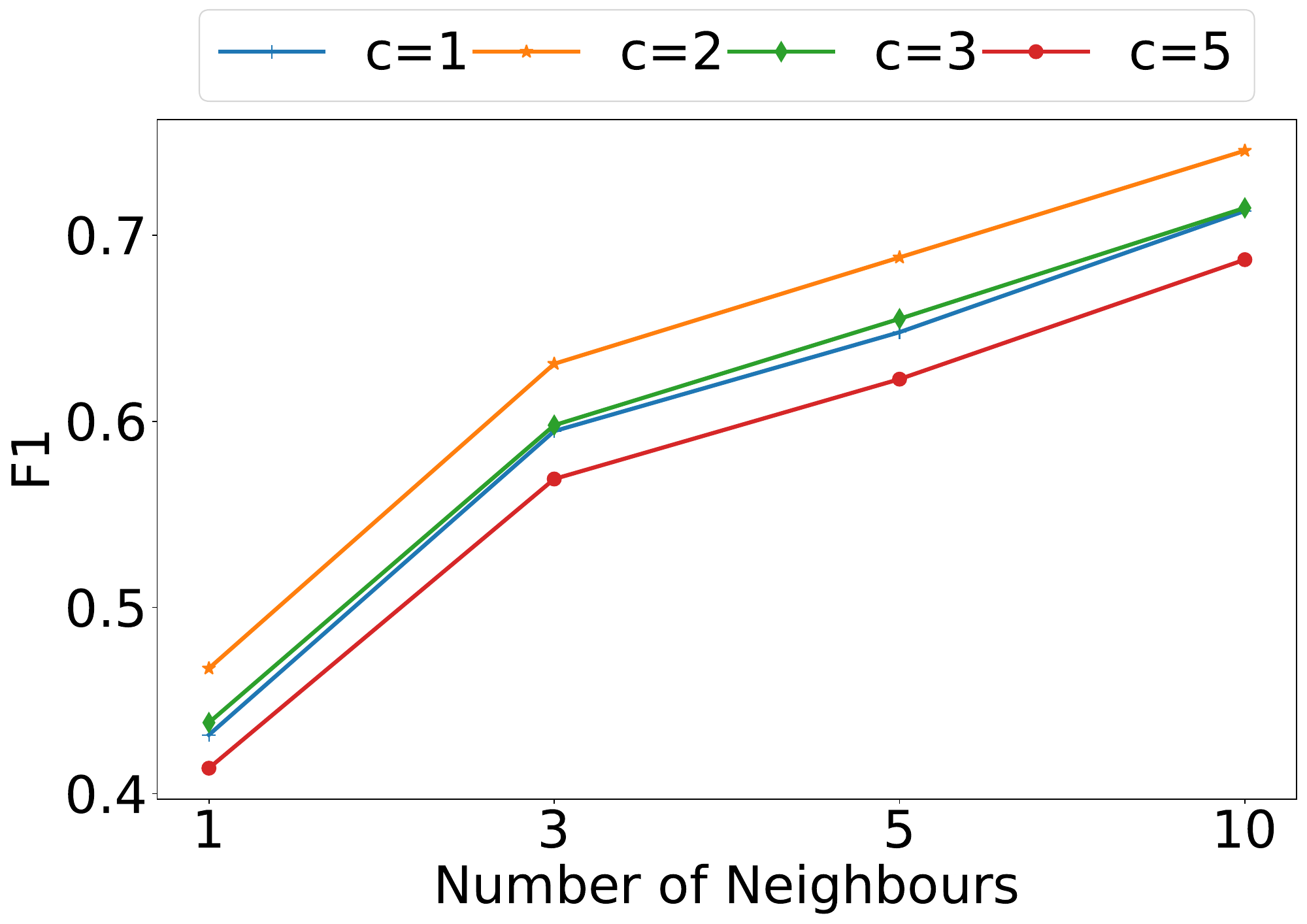}}}
    \caption{The impact of the parameters of the random walk (Algorithm~\ref{algorithm:random-walk}) on the F1 scores of \vexirtovec.}
    \Description{Impact of Random Walk Parameters}
    \label{fig:ablation1}
    \vspace{-\baselineskip}
\end{figure}

Similar to the experiment on $k$, we vary the values of $c$, the number of times each basic block is visited by Algorithm~\ref{algorithm:random-walk}, and record F1 scores.
Figure~\ref{fig:ablation1}(b) shows the results of this experiment.
F1 scores increase from $c=1$ to $c=2$.
However, upon increasing $c$ further to $3$ and $5$, we observe a reduction in F1 scores.
The results obtained for $c=3$ closely match those seen with $c=1$.
This behavior indicates a similar trend as that of $k$:
it is better to visit a basic block multiple times than just once.
However, there seems to exist a limit to how much information can be derived from the extra visits.
In our case, this limit is two.

\subsection{RQ6: Effectiveness of Normalization}
\label{sub:eval_normalization}

To answer \hyperref[RQ6]{RQ6}, we have created four sets of normalizations, each containing an increasing collection of the transformations introduced in Section~\ref{sub:normalization}: 

\begin{enumerate}
    \item[$N0$:] No normalization.
    \item[$N1$:] Register optimizations, Redundant Write Elimination, and Copy propagation.
    \item[$N2$:] $N1$ along with Constant Propagation, Constant Folding, and Common Subexpression Elimination.
    \item[$N3$:] All the normalizing transformations described in Section~\ref{sub:normalization}.
\end{enumerate}

Figure~\ref{fig:norm-ablation}(a) shows the F1 score measured when \vexirtovec is equipped with each of these sets of normalizing transformations.
In these experiments, a new model is trained from scratch with the obtained embeddings generated by the normalized peepholes.
We observe an increasing trend in performance with the increase in normalization level.
Consequently, $N3$ yields the highest F1 score.
$N0$, in turn, yields the lowest.
Nevertheless, even in this poor configuration---without any support of normalization---\vexirtovec{} is still able to produce higher F1 scores than the baseline approaches.
This empirical observation supports the trend observed in Examples~\ref{ex:opt-steps} and \ref{ex:final_optimization}.
These examples provide intuition on how successive normalizing transformations simplify peepholes, hence filtering out non-essential syntactic elements and revealing their essential semantic characteristics.

\begin{figure}[ht]
    \vspace{-\baselineskip}
    \centering
    \subfloat[Levels of Normalizations on \vexirtovec]{{\includegraphics[width=0.39\linewidth]{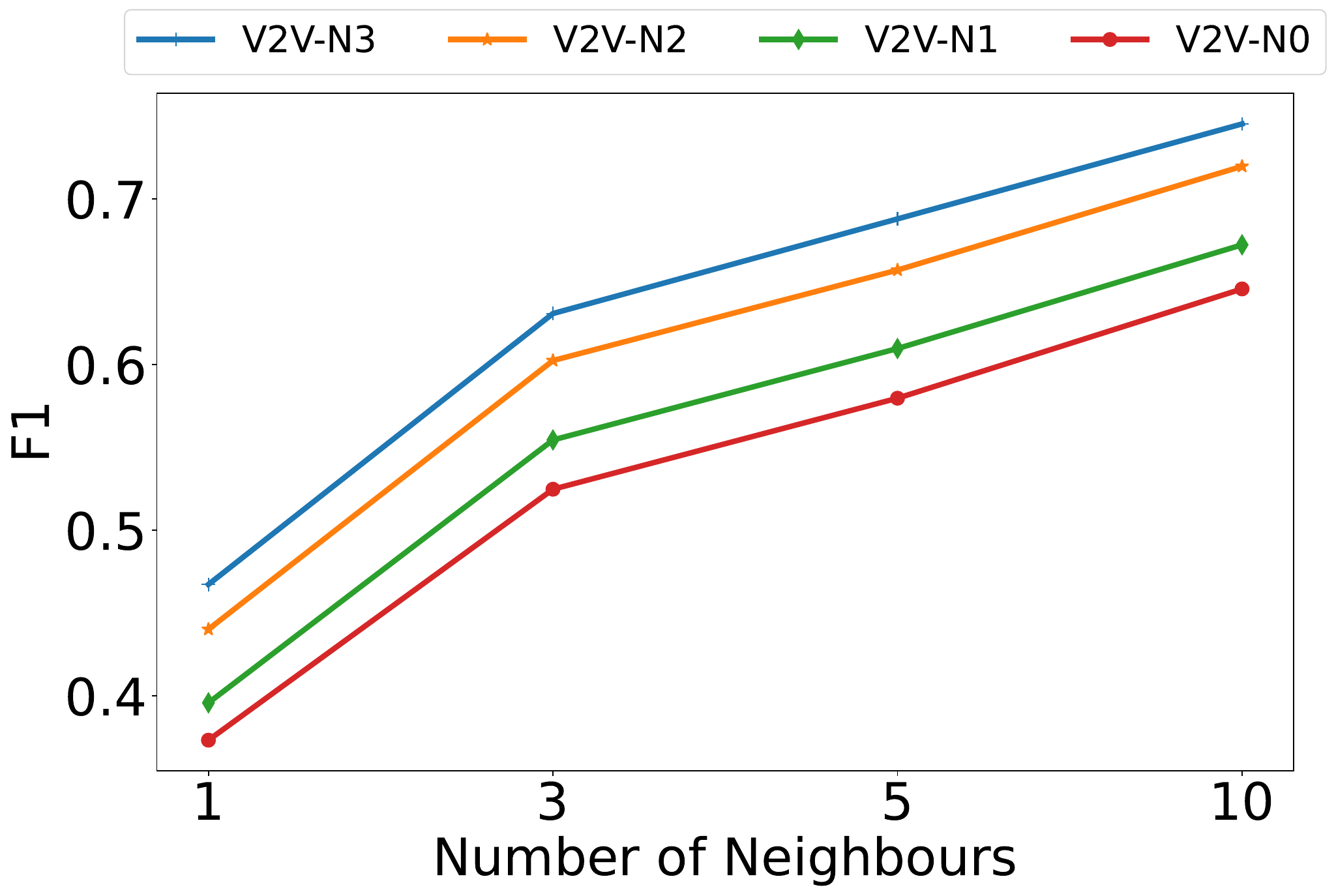}} }%
    \qquad
    \subfloat[Normalizations on BinFinder \& OPC]{{\includegraphics[width=0.39\linewidth]{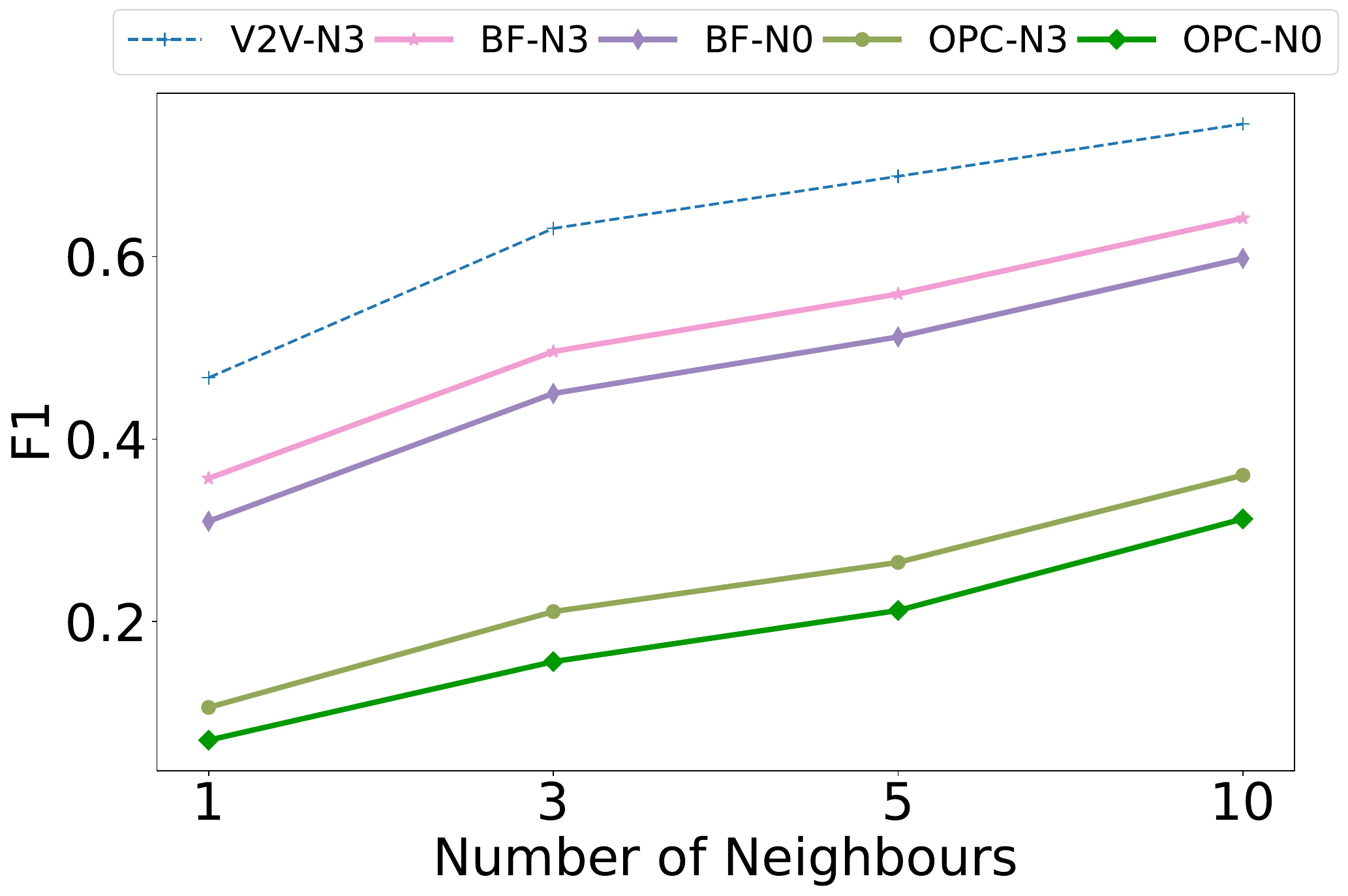}} }%
    \caption{The impact of the normalizing transformations (Sec.~\ref{sub:normalization}) on the F1 scores of \vexirtovec.}
    \Description{Effectiveness of Normalization}
    \label{fig:norm-ablation}
    \vspace{-\baselineskip}
\end{figure}

\subsection{RQ7: Impact of Normalization on Other Baselines}
\label{sub:eval_norm_baselines}

Although we have designed the normalizing transformations of Section~\ref{sub:normalization} to work in tandem with \vexirtovec, they can still be of service in other binary similarity tools.
Question \hyperref[RQ7]{RQ7} explores this possibility.
To answer this research question, we consider two baselines: BinFinder and Histograms of Opcodes.
These are the only baselines that we see how to augment with normalizations.

We have modified the implementation of these baselines, so that they would extract embeddings from normalized peepholes.
Peepholes are extracted from the very Algorithm~\ref{algorithm:random-walk} that empowers \vexirtovec.
Decomposing functions into peepholes is fundamental to carry out this experiment, for the normalizing transformations were designed to work on straight-line sequences of instructions.
Embeddings are extracted per peephole and then merged via vector addition, as done for \vexirtovec.

Figure~\ref{fig:norm-ablation}(b) shows the resulting F1 scores for these two modified baselines, using two versions of each one of them.
The first version, with the suffix $N0$, is the original implementation without any normalization.
The other version, with the suffix $N3$, is the modified version, which incorporates all the transformations described in Section~\ref{sub:normalization}.
Figure~\ref{fig:norm-ablation}(b) shows that normalization improves F1 scores in both the baselines.
Notice, however, that even when equipped with the highest normalization level ($N3$, as explained in Section~\ref{sub:eval_normalization}), the baselines still lag behind \vexirtovec{} running on non-normalized peepholes.
Nevertheless, this experiment implies that compiler-inspired normalizations, as a general pre-processing transformation, are general enough to be deployed onto different embedding functions.

\subsection{RQ8: Contribution of Strings and Library Calls}
\label{sub:eval_strings}

As Section~\ref{sss:strings} explains, the final embedding $\embedFFinal$ that characterizes the \vexirtovec{} representation of a function includes information extracted from strings and library calls.
This information is encoded in the $\embedEqn{S}$ and $\embedEqn{L}$ vectors that Section~\ref{sub:ean} introduces.
\hyperref[RQ8]{RQ8} asks for the importance of these vectors in the design of \vexirtovec.
This section provides an answer to this question.

To answer \hyperref[RQ8]{RQ8}, we train the \vexnet model only by considering $\llangle \embedEqn{O}, \embedEqn{T}, \embedEqn{A} \rrangle$; that is, without adding in the $\embedEqn{S}$ and $\embedEqn{L}$ vectors.
Figure~\ref{fig:ablation2}(a) shows the resulting F1 scores.
On average, strings and library calls improve the F1 scores across different neighbors by about $0.1$ ($20\%$).
Thus, strings and library calls provide \vexirtovec{} with non-trivial information.
They enhance the quality of the function embeddings and, consequently, the performance of \vexnet.
Nevertheless, even without strings and library calls, \vexirtovec{} remains more precise than the competing baselines.
For instance, while BinFinder and SAFE achieve average F1 scores of $0.58$ and $0.32$, respectively, \vexirtovec, even without the extra data provided by strings and library calls, achieves an F1 score of $0.63$.

\begin{figure}[ht]
    \centering
    \vspace{-\baselineskip}
    \subfloat[Impact of Strings and Library Calls]{\includegraphics[scale=0.4]{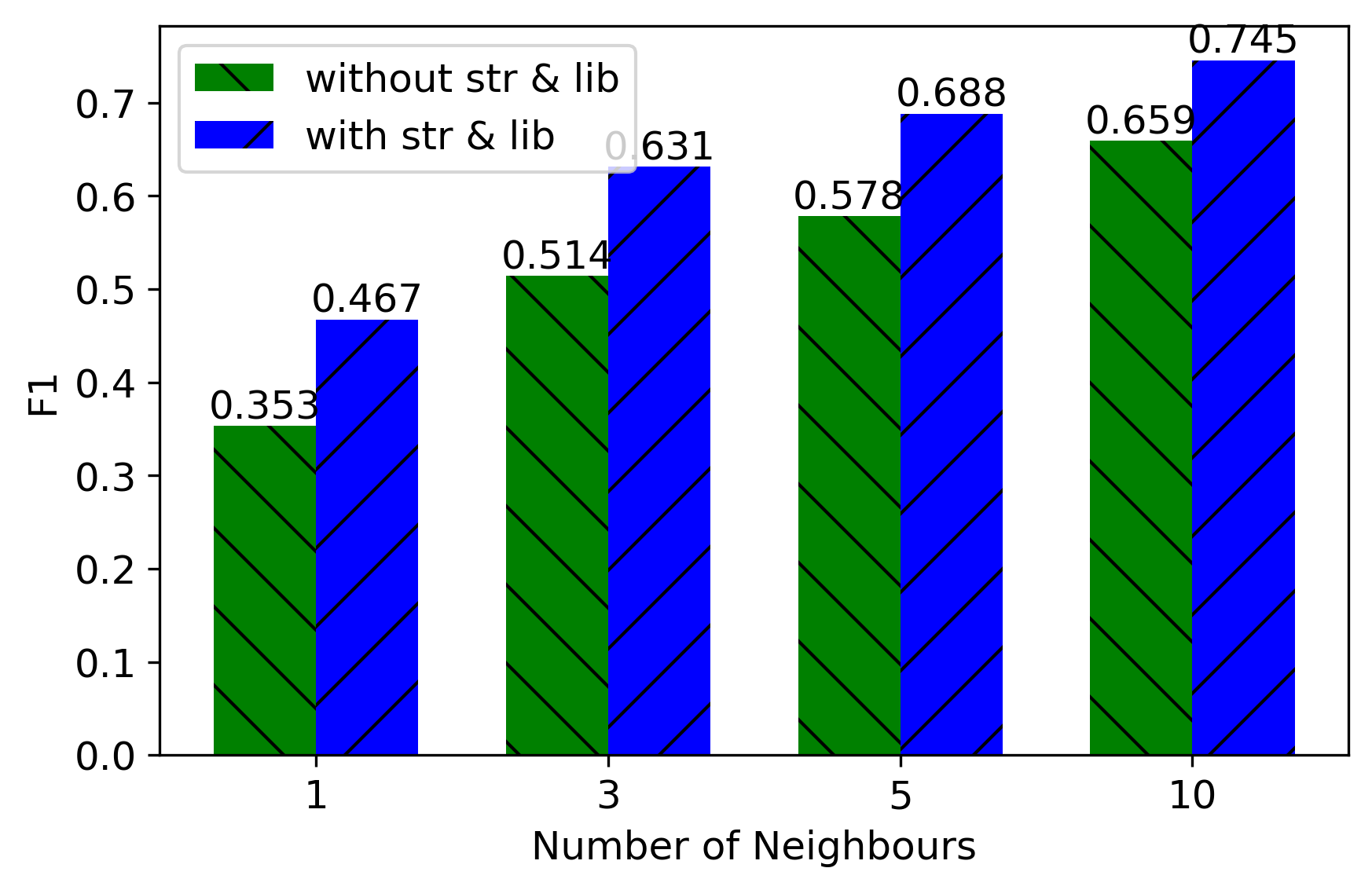}} \qquad
    \subfloat[Heatmap of Attention Scores]{{\includegraphics[width=0.3\linewidth]{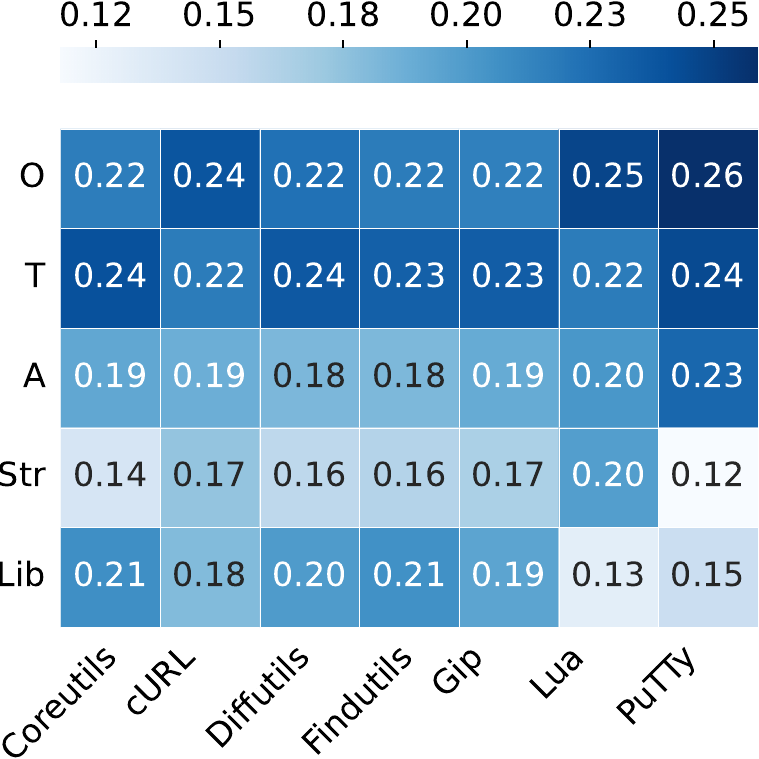}}}
    \Description{Impact of Random Walk Parameters and Effectiveness of Normalization}
    \caption{The contribution of different entities (see Equation~\ref{eq:entities}) to the F1 scores of \vexirtovec.}
    \label{fig:ablation2}
    \vspace{-\baselineskip}
\end{figure}

Data extracted from strings and library calls is important because it tends to remain unchanged, even in very aggressive adversarial settings.
Thus, unless the compiler determines that the relevant portion of the code is \textit{dead}, strings and library calls will persist in the different versions of the binary code.
However, strings and library calls alone cannot be used as code embedding for two reasons.
First, because many functions simply do not contain any data of this kind.
Second, different functions might still contain exactly the same data, such as calls to I/O operations, for instance.
Therefore, while strings and library calls cannot be the primary representation of a function, their inclusion alongside other features considerably improves the ability of \vexirtovec{} to match binary code. 

\subsection{RQ9: Attention Weights}
\label{sub:attention_eval}

As Section~\ref{sub:ean} explains, the final embedding that characterizes a function includes information taken from five vectors:
operands ($\embedEqn{O}$), types ($\embedEqn{T}$), arguments ($\embedEqn{A}$), strings  ($\embedEqn{S}$) and library calls ($\embedEqn{L}$).
Section~\ref{sub:eval_strings} provides some evidence that strings and library calls, although not essential, are important to ensure the high precision of \vexirtovec.
In this section, we analyze how the model perceives the other entities by examining the learned attention values.

Figure~\ref{fig:ablation2}(b) examines the learned attention values of the different vectors that constitute \vexirtovec.
The figure includes data taken from the entire test set described in Section~\ref{sec:experimental-setup}.
The heatmap reveals that \vexirtovec assigns varying levels of importance to different features across various datasets (Coreutils, cURL, Diffutils, Findutils, Gzip, Lua, and PuTTY).
Generally, Opcodes (O), Types (T), and Arguments (A) receive higher attention scores, highlighting their critical role in function representation.
However, Strings (Str) and External Library calls (Lib) might, sometimes, receive higher attention than Arguments (A).
This observation is especially true when they provide unique information.
For instance, if a particular library call is only invoked in a certain function $f$, this feature will be essential to distinguish $f$ from other routines.
Notice that the model does not give any attention weights to strings and library calls if they are not present in the function.

We observe that unique or less commonly used strings tend to receive higher attention scores.
For example, the string ``all'' which appears $64$ times in our entire dataset, receives an attention score of 0.3. In contrast, commonly used strings like ``\%s'', which appears $14K$ times, get a much lower attention score ($0.07$).
Similarly, specialized or domain-specific library functions, such as \texttt{pow} and \texttt{sin}, receive more attention weights than commonly used functions like \texttt{malloc} and \texttt{free}. Hence, on average, projects like Coreutils and PuTTY, which have many common strings, receive a lower attention weight for strings.

Figure~\ref{fig:ablation2}(b) shows that \vexirtovec learns to assign varying degrees of importance to different entities, aligning with their contextual contributions for binary similarity. 
This adaptive approach stands in contrast with works like IR2Vec~\cite{VenkataKeerthy-2020-IR2Vec}, which assigns fixed weights to Opcodes, Types, and Arguments heuristically, with a predetermined importance order of $O > T > A$.

%% file: 7.relatedWorks.tex
\section{Related Work}
\label{sec:relatedworks}

Given the practical importance of binary similarity, extensive research has been dedicated to studying this problem\footnote{As testimony to this variety, the ``{\it Awesome Binary Similarity}'' webpage catalogs about $214$ publications as on June 2024 - \url{https://github.com/SystemSecurityStorm/Awesome-Binary-Similarity}}.
This section covers different facets of this literature.
Notice that much of this literature is defined by the fact that
determining program equivalence is an undecidable problem~\cite{Rice53}; hence, solutions to binary similarity are based on heuristics.
Early implementations of such heuristics worked on source code~\cite{Hunt77}.
As a summary of this section, Figure~\ref{fig:all-configurations} positions our work within this literature.

\begin{figure}[ht]
    \vspace{-0.5\baselineskip}
    \centering
    \includegraphics[scale=0.48]{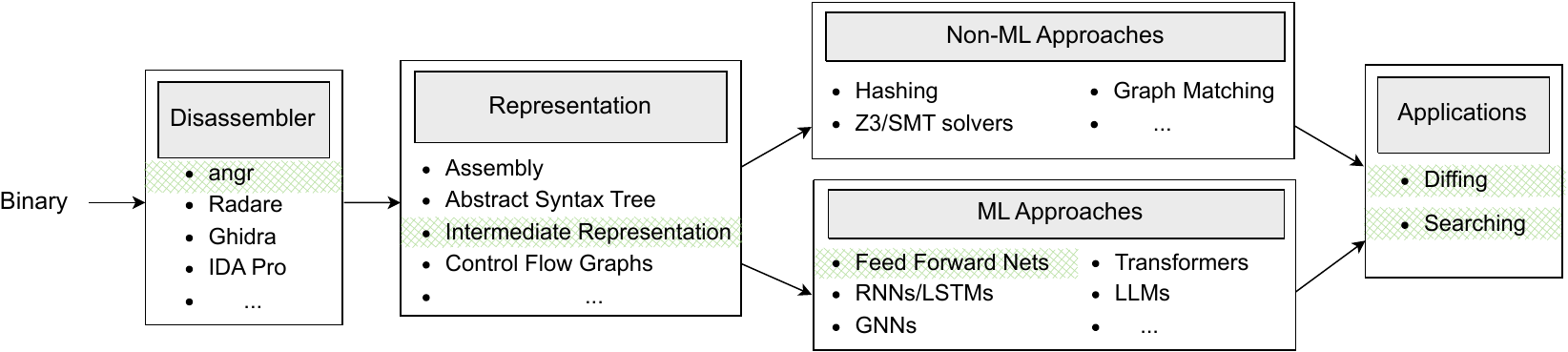}
    \caption{The different elements that constitute the Binary Similarity problem. The options used by \vexirtovec under each facet are highlighted.}
    \Description{The many parts involved in the implementation of binary similarity tools. We have highlighted the elements of
    this taxonomy that are present in the design of \vexirtovec.}
    \label{fig:all-configurations}
    \vspace{-\baselineskip}
\end{figure}

Binary code became a more intense focus of research in the nineties.
Those initial efforts were mostly centered on sequence alignment algorithms~\cite{Coppieters95,Baker98} and hash functions~\cite{Wang00}. Canonicalization techniques, like register renaming~\cite{Debray00}, were proposed around that time.
These techniques are still in use today~\cite{Collyer23,He24,massarelli2019safe,Wang22jTrans,li2021palmtree}.
However, nowadays, research on binary similarity is directed towards machine-learning approaches~\cite{ding2019asm2vec,massarelli2019safe,duan2020deepbindiff,li2021palmtree,wang2023sem2vec,Yu2020OrderMatters}.

\subsection{Binary Similarity}
\label{sub:xisa_rw}

In Table~\ref{tab:survey}, we provide a detailed comparison of our approach with other related approaches on various design choices and applications, some of which we explain below.
To provide the reader with some perspective on how different solutions to adversarial binary similarity compare, the column ``Works Compared'' in Table~\ref{tab:survey} shows which tools have been used as baselines during the evaluation of different approaches.

\input{tables/relatedworks-table}

\paragraph{Disassemblers and Input Representations}
Binary similarity analysis begins by disassembling the input binary using disassemblers and binary analysis tools~\cite{wang2017angr, IDAPro, Ghidra, Radare, brumley2011BAP}. 
These tools translate the binary code into human-readable formats like assembly code, Abstract Syntax Tree (AST), or Intermediate Representations (IR).
Different approaches rely on different representations; however, a vast majority of binary similarity tools use assembly code~\cite{Qasem2023Binfinder-AsiaCCS, wang2023sem2vec, zhu2023ktrans, Wang22jTrans, ahn2022BinShot, li2021palmtree, jia2021codee, duan2020deepbindiff, pei2020trex, Yu2020OrderMatters, massarelli2019safe, ding2019asm2vec, zuo2018innereye, xu2017gemini}.
Works like Asteria~\cite{yang2021asteria} and Asteria-Pro~\cite{yang2023asteriapro} use AST, whereas \citet{yaniv2017pldi} and \citet{peng2021oscar} work on the LLVM IR.
This paper uses the \vexir representation generated by \angr to model binary similarity.
We chose \vexir because it is architecture-neutral and open-source.
\vexirtovec{} is heavily engineered to work on \vexir; however, its ideas could be employed on other representations.
Section~\ref{sub:eval_norm_baselines} supports this statement, demonstrating that the peepholes and normalizing transformations of Section~\ref{sub:normalization} could be used to enhance different binary similarity tools.

\paragraph{Compilation Configurations}
Adversarial binary similarity explores the challenges of identifying differences in binaries resulting from various compilation configurations, including cross-compiler, cross-optimization, and cross-architecture settings. These variations pose significant challenges to the binary classifier.
Consequently, they have received considerable research attention. However, not all existing implementations of binary similarity tools natively support all three scenarios.

Prior works by \citet{pewny2015sp} and \citet{Redmond18} initiated research on cross-architecture similarity. Their approaches, however, often suffer from limitations such as slow execution speeds~\cite{eschweiler2016discovre,qian2016bugsearch}. Existing works that rely on assembly codes for modeling binary similarity~\cite{duan2020deepbindiff, wang2023sem2vec, li2021palmtree, ahn2022BinShot} are architecture-specific and cannot perform cross-architecture binary similarity analysis. This limitation is often overcome by the usage of feature-based representations~\cite{Qasem2023Binfinder-AsiaCCS, kim23binaryfeatures, xu2017gemini, eschweiler2016discovre, chandramohan2016bingo}, with features extracted using binary analysis tools. Alternatively, unified representation spaces~\cite{pei2020trex,Kim22,zhang2020similarity} are created by training models on assembly codes from multiple architectures; however, they do not generalize to unseen architectures.

\vexirtovec overcomes these limitations by leveraging \vexir, the Intermediate Representation (IR) used by analysis tools like \angr and Valgrind. Vexir offers a more architecture-neutral view of the program compared to assembly code. Consequently, the implementation of \vexirtovec{} evaluated in this paper is effective in different adversarial settings, including the cross-compiler, the cross-optimization, and the cross-architecture scenarios.

\paragraph{Approach}
Early methods for comparing binaries relied on statistical analysis and static code inspection.
 Tools like BinDiff~\cite{bindiff} utilize graph isomorphism algorithms to identify code similarities.
Other approaches employ specialized hashing techniques~\cite{funcsimsearch, yaniv2017pldi} or solvers like Z3~\cite{pewny2015sp} and SMT~\cite{yaniv2016pldi} to perform more intricate comparisons.

More recent binary similarity techniques have seen the introduction of Machine Learning (ML) for generating n-dimensional vector representations of binaries.
 Asm2Vec~\cite{ding2019asm2vec} and SAFE~\cite{massarelli2019safe} adapt the word2vec model~\cite{mikolov2013word2vec} for this purpose. However, the majority of current research leverages transformers~\cite{Vaswani2017attentionTransformers} such as BERT~\cite{devlin2019bert} and RoBERTa~\cite{Liu2019Roberta}. These methods often model Control Flow Graphs (CFGs)~\cite{Luo23, Yu2020OrderMatters, qian2016bugsearch} using Graph Neural Networks (GNNs)~\cite{GNNBook2022}. Works like sem2vec~\cite{wang2023sem2vec} even combine symbolic execution powered by Z3, GNNs, and transformers for a more comprehensive representation learning process.

ML-based approaches can represent binaries either as feature vectors or by learning distributed vector representations.
 Tools like BinFinder~\cite{Qasem2023Binfinder-AsiaCCS} rely on predefined features like caller/callee information, status flags, and library calls. In contrast, Asm2Vec\cite{ding2019asm2vec} and Codee~\cite{jia2021codee} learn the input assembly representation by employing variations of the word2vec model~\cite{mikolov2013word2vec}.
\vexirtovec, in turn, leverages knowledge graph embeddings~\cite{transe-Bordes:2013:TEM:2999792.2999923} to model binary functions as distributed vectors.

\paragraph{Out-of-vocabulary Issue}
One of the main challenges in existing works is the out-of-vocabulary (OOV) problem. Approaches that learn instruction-level or basic block-level embeddings are highly prone to this issue~\cite{massarelli2019safe, zuo2018innereye, Redmond18}. In our experiments, we observed that SAFE encounters a very high number of OOV tokens, resulting in significant performance degradation. Works that use token-level embeddings generally face fewer OOV problems~\cite{ding2019asm2vec, duan2020deepbindiff}. However, solely encoding tokens can overlook the semantics of the instruction.

Canonicalization and normalization techniques can help mitigate OOV issues to some extent. Nevertheless, OOV problems can still arise if numeric and string constants are not properly handled. Works that consider such constants as tokens within their vocabulary often suffer from OOV issues~\cite{wang2023sem2vec, li2021palmtree}. Improper canonicalization can lead to a very large vocabulary, increasing the likelihood of OOV problems~\cite{zhu2023ktrans}. Additionally, works that model instructions and basic blocks as tokens and use LSTMs and transformers can suffer from OOV issues due to their fixed-length input requirements.

\vexirtovec mitigates the OOV issue by learning entity-level representations of instructions with carefully designed canonicalization, resulting in a smaller vocabulary. During fine-tuning, instruction-level semantics are effectively captured, addressing the shortcomings of previous methods.

\paragraph{Resources}
Several existing binary similarity approaches face limitations when dealing with large binaries containing thousands of basic blocks~\cite{haq2021binsurvey}.
These limitations stem from the time and resources required for training or deployment.
Techniques that leverage modern language models~\cite{ahn2022BinShot, wang2023sem2vec, zhu2023ktrans, Luo23, li2021palmtree,Wang22jTrans,peng2021oscar,Yu2020OrderMatters,pei2020trex} often require training millions of parameters.
This requirement leads to lengthy training times and necessitates significant computational resources.
For example, many of these approaches rely on large clusters equipped with multiple high-end GPUs, such as A100 and V100~\cite{Qasem2023Binfinder-AsiaCCS,zhu2023ktrans,Wang22jTrans,ahn2022BinShot,peng2021oscar,pei2020trex}.
Still, training can take weeks.
These demands translate to substantial resource and time constraints during deployment, hindering practical usability.

In contrast, \vexirtovec{} relies on lightweight feed-forward neural networks to model binary similarity.
This low-computational requirement enables our model to train effectively on a workstation equipped with a single RTX3080 GPU with 10 GB of memory.
As seen in Section~\ref{sub:scalability}, each training epoch takes 5-8 seconds, hence being orders of magnitude faster than other tools.

\paragraph{Application}
Existing binary similarity approaches can be broadly categorized into two groups:
those that learn a similarity score and those that learn a distance measure.
The first category focuses on predicting whether a pair of binary snippets are similar or dissimilar~\cite{duan2020deepbindiff, zuo2018innereye, zeek2018plas}.
The typical application of these methods is binary diffing, not searching.
They become impractical for searching due to their reliance on pairwise comparisons, resulting in a quadratic time complexity, as we have already explained in Section~\ref{sec:introduction}.

In contrast, approaches that learn a distance measure~\cite{massarelli2019safe, Qasem2023Binfinder-AsiaCCS} are applicable to both searching and diffing tasks.
\vexirtovec{} falls into this category, leveraging the Siamese networks of Section~\ref{sub:vexnet} to learn a distance measure between programs.
This design makes \vexirtovec{} suitable for both diffing and searching.

\subsection{Our Work in Perspective}
\label{sub:perspective}

\vexirtovec involves a two-step learning process - pre-training a vocabulary to derive the final representation of the function as a series of lookups, followed by fine-tuning with \vexnet to learn similarity using global attention.
This way of learning vocabulary representation draws its inspiration from our earlier work: IR2Vec~\cite{VenkataKeerthy-2020-IR2Vec}, which uses representation learning for obtaining program embeddings.
However, there are several key differences between the two methods:

\begin{description}
    \item[Embeddings:] IR2Vec models function embeddings through a heuristic-based accumulation of the entities of the LLVM-IR instructions, whereas \vexirtovec uses global attention to learn to combine the different entities of the \vexir instructions and metadata. Section~\ref{sub:attention_eval} shows that using the attention mechanism enhances the effectiveness of our proposed approach.
    \item[Sampling:] \vexirtovec is extracted from our notion of peepholes (Definition~\ref{def:peephole}). The same program instruction might appear in multiple peepholes or multiple times within the same peephole. As seen in Section~\ref{sub:random_walk}, the length of peepholes and the minimum number of times each basic block contributes to them are configurable parameters of \vexirtovec. 
    In contrast, IR2Vec is extracted from LLVM IR, where each IR instruction is visited exactly once.
    
    \item[Normalization:] The construction of \vexirtovec{} involves normalizing the \vexir using unsound architecture-independent optimizations implemented as \optengine, as explained in Section~\ref{sub:normalization}. In contrast, IR2Vec does not use any such form of normalization.
\end{description}

\paragraph{Canonicalization and Normalization}
As explained in Section~\ref{sub:canonicalization}, to construct the \vexirtovec{} embedding of a program,  we simplify its intermediate representation, replacing concrete syntax (constants, bitwidths, types, etc) with abstract placeholders.
Previous works on binary similarity have used similar forms of canonicalization.
The most common techniques consist in masking out constants, pointers, and registers~\cite{ding2019asm2vec,duan2020deepbindiff,zuo2018innereye, massarelli2019safe, farhadi2014binclone,wang2023sem2vec,li2021palmtree}. 

Optimizations, similar in purpose as the normalization step of Section~\ref{sub:normalization}, have also been used in previous works.
For instance, VulHawk~\cite{Luo23} uses \textit{def-use} information to prune off redundant and unused instructions in IDA-Pro Microcode~\cite{IDAPro}.
As another example, \citet{yaniv2017pldi} lift the binaries to LLVM IR~\cite{Lattner:2004:llvm} and decompose the functions into \textit{strands}, which are then subject to the off-the-shelf optimizations of LLVM.
These optimized \textit{strands} are then used for computing similarity through statistical methods.
Their notion of strands corresponds to the \textit{program slices} proposed by Weiser~\cite{weiser1984programSlicing}.
In this regard, \citeauthor{yaniv2017pldi}'s work differs from ours in two aspects: first, their optimizations are sound; second, they are not restricted to straight-line code sequences.
However, the differences between our work and \citeauthor{yaniv2017pldi} go much beyond optimizations: they extract embeddings from single program slices, whereas we extract them from peepholes, which might include multiple occurrences of the same instruction.
In terms of software engineering, we notice that extracting executable program slices is much more complex than extracting peepholes.
Due to this complexity and a lack of a public artifact, we could not compare \citeauthor{yaniv2017pldi}'s work and ours.

%% file: tables/relatedworks-table.tex
\begin{table}[ht]
    \centering
    \caption{Comparison of \vexirtovec with earlier works. XO, XC, and XA indicate cross-optimization, cross-compiler, and cross-architecture support. 
    D and S indicate support for Diffing and Searching tasks.}
    \label{tab:survey}
   \resizebox{\textwidth}{!}{%
    \begin{tblr}{
        colspec={llllccccllcc},
        cell{1}{2-5} = {r=2}{c, font=\bfseries},
        cell{1}{6} = {c=3}{c, font=\bfseries},
        cell{1}{9} = {c=2}{c, font=\bfseries},
        cell{1}{11} = {c=2}{c, font=\bfseries},
        cell{2}{9-10} = {}{c},
        row{odd}={bg=lightgray}, 
        row{1}={bg=white}, 
        column{1}={bg=white, font=\bfseries},
        row{2}={font=\bfseries},
    }
    \toprule
    & Name & Works Compared & Disassembler & {Input \\Representation} & Comp Conf. & & & Approach &  & Tasks & \\
    \cmidrule[lr=-0.5]{6-8}
    \cmidrule[lr=-0.5]{9-10}
    \cmidrule[lr=-0.5]{11-12}
    & & & & & XO & XC & XA & Models used & GPUs & D & S \\
    \midrule
    $\Rightarrow$ & \vexirtovec & \cite{Qasem2023Binfinder-AsiaCCS, Damasio23, massarelli2019safe,duan2020deepbindiff, bindiff} & angr & VEX IR & Y & Y & Y & FCNN & 1 RTX 3080 & Y & Y \\
    \cmidrule[gray]{2-13}\morecmidrules\cmidrule[gray]{2-13}
    \SetCell[r=5]{c}{\rotatebox{90}{2023}} & BinFinder~\cite{Qasem2023Binfinder-AsiaCCS} & \cite{massarelli2019safe,zeek2018plas,ding2019asm2vec,Li19GraphMatchingNetworks} & angr & Assembly & Y & Y & Y & FCNN & 8 TITAN & Y & Y \\
    & Sem2Vec~\cite{wang2023sem2vec} & \cite{bindiff,ding2019asm2vec,massarelli2019safe,Yu2020OrderMatters,li2021palmtree} & angr & Assembly & Y & Y & & {RoBERTa,\\ LSTM, GNN} & 1 Tesla V100 32G & Y & Y \\
    & kTrans~\cite{zhu2023ktrans} & \cite{xu2017gemini,Wang22jTrans,li2021palmtree} & IDA-Pro & Assembly & Y & & & BERT & 4 V100 & & Y \\
    & VulHawk~\cite{Luo23} & \cite{ding2019asm2vec, yang2021asteria, li2021palmtree, Li19GraphMatchingNetworks, massarelli2019safe, pei2020trex} & IDA-Pro & MicroCode & Y & Y & Y & {RoBERTa,\\ GCN, ResNet} & 1 RTX 3090 & Y & Y \\
    & Asteria-pro~\cite{yang2023asteriapro} & \cite{xu2017gemini,massarelli2019safe,pei2020trex} & IDA-Pro & AST & & Y & Y & Tree LSTM & & & Y \\
    \cmidrule[gray]{2-13}
    \SetCell[r=2]{c}{\rotatebox{90}{2022}} & jTrans~\cite{Wang22jTrans} & \cite{Feng2016Genius,xu2017gemini,massarelli2019safe,ding2019asm2vec,Yu2020OrderMatters} & IDA-Pro & Assembly & Y & & & BERT & 8 A100 & Y & \\
    & BinShot~\cite{ahn2022BinShot} & \cite{xu2017gemini,ding2019asm2vec,li2021palmtree} & IDA-Pro & Assembly & Y & Y & & BERT & 2 RTX A6000 & Y & \\
    \cmidrule[gray]{2-13}
    \SetCell[r=4]{c}{\rotatebox{90}{2021}} & Oscar~\cite{peng2021oscar} & \cite{bindiff,ding2019asm2vec,Yu2020OrderMatters} & RetDec & LLVM IR & Y & & & RoBERTa & 8 Tesla V100 & Y & \\
    & PalmTree~\cite{li2021palmtree} & \cite{xu2017gemini,Guo19DeepVSA,chua17eklavya} & BinaryNinja & Assembly & Y & Y & & BERT & 1 GTX 2080Ti & & Y \\
    & Asteria~\cite{yang2021asteria} & \cite{xu2017gemini,Feng2016Genius,eschweiler2016discovre} & IDA-Pro & AST & & & Y & Tree LSTM & & Y & \\
    & Codee~\cite{jia2021codee} & \cite{ding2019asm2vec,xu2017gemini,massarelli2019safe,duan2020deepbindiff,Yu2020OrderMatters} & IDA-Pro, angr & Assembly & Y & Y & Y & Word2Vec & 2 GTX RTX5000 & & Y \\
    \cmidrule[gray]{2-13}
    \SetCell[r=3]{c}{\rotatebox{90}{2020}} & DeepBinDiff~\cite{duan2020deepbindiff} & \cite{ding2019asm2vec,bindiff} & angr & Assembly & Y & & & DeepWalk & Does not use & Y & \\
    & Trex~\cite{pei2020trex} & \cite{ding2019asm2vec,massarelli2019safe,xu2017gemini} &  & Assembly & Y & Y & Y & {Hierarchical\\ Transformers} & 8 RTX 2080-Ti & & Y \\
    & Order Matters~\cite{Yu2020OrderMatters} & \cite{xu2017gemini} &  & Assembly & & & Y & {BERT, ResNet,\\ GRU} & & & Y \\
    \cmidrule[gray]{2-13}
    \SetCell[r=3]{c}{\rotatebox{90}{2019}} & SAFE~\cite{massarelli2019safe} & \cite{xu2017gemini} & Radare2 & Assembly & Y & Y & Y & RNNs & K80 & & Y \\
    & Asm2vec~\cite{ding2019asm2vec} & \cite{chandramohan2016bingo,funcsimsearch} & IDA-Pro & Assembly & Y & Y & & Word2Vec & & & Y \\
    & Innereye~\cite{zuo2018innereye} & -- & BAP & Assembly & Y & & Y & FCNN, LSTM & Does not use & & Y \\
    \cmidrule[gray]{2-13}
    \SetCell[r=2]{c}{\rotatebox{90}{2017}} & Gemini~\cite{xu2017gemini} & ~\cite{Feng2016Genius} & IDA-Pro & Assembly & & & & structure2vec & 1 GTX 1080 & Y & \\
    & \citet{yaniv2017pldi} & -- & angr & LLVM IR & Y & Y & Y & Non-ML & & & Y \\
    \cmidrule[gray]{2-13}
     & BinDiff~\cite{bindiff} & -- & Ghidra & -- & Y & Y & Y & Non-ML & & Y & \\
    \bottomrule
    \end{tblr}
    }
\end{table}

%% file: 8.conclusions.tex
\section{Conclusions and Future Works}
\label{sec:conclusion}

We introduce \vexirtovec{}, a \vexir-based embedding framework designed to solve binary similarity tasks. 
\vexirtovec{} is architecture-neutral and can be tuned to solve different binary similarity problems, such as diffing or searching.
The process of constructing this embedding follows a three-phase approach.
In the first phase, we decompose the functions into peepholes derived from their CFG and normalize them by using our \optenginelong (\optengine).
In the second phase, we follow an unsupervised pre-training to learn the vocabulary of \vexir so as to represent the peepholes.
In the third phase, we train \vexnet, a simple feed-forward Siamese network, to learn the similarity metric.
Pre-training and fine-tuning are done once, independent of the binary similarity task.
We demonstrate that \vexirtovec{} is more scalable and more accurate than previous works, avoiding out-of-vocabulary issues that are common in tools that deal with the binary similarity problem.

We implement the core ideas of \vexirtovec{} in a specific setting: \vexir. However, we believe these ideas can be applied to other program representations as well. For example, 
the concept of peepholes and normalization can be used to enhance other feature-based or embedding-based representations.
We demonstrate how they can be adapted to work with previous approaches like BinFinder and histograms of opcodes.
Applying these ideas to settings beyond \vexir is a future research direction we hope to explore. Moreover, we plan to extend \vexirtovec{} to include dynamic profile information for increased precision.
We also plan to perform binary-source code matching that can be helpful in identifying if a vulnerable source code snippet is present in the given binary.
We plan to open-source its code and relevant datasets in the near future.

%% file: 9.appendix.tex
\section{Appendix}
\label{sec:appendix}

This section contains additional tables and data that can be of interest to readers interested in reproducing the results in this paper but that are not essential to understanding its core ideas.

\subsection{Description of Dataset}
~\label{appendix:dataset-desc}
Details and statistics of our dataset for our binary similarity experiments are shown in Table ~\ref{tab:dataset}. 
As it can be seen, the number of functions differs across the architectures due to the impacts of optimizations and the differences in compilers.
Coreutils and PuTTY form a significant portion of the dataset. The binary sizes vary from about $9 KB$--$7.6MB$.
\input{tables/Dataset_description}

\subsection{Peepholes}
\label{sec:graph-study}

In Table~\ref{tab:peepholes}, we show the average number of basic blocks (vertices), the average number of edges, and the average length of the longest straight-line path in the Control-Flow Graph.
These values are averaged across all functions of all binaries in our dataset (Section~\vref{sec:experimental-setup}).
Additionally, we also show the average number of peepholes obtained for different values of $k$ (with $c$ set to $2$).

\input{tables/Peepholes_k-variation}

There is a decreasing trend in the number of peepholes with the increase in values of $k$. However, $k$ being the maximum number of basic blocks in a peephole, there is little to no difference in the number of peepholes generated when $k$ matches the number of basic blocks. (Notice the entries corresponding to $k=36$ to $k=144$.)

In practice, the number of peepholes generated is typically less than the worst-case defined by Theorem~\vref{theo:peephole}, and tends to a value close to $c|V|/2$ (especially when $k >1$).  Beyond a threshold, $k$ does not appear to impact the number of peepholes generated.

\subsection{Cross-Optimization Binary Diffing}
This section provides additional results on the cross-optimization binary diffing experiment described in Section~\vref{sec:cross-opt-diffing}.

\input{tables/ablation-exp1}

In Table~\ref{tab:crossopt-O1-O3}, we show the detailed result of Cross-Optimization binary diffing involving a comparison between O1 and O3 optimization levels. 
As it can be observed, \vexirtovec achieves the highest precision and recall scores across all configurations.

\subsection{Vocabulary - Analogies}
\label{appendix:clusters-analogies}

In Table~\ref{tab:analogies}, we show the complete list of analogies that we use for evaluating the vocabulary $\mathcal{V}_{lookup}$ in Section~\vref{sec:vocab-eval}.

We are able to correctly answer both correct syntactic and semantic analogies.  Intrinsic syntactic analogies are like add : addf :: sub : subf, GetI: PutI :: Get : Put, while semantic analogies like shl : mul :: shr : divmod and or : and :: shr : shl. 

\input{tables/analogies-full}

%% file: tables/Dataset_description.tex
\begin{table}[ht]
\centering
\caption{Description of dataset}
\label{tab:dataset}
\resizebox{0.8\columnwidth}{!}{%
\begin{tabular}{llrrrrc}
\toprule
\multicolumn{1}{c}{\textbf{Projects}} & \multicolumn{1}{c}{\textbf{Arch}} & \multicolumn{1}{c}{\textbf{\#Functions}} & \textbf{\makecell{\#Functions\\ in Test set}} & \textbf{GT pairs} & \textbf{\makecell{\#Binaries}} & \multicolumn{1}{c}{\textbf{Binary sizes}}\\
\midrule
\multirow{2}{*}{Coreutils} & ARM & 751K & 128K & \multirow{2}{*}{98K} & \multirow{2}{*}{13,320} & \multirow{2}{*}{9.7 KB - 1.3 MB} \\
 & x86 & 843K & 113K &  &  &  \\
 \midrule
\multirow{2}{*}{Diffutils} & ARM & 44K & 7K & \multirow{2}{*}{5K} & \multirow{2}{*}{702} & \multirow{2}{*}{9.7 KB - 1.5 MB} \\
 & x86 & 46K & 6K &  &  &  \\
 \midrule
\multirow{2}{*}{Findutils} & ARM & 61K & 5K & \multirow{2}{*}{4K} & \multirow{2}{*}{690} & \multirow{2}{*}{9.8 KB - 2.2 MB} \\
 & x86 & 71K & 4K &  &  &  \\
 \midrule
\multirow{2}{*}{cURL} & ARM & 9K & 9K & \multirow{2}{*}{5K} & \multirow{2}{*}{60} & \multirow{2}{*}{762.8 KB - 1.1 MB} \\
 & x86 & 5K & 5K &  &  &  \\
 \midrule
\multirow{2}{*}{Lua} & ARM & 40K & 40K & \multirow{2}{*}{23K} & \multirow{2}{*}{120} & \multirow{2}{*}{554.3 KB - 2.3 MB} \\
 & x86 & 41K & 41K &  &  &  \\
 \midrule
\multirow{2}{*}{PuTTY} & ARM & 436K & 436K & \multirow{2}{*}{325K} & \multirow{2}{*}{780} & \multirow{2}{*}{9.2 KB - 7.6 MB} \\
 & x86 & 428K & 428K &  &  &  \\
 \midrule
\multirow{2}{*}{Gzip} & ARM & 4K & 4K & \multirow{2}{*}{3K} & \multirow{2}{*}{60} & \multirow{2}{*}{205 KB - 547.1 KB} \\
 & x86 & 4K & 4K &  &  &  \\
 \bottomrule
\end{tabular}%
}
\end{table}

%% file: tables/Peepholes_k-variation.tex
\begin{table*}[ht]
\centering
\caption{Peepholes Trend with varying k}
\label{tab:peepholes}
\resizebox{\columnwidth}{!}{%
\begin{tabular}{l|r|r|r|rrrrrrrrrrr}
\toprule
\multicolumn{1}{c}{\textbf{Compiler}} &    \multicolumn{1}{c}{\textbf{Avg.}} &     \multicolumn{1}{c}{\textbf{Avg.}} & \multicolumn{1}{c}{\textbf{Avg. Longest}} & \multicolumn{11}{c}{\textbf{Average number of Peepholes}} \\
{} & \textbf{block\_count} & \textbf{edge\_count} & \textbf{Path in DAG}     &                          \textbf{k=1} &    \textbf{k=3} &    \textbf{k=5} &    \textbf{k=7} &   \textbf{k=10} &   \textbf{k=12} &   \textbf{k=36} &   \textbf{k=64} &   \textbf{k=72} &  \textbf{k=100} &  \textbf{k=144} \\
\midrule
x86-clang-8-O0 &              22.23 &              28.80 &                    7.57 &                       44.45 & 26.42 & 22.98 & 21.65 & 20.94 & 20.73 & 20.10 & 20.05 & 20.03 & 20.92 & 20.06 \\
x86-clang-8-O1 &              14.82 &              18.91 &                    6.42 &                       29.63 & 18.59 & 16.65 & 15.99 & 15.66 & 15.52 & 15.36 & 15.43 & 15.41 & 15.72 & 15.42 \\
x86-clang-8-O2 &              25.92 &              35.39 &                    9.50 &                       51.84 & 31.36 & 27.34 & 25.96 & 24.96 & 24.75 & 24.25 & 24.20 & 24.25 & 25.10 & 24.29 \\
x86-clang-8-O3 &              27.54 &              38.17 &                    9.47 &                       55.08 & 33.38 & 29.24 & 27.51 & 26.73 & 26.44 & 25.94 & 25.93 & 25.98 & 26.82 & 26.00 \\
x86-clang-8-Os &              20.16 &              26.68 &                    7.64 &                       40.32 & 24.64 & 21.72 & 20.75 & 20.18 & 20.04 & 19.65 & 19.73 & 19.76 & 20.26 & 19.84 \\
x86-gcc-8-O0   &              16.78 &              21.31 &                    6.82 &                       33.56 & 20.82 & 18.37 & 17.53 & 17.04 & 16.92 & 16.64 & 16.66 & 16.65 & 17.12 & 16.71 \\
x86-gcc-8-O1   &              21.51 &              28.34 &                    7.74 &                       43.03 & 26.55 & 23.51 & 22.42 & 21.76 & 21.48 & 21.27 & 21.22 & 21.15 & 21.83 & 21.30 \\
x86-gcc-8-O2   &              20.63 &              27.19 &                    7.75 &                       41.25 & 25.36 & 22.20 & 21.10 & 20.51 & 20.24 & 20.12 & 20.02 & 20.03 & 20.60 & 19.99 \\
x86-gcc-8-O3   &              27.57 &              37.33 &                   10.24 &                       55.13 & 33.51 & 29.13 & 27.50 & 26.40 & 26.22 & 25.71 & 25.67 & 25.72 & 26.62 & 25.59 \\
x86-gcc-8-Os   &              19.21 &              25.16 &                    7.10 &                       38.42 & 23.45 & 20.51 & 19.58 & 18.83 & 18.68 & 18.44 & 18.56 & 18.50 & 19.01 & 18.48 \\
\bottomrule
\end{tabular}%
}
\end{table*}

%% file: tables/ablation-exp1.tex
\begin{table*}[ht]
\centering
\caption{Cross-Optimization Binary Diffing - O1 Vs. O3}
\label{tab:crossopt-O1-O3}
\resizebox{0.9\textwidth}{!}{
\begin{tabular}{l|cccccc|cccccc}
\toprule
                                               & \multicolumn{6}{|c}{\textbf{Precision}} & \multicolumn{6}{|c}{\textbf{Recall}} \\
                                               &   \textbf{BinDiff} & \textbf{DBD} & \textbf{SAFE} & \textbf{OPC} & \textbf{BinFinder} & \textsc{\textbf{VexIR2Vec}} & \textbf{BinDiff} &   \textbf{DBD} & \textbf{SAFE} & \textbf{OPC} & \textbf{BinFinder} & \textsc{\textbf{VexIR2Vec}} \\
\midrule
& \multicolumn{12}{c}{\textbf{\underline{ARM - Clang12}}} \\
\textbf{Coreutils} &               0.32 &         0.46 &          0.28 &              0.47 &               0.42 &      \textbf{0.57} &             0.20 &           0.58 &          0.22 &              0.62 &               0.56 &      \textbf{0.76} \\
                              \textbf{cURL} &               0.44 &              &          0.51 &              0.71 &               0.68 &      \textbf{0.82} &             0.42 &                &          0.34 &              0.82 &               0.78 &      \textbf{0.94} \\
                              \textbf{Diffutils} &               0.51 &         0.45 &          0.40 &              0.62 &               0.55 &      \textbf{0.71} &             0.37 &           0.60 &          0.29 &              0.80 &               0.71 &      \textbf{0.93} \\
                              \textbf{Findutils} &               0.51 &         0.52 &          0.38 &              0.54 &               0.51 &      \textbf{0.68} &             0.30 &           0.68 &          0.31 &              0.67 &               0.63 &      \textbf{0.83} \\
                              \textbf{Gzip} &               0.44 &         0.34 &          0.43 &              0.53 &               0.50 &      \textbf{0.65} &             0.40 &           0.77 &          0.30 &              0.68 &               0.64 &      \textbf{0.83} \\
                              \textbf{Lua} &               0.34 &              &          0.31 &              0.24 &               0.27 &      \textbf{0.47} &             0.33 &                &          0.20 &              0.38 &               0.43 &      \textbf{0.74} \\
                              \textbf{PuTTY} &               0.29 &              &          0.24 &              0.28 &               0.29 &      \textbf{0.41} &             0.21 &                &          0.24 &              0.38 &               0.39 &      \textbf{0.54} \\
\cline{1-13}
& \multicolumn{12}{c}{\textbf{\underline{ARM - GCC8}}} \\
              \textbf{Coreutils} &               0.32 &         0.40 &          0.34 &              0.60 &               0.56 &      \textbf{0.67} &             0.22 &           0.62 &          0.27 &              0.69 &               0.65 &      \textbf{0.77} \\
                                \textbf{cURL} &               0.39 &              &          0.46 &              0.84 &               0.71 &      \textbf{0.91} &             0.31 &                &          0.35 &              0.89 &               0.77 &      \textbf{0.97} \\
                                \textbf{Diffutils} &               0.44 &         0.48 &          0.55 &              0.74 &               0.73 &      \textbf{0.83} &             0.35 &           0.81 &          0.39 &              0.84 &               0.83 &      \textbf{0.94} \\
                                \textbf{Findutils} &               0.41 &         0.43 &          0.43 &              0.63 &               0.61 &      \textbf{0.76} &             0.28 &           0.68 &          0.36 &              0.69 &               0.67 &      \textbf{0.83} \\
                                \textbf{Gzip} &               0.36 &         0.44 &          0.51 &              0.56 &               0.73 &       \textbf{0.8} &             0.31 &           0.77 &          0.36 &              0.64 &               0.84 &      \textbf{0.91} \\
                                \textbf{Lua} &               0.31 &              &          0.29 &              0.35 &               0.41 &      \textbf{0.59} &             0.30 &                &          0.22 &              0.44 &               0.51 &      \textbf{0.74} \\
                                \textbf{PuTTY} &               0.28 &              &          0.23 &              0.31 &          0.45          &      \textbf{0.54} &             0.20 &                &          0.26 &              0.35 &        0.43            &       \textbf{0.6} \\
\midrule
& \multicolumn{12}{c}{\textbf{\underline{x86 - Clang12}}} \\
 \textbf{Coreutils} &               0.40 &         0.21 &          0.32 &              0.46 &               0.39 &      \textbf{0.59} &             0.36 &           0.67 &          0.25 &              0.61 &               0.52 &      \textbf{0.79} \\
                                \textbf{cURL} &               0.60 &              &          0.57 &              0.73 &               0.69 &      \textbf{0.83} &             0.60 &                &          0.41 &              0.84 &               0.79 &      \textbf{0.95} \\
                                \textbf{Diffutils} &               0.62 &         0.25 &          0.46 &              0.63 &               0.51 &       \textbf{0.7} &             0.62 &           0.83 &          0.33 &              0.82 &               0.67 &      \textbf{0.91} \\
                                \textbf{Findutils} &               0.52 &         0.24 &          0.43 &              0.53 &               0.51 &      \textbf{0.69} &             0.50 &           0.74 &          0.35 &              0.65 &               0.64 &      \textbf{0.85} \\
                                \textbf{Gzip} &               0.40 &         0.17 &          0.52 &              0.47 &               0.48 &      \textbf{0.68} &             0.38 &  \textbf{0.93} &          0.37 &              0.60 &               0.62 &               0.88 \\
                                \textbf{Lua} &               0.40 &              &          0.34 &              0.24 &               0.23 &      \textbf{0.48} &             0.39 &                &          0.24 &              0.38 &               0.36 &      \textbf{0.74} \\
                                \textbf{PuTTY} &               0.34 &              &          0.27 &              0.23 &               0.27 &      \textbf{0.41} &             0.33 &                &          0.27 &              0.30 &               0.35 &      \textbf{0.53} \\

\cline{1-13}
& \multicolumn{12}{c}{\textbf{\underline{x86 - GCC8}}} \\
               \textbf{Coreutils} &               0.46 &         0.30 &          0.33 &              0.58 &               0.52 &      \textbf{0.67} &             0.43 &           0.62 &          0.28 &              0.66 &               0.59 &      \textbf{0.76} \\
                                \textbf{cURL} &               0.48 &              &          0.56 &              0.76 &               0.75 &      \textbf{0.89} &             0.62 &                &          0.41 &              0.82 &               0.81 &      \textbf{0.96} \\
                                \textbf{Diffutils} &               0.58 &         0.36 &          0.51 &              0.76 &               0.69 &      \textbf{0.82} &             0.61 &           0.85 &          0.37 &              0.87 &               0.79 &      \textbf{0.93} \\
                                \textbf{Findutils} &               0.49 &         0.28 &          0.43 &              0.58 &               0.61 &      \textbf{0.71} &             0.47 &           0.70 &          0.36 &              0.66 &               0.69 &      \textbf{0.81} \\
                                \textbf{Gzip} &               0.39 &         0.22 &          0.53 &              0.46 &               0.61 &      \textbf{0.78} &             0.39 &           0.87 &          0.39 &              0.53 &               0.71 &       \textbf{0.9} \\
                                \textbf{Lua} &               0.36 &              &          0.33 &              0.37 &               0.38 &      \textbf{0.55} &             0.36 &                &          0.26 &              0.46 &               0.47 &      \textbf{0.69} \\
                                \textbf{PuTTY} &               0.33 &              &          0.25 &              0.30 &               0.33 &       \textbf{0.5} &             0.32 &                &          0.29 &              0.33 &               0.36 &      \textbf{0.54} \\

\bottomrule
\end{tabular}
}
\end{table*}

%% file: tables/analogies-full.tex
\begin{table}[ht]
\centering
\caption{List of Analogies}
\label{tab:analogies}
\resizebox{\textwidth}{!}{
\begin{tabular}{llll}
\toprule
\multicolumn{4}{c}{\textbf{List of Analogies}} \\
\midrule
getI : putI :: get : put & divfv : divf :: vector : float & get : load :: put : store & subv : subfv :: vector : float \\
get : put :: load : store & addv : vector :: addfv : float & get : put :: load : wrtmp & mulv : vector :: mulfv : float \\
store : variable :: put : register & add : integer :: addfv : vector & get : register :: store : constant & add : integer :: addfv : float \\
put : register :: load : variable & sub : subfv :: integer : vector & put : register :: store : constant & sub : subfv :: integer : float \\
orv : or :: xorv : xor & mul : integer :: mulfv : float & or : orv :: shr : shrnv & mul : integer :: mulfv : vector \\
or : and :: orv : andv & divmod : integer :: divfv : vector & or : and :: shr : shl & divmod : integer :: divfv : float \\
and : or :: add : sub & orv : vector :: or : integer & mul : divmod :: and : or & not : integer :: notv : vector \\
shr : shl :: shrnv : shlnv & andv : vector :: and : integer & reinterpif : reinterpfi :: convif : convfi & shrnv : vector :: shr : integer \\
get : register :: geti : constant & shl : integer :: shrnv : vector & put : register :: puti : register & xor : xorv :: integer : vector \\
sub : subv :: add : addv & or : andv :: integer : vector & subf : addf :: subfv : addfv & not : integer :: negf : float \\
add : addf :: sub : subf & cmple : cmplt :: cmplefv : cmpltfv & addf : addfv :: subf : subfv & geti : integer :: get : register \\
add : sub :: addf : subf & ext : integer :: extf : float & add : sub :: addfv : subfv & extf : float :: extv : vector \\
add : addf :: mul : mulf & ext : integer :: extv : vector & add : addf :: mul : mull & hlextv : ext :: htruncv : htrunc \\
add : sub :: mul : div & trunc : integer :: truncv : vector & add : sub :: mul :: divmod & truncv : vector :: truncf : float \\
addf : subf :: mulf : divf & htrunc : trunc :: hlext : ext & addv : subv :: addf : subf & dirty : function :: if : store \\
addv : subv :: addfv : subfv & if : variable :: dirty : function & mulf : divf :: mulfv : divfv & if : exit :: put : register \\
shl : mul :: shr : divmod & store : variable :: put : register & shl : mul :: sar : divmod & load : variable :: get : register \\
shl : mul :: shr : div & get : register :: store : variable & shl : mul :: sar : div & put : register :: load : constant \\
add : sub :: mul : shr & maxv : maxfv :: minv : minfv & add : sub :: mul : sar & maxv : minv :: addv : subv \\
add : integer :: addf : float & maxfv : minfv :: addfv : subfv & subf : float :: sub : integer & if : else :: get : put \\
mul : integer :: mulf : float & or : and :: ext : trunc & divf : float :: divmod : integer & ext : trunc :: get : put \\
add : addv :: integer : vector & sqrtf : float :: sqrtfv : vector & integer : sub :: vector : subv & mulv : vector :: mul : integer \\
\bottomrule
\end{tabular}
}
\end{table}

%% file: 0.VexIR2Vec_Main.bbl

\begin{thebibliography}{114}


\ifx \showCODEN    \undefined \def \showCODEN     #1{\unskip}     \fi
\ifx \showDOI      \undefined \def \showDOI       #1{#1}\fi
\ifx \showISBNx    \undefined \def \showISBNx     #1{\unskip}     \fi
\ifx \showISBNxiii \undefined \def \showISBNxiii  #1{\unskip}     \fi
\ifx \showISSN     \undefined \def \showISSN      #1{\unskip}     \fi
\ifx \showLCCN     \undefined \def \showLCCN      #1{\unskip}     \fi
\ifx \shownote     \undefined \def \shownote      #1{#1}          \fi
\ifx \showarticletitle \undefined \def \showarticletitle #1{#1}   \fi
\ifx \showURL      \undefined \def \showURL       {\relax}        \fi
\providecommand\bibfield[2]{#2}
\providecommand\bibinfo[2]{#2}
\providecommand\natexlab[1]{#1}
\providecommand\showeprint[2][]{arXiv:#2}

\bibitem[Ahn et~al\mbox{.}(2022)]%
        {ahn2022BinShot}
\bibfield{author}{\bibinfo{person}{Sunwoo Ahn}, \bibinfo{person}{Seonggwan Ahn}, \bibinfo{person}{Hyungjoon Koo}, {and} \bibinfo{person}{Yunheung Paek}.} \bibinfo{year}{2022}\natexlab{}.
\newblock \showarticletitle{Practical Binary Code Similarity Detection with BERT-based Transferable Similarity Learning}. In \bibinfo{booktitle}{\emph{Proceedings of the 38th Annual Computer Security Applications Conference}} (Austin, TX, USA) \emph{(\bibinfo{series}{ACSAC '22})}. \bibinfo{publisher}{Association for Computing Machinery}, \bibinfo{address}{New York, NY, USA}, \bibinfo{pages}{361–374}.
\newblock
\showISBNx{9781450397599}
\urldef\tempurl%
\url{https://doi.org/10.1145/3564625.3567975}
\showDOI{\tempurl}


\bibitem[Akiba et~al\mbox{.}(2019)]%
        {akiba2019Optuna}
\bibfield{author}{\bibinfo{person}{Takuya Akiba}, \bibinfo{person}{Shotaro Sano}, \bibinfo{person}{Toshihiko Yanase}, \bibinfo{person}{Takeru Ohta}, {and} \bibinfo{person}{Masanori Koyama}.} \bibinfo{year}{2019}\natexlab{}.
\newblock \showarticletitle{Optuna: A Next-generation Hyperparameter Optimization Framework}. In \bibinfo{booktitle}{\emph{Proceedings of the 25th ACM SIGKDD International Conference on Knowledge Discovery \& Data Mining}} (Anchorage, AK, USA) \emph{(\bibinfo{series}{KDD '19})}. \bibinfo{publisher}{Association for Computing Machinery}, \bibinfo{address}{New York, NY, USA}, \bibinfo{pages}{2623–2631}.
\newblock
\showISBNx{9781450362016}
\urldef\tempurl%
\url{https://doi.org/10.1145/3292500.3330701}
\showDOI{\tempurl}


\bibitem[Ayupov et~al\mbox{.}(2024)]%
        {Ayupov24}
\bibfield{author}{\bibinfo{person}{Amir Ayupov}, \bibinfo{person}{Maksim Panchenko}, {and} \bibinfo{person}{Sergey Pupyrev}.} \bibinfo{year}{2024}\natexlab{}.
\newblock \bibinfo{title}{Stale Profile Matching}.
\newblock
\newblock
\showeprint[arxiv]{2401.17168}~[cs.PL]


\bibitem[Bahdanau et~al\mbox{.}(2015)]%
        {bahdanau2015globalAttention}
\bibfield{author}{\bibinfo{person}{Dzmitry Bahdanau}, \bibinfo{person}{Kyunghyun Cho}, {and} \bibinfo{person}{Yoshua Bengio}.} \bibinfo{year}{2015}\natexlab{}.
\newblock \showarticletitle{Neural Machine Translation by Jointly Learning to Align and Translate}. In \bibinfo{booktitle}{\emph{3rd International Conference on Learning Representations, {ICLR} 2015, May 7-9, 2015, Conference Track Proceedings}}, \bibfield{editor}{\bibinfo{person}{Yoshua Bengio} {and} \bibinfo{person}{Yann LeCun}} (Eds.). \bibinfo{address}{San Diego, CA, USA,}.
\newblock
\urldef\tempurl%
\url{http://arxiv.org/abs/1409.0473}
\showURL{%
\tempurl}


\bibitem[Baker and Manber(1998)]%
        {Baker98}
\bibfield{author}{\bibinfo{person}{Brenda~S. Baker} {and} \bibinfo{person}{Udi Manber}.} \bibinfo{year}{1998}\natexlab{}.
\newblock \showarticletitle{Deducing Similarities in Java Sources from Bytecodes}. In \bibinfo{booktitle}{\emph{1998 USENIX Annual Technical Conference (USENIX ATC 98)}}. \bibinfo{publisher}{USENIX Association}, \bibinfo{address}{New Orleans, LA}.
\newblock
\urldef\tempurl%
\url{https://www.usenix.org/conference/1998-usenix-annual-technical-conference/deducing-similarities-java-sources-bytecodes}
\showURL{%
\tempurl}


\bibitem[Bengio et~al\mbox{.}(2013)]%
        {replearning-review}
\bibfield{author}{\bibinfo{person}{Yoshua Bengio}, \bibinfo{person}{Aaron Courville}, {and} \bibinfo{person}{Pascal Vincent}.} \bibinfo{year}{2013}\natexlab{}.
\newblock \showarticletitle{Representation Learning: A Review and New Perspectives}.
\newblock \bibinfo{journal}{\emph{IEEE Trans. Pattern Anal. Mach. Intell.}} \bibinfo{volume}{35}, \bibinfo{number}{8} (\bibinfo{date}{Aug.} \bibinfo{year}{2013}), \bibinfo{pages}{1798--1828}.
\newblock
\showISSN{0162-8828}
\urldef\tempurl%
\url{https://doi.org/10.1109/TPAMI.2013.50}
\showDOI{\tempurl}


\bibitem[Bentley(1975)]%
        {bentley1975kdtrees}
\bibfield{author}{\bibinfo{person}{Jon~Louis Bentley}.} \bibinfo{year}{1975}\natexlab{}.
\newblock \showarticletitle{Multidimensional binary search trees used for associative searching}.
\newblock \bibinfo{journal}{\emph{Commun. ACM}} \bibinfo{volume}{18}, \bibinfo{number}{9} (\bibinfo{date}{sep} \bibinfo{year}{1975}), \bibinfo{pages}{509–517}.
\newblock
\showISSN{0001-0782}
\urldef\tempurl%
\url{https://doi.org/10.1145/361002.361007}
\showDOI{\tempurl}


\bibitem[Bojanowski et~al\mbox{.}(2017)]%
        {bojanowsk-2017-fasttext}
\bibfield{author}{\bibinfo{person}{Piotr Bojanowski}, \bibinfo{person}{Edouard Grave}, \bibinfo{person}{Armand Joulin}, {and} \bibinfo{person}{Tomas Mikolov}.} \bibinfo{year}{2017}\natexlab{}.
\newblock \showarticletitle{Enriching Word Vectors with Subword Information}.
\newblock \bibinfo{journal}{\emph{Transactions of the Association for Computational Linguistics}}  \bibinfo{volume}{5} (\bibinfo{year}{2017}), \bibinfo{pages}{135--146}.
\newblock
\urldef\tempurl%
\url{https://doi.org/10.1162/tacl_a_00051}
\showDOI{\tempurl}


\bibitem[Bordes et~al\mbox{.}(2013)]%
        {transe-Bordes:2013:TEM:2999792.2999923}
\bibfield{author}{\bibinfo{person}{A Bordes}, \bibinfo{person}{N Usunier}, \bibinfo{person}{A Garcia-Dur\'{a}n}, \bibinfo{person}{J Weston}, {and} \bibinfo{person}{O Yakhnenko}.} \bibinfo{year}{2013}\natexlab{}.
\newblock \showarticletitle{Translating Embeddings for Modeling Multi-relational Data} \emph{(\bibinfo{series}{NIPS'13})}. \bibinfo{pages}{2787--2795}.
\newblock
\urldef\tempurl%
\url{http://dl.acm.org/citation.cfm?id=2999792.2999923}
\showURL{%
\tempurl}


\bibitem[Bourquin et~al\mbox{.}(2013)]%
        {Bourquin13}
\bibfield{author}{\bibinfo{person}{Martial Bourquin}, \bibinfo{person}{Andy King}, {and} \bibinfo{person}{Edward Robbins}.} \bibinfo{year}{2013}\natexlab{}.
\newblock \showarticletitle{BinSlayer: Accurate Comparison of Binary Executables}. In \bibinfo{booktitle}{\emph{PPREW}} (Rome, Italy). \bibinfo{publisher}{Association for Computing Machinery}, \bibinfo{address}{New York, NY, USA}, Article \bibinfo{articleno}{4}, \bibinfo{numpages}{10}~pages.
\newblock
\showISBNx{9781450318570}
\urldef\tempurl%
\url{https://doi.org/10.1145/2430553.2430557}
\showDOI{\tempurl}


\bibitem[Brumley et~al\mbox{.}(2011)]%
        {brumley2011BAP}
\bibfield{author}{\bibinfo{person}{David Brumley}, \bibinfo{person}{Ivan Jager}, \bibinfo{person}{Thanassis Avgerinos}, {and} \bibinfo{person}{Edward~J. Schwartz}.} \bibinfo{year}{2011}\natexlab{}.
\newblock \showarticletitle{BAP: a binary analysis platform} \emph{(\bibinfo{series}{CAV'11})}. \bibinfo{publisher}{Springer-Verlag}, \bibinfo{address}{Berlin, Heidelberg}, \bibinfo{pages}{463–469}.
\newblock
\showISBNx{9783642221095}


\bibitem[Bucek et~al\mbox{.}(2018)]%
        {spec17-Bucek:2018:SCN:3185768.3185771}
\bibfield{author}{\bibinfo{person}{James Bucek}, \bibinfo{person}{Klaus-Dieter Lange}, {and} \bibinfo{person}{J\'{o}akim v. Kistowski}.} \bibinfo{year}{2018}\natexlab{}.
\newblock \showarticletitle{SPEC CPU2017: Next-Generation Compute Benchmark}. In \bibinfo{booktitle}{\emph{Companion of the 2018 ACM/SPEC International Conference on Performance Engineering}} (Berlin, Germany) \emph{(\bibinfo{series}{ICPE '18})}. \bibinfo{publisher}{ACM}, \bibinfo{address}{New York, NY, USA}, \bibinfo{pages}{41--42}.
\newblock
\showISBNx{978-1-4503-5629-9}
\urldef\tempurl%
\url{https://doi.org/10.1145/3185768.3185771}
\showDOI{\tempurl}


\bibitem[Chandramohan et~al\mbox{.}(2016)]%
        {chandramohan2016bingo}
\bibfield{author}{\bibinfo{person}{Mahinthan Chandramohan}, \bibinfo{person}{Yinxing Xue}, \bibinfo{person}{Zhengzi Xu}, \bibinfo{person}{Yang Liu}, \bibinfo{person}{Chia~Yuan Cho}, {and} \bibinfo{person}{Hee Beng~Kuan Tan}.} \bibinfo{year}{2016}\natexlab{}.
\newblock \showarticletitle{BinGo: Cross-Architecture Cross-OS Binary Search}. In \bibinfo{booktitle}{\emph{Proceedings of the 2016 24th ACM SIGSOFT International Symposium on Foundations of Software Engineering}} (Seattle, WA, USA) \emph{(\bibinfo{series}{FSE 2016})}. \bibinfo{publisher}{Association for Computing Machinery}, \bibinfo{address}{New York, NY, USA}, \bibinfo{pages}{678–689}.
\newblock
\showISBNx{9781450342186}
\urldef\tempurl%
\url{https://doi.org/10.1145/2950290.2950350}
\showDOI{\tempurl}


\bibitem[Chen et~al\mbox{.}(2020)]%
        {chen2020ContrastiveLearning}
\bibfield{author}{\bibinfo{person}{Ting Chen}, \bibinfo{person}{Simon Kornblith}, \bibinfo{person}{Mohammad Norouzi}, {and} \bibinfo{person}{Geoffrey Hinton}.} \bibinfo{year}{2020}\natexlab{}.
\newblock \showarticletitle{A simple framework for contrastive learning of visual representations}. In \bibinfo{booktitle}{\emph{Proceedings of the 37th International Conference on Machine Learning}} \emph{(\bibinfo{series}{ICML'20})}. \bibinfo{publisher}{JMLR.org}, Article \bibinfo{articleno}{149}, \bibinfo{numpages}{11}~pages.
\newblock


\bibitem[Chua et~al\mbox{.}(2017)]%
        {chua17eklavya}
\bibfield{author}{\bibinfo{person}{Zheng~Leong Chua}, \bibinfo{person}{Shiqi Shen}, \bibinfo{person}{Prateek Saxena}, {and} \bibinfo{person}{Zhenkai Liang}.} \bibinfo{year}{2017}\natexlab{}.
\newblock \showarticletitle{Neural Nets Can Learn Function Type Signatures From Binaries}. In \bibinfo{booktitle}{\emph{26th USENIX Security Symposium (USENIX Security 17)}}. \bibinfo{publisher}{USENIX Association}, \bibinfo{address}{Vancouver, BC}, \bibinfo{pages}{99--116}.
\newblock
\showISBNx{978-1-931971-40-9}
\urldef\tempurl%
\url{https://www.usenix.org/conference/usenixsecurity17/technical-sessions/presentation/chua}
\showURL{%
\tempurl}


\bibitem[Collyer et~al\mbox{.}(2023)]%
        {Collyer23}
\bibfield{author}{\bibinfo{person}{Josh Collyer}, \bibinfo{person}{Tim Watson}, {and} \bibinfo{person}{Iain Phillips}.} \bibinfo{year}{2023}\natexlab{}.
\newblock \bibinfo{title}{FASER: Binary Code Similarity Search through the use of Intermediate Representations}.
\newblock
\newblock
\showeprint[arxiv]{2310.03605}


\bibitem[Coppieters(1995)]%
        {Coppieters95}
\bibfield{author}{\bibinfo{person}{Krish Coppieters}.} \bibinfo{year}{1995}\natexlab{}.
\newblock \bibinfo{title}{A Cross-Platform Binary Diff}.
\newblock \bibinfo{howpublished}{\url{https://www.drdobbs.com/embedded-systems/a-cross-platform-binary-diff/184409550}}.
\newblock
\newblock
\shownote{[Online; accessed 12-Nov-2023]}.


\bibitem[{Coreutils}(2024)]%
        {coreutils}
\bibfield{author}{\bibinfo{person}{{Coreutils}}.} \bibinfo{year}{2024}\natexlab{}.
\newblock \bibinfo{title}{{GNU Coreutils}}.
\newblock \bibinfo{howpublished}{\url{https://www.gnu.org/software/coreutils/}}.
\newblock
\newblock
\shownote{[version 9.0; Online; accessed 08-May-2024]}.


\bibitem[Cytron et~al\mbox{.}(1991)]%
        {ssa:cytron1991efficiently}
\bibfield{author}{\bibinfo{person}{Ron Cytron}, \bibinfo{person}{Jeanne Ferrante}, \bibinfo{person}{Barry~K Rosen}, \bibinfo{person}{Mark~N Wegman}, {and} \bibinfo{person}{F~Kenneth Zadeck}.} \bibinfo{year}{1991}\natexlab{}.
\newblock \showarticletitle{Efficiently computing static single assignment form and the control dependence graph}.
\newblock \bibinfo{journal}{\emph{ACM Transactions on Programming Languages and Systems (TOPLAS)}} \bibinfo{volume}{13}, \bibinfo{number}{4} (\bibinfo{year}{1991}), \bibinfo{pages}{451--490}.
\newblock


\bibitem[Dam\'{a}sio et~al\mbox{.}(2023)]%
        {Damasio23}
\bibfield{author}{\bibinfo{person}{Tha\'{\i}s Dam\'{a}sio}, \bibinfo{person}{Michael Canesche}, \bibinfo{person}{Vin\'{\i}cius Pacheco}, \bibinfo{person}{Marcus Botacin}, \bibinfo{person}{Anderson Faustino~da Silva}, {and} \bibinfo{person}{Fernando~M. Quint\~{a}o Pereira}.} \bibinfo{year}{2023}\natexlab{}.
\newblock \showarticletitle{A Game-Based Framework to Compare Program Classifiers and Evaders}. In \bibinfo{booktitle}{\emph{Proceedings of the 21st ACM/IEEE International Symposium on Code Generation and Optimization}} (Montr\'{e}al, QC, Canada) \emph{(\bibinfo{series}{CGO 2023})}. \bibinfo{publisher}{Association for Computing Machinery}, \bibinfo{address}{New York, NY, USA}, \bibinfo{pages}{108–121}.
\newblock
\showISBNx{9798400701016}
\urldef\tempurl%
\url{https://doi.org/10.1145/3579990.3580012}
\showDOI{\tempurl}


\bibitem[David et~al\mbox{.}(2016)]%
        {yaniv2016pldi}
\bibfield{author}{\bibinfo{person}{Yaniv David}, \bibinfo{person}{Nimrod Partush}, {and} \bibinfo{person}{Eran Yahav}.} \bibinfo{year}{2016}\natexlab{}.
\newblock \showarticletitle{Statistical Similarity of Binaries}.
\newblock \bibinfo{journal}{\emph{SIGPLAN Not.}} \bibinfo{volume}{51}, \bibinfo{number}{6} (\bibinfo{date}{jun} \bibinfo{year}{2016}), \bibinfo{pages}{266–280}.
\newblock
\showISSN{0362-1340}
\urldef\tempurl%
\url{https://doi.org/10.1145/2980983.2908126}
\showDOI{\tempurl}


\bibitem[David et~al\mbox{.}(2017)]%
        {yaniv2017pldi}
\bibfield{author}{\bibinfo{person}{Yaniv David}, \bibinfo{person}{Nimrod Partush}, {and} \bibinfo{person}{Eran Yahav}.} \bibinfo{year}{2017}\natexlab{}.
\newblock \showarticletitle{Similarity of Binaries through Re-Optimization}. In \bibinfo{booktitle}{\emph{Proceedings of the 38th ACM SIGPLAN Conference on Programming Language Design and Implementation}} (Barcelona, Spain) \emph{(\bibinfo{series}{PLDI 2017})}. \bibinfo{publisher}{Association for Computing Machinery}, \bibinfo{address}{New York, NY, USA}, \bibinfo{pages}{79–94}.
\newblock
\showISBNx{9781450349888}
\urldef\tempurl%
\url{https://doi.org/10.1145/3062341.3062387}
\showDOI{\tempurl}


\bibitem[Debray et~al\mbox{.}(2000)]%
        {Debray00}
\bibfield{author}{\bibinfo{person}{Saumya~K. Debray}, \bibinfo{person}{William Evans}, \bibinfo{person}{Robert Muth}, {and} \bibinfo{person}{Bjorn De~Sutter}.} \bibinfo{year}{2000}\natexlab{}.
\newblock \showarticletitle{Compiler Techniques for Code Compaction}.
\newblock \bibinfo{journal}{\emph{ACM Trans. Program. Lang. Syst.}} \bibinfo{volume}{22}, \bibinfo{number}{2} (\bibinfo{date}{mar} \bibinfo{year}{2000}), \bibinfo{pages}{378–415}.
\newblock
\showISSN{0164-0925}
\urldef\tempurl%
\url{https://doi.org/10.1145/349214.349233}
\showDOI{\tempurl}


\bibitem[Devlin et~al\mbox{.}(2019)]%
        {devlin2019bert}
\bibfield{author}{\bibinfo{person}{Jacob Devlin}, \bibinfo{person}{Ming-Wei Chang}, \bibinfo{person}{Kenton Lee}, {and} \bibinfo{person}{Kristina Toutanova}.} \bibinfo{year}{2019}\natexlab{}.
\newblock \showarticletitle{{BERT}: Pre-training of Deep Bidirectional Transformers for Language Understanding}. In \bibinfo{booktitle}{\emph{Proceedings of the 2019 Conference of the North {A}merican Chapter of the Association for Computational Linguistics: Human Language Technologies, Volume 1 (Long and Short Papers)}}. \bibinfo{publisher}{Association for Computational Linguistics}, \bibinfo{address}{Minneapolis, Minnesota}, \bibinfo{pages}{4171--4186}.
\newblock
\urldef\tempurl%
\url{https://doi.org/10.18653/v1/N19-1423}
\showDOI{\tempurl}


\bibitem[Diffutils(2024)]%
        {diffutils}
\bibfield{author}{\bibinfo{person}{Diffutils}.} \bibinfo{year}{2024}\natexlab{}.
\newblock \bibinfo{title}{{GNU Diffutils}}.
\newblock \bibinfo{howpublished}{\url{https://www.gnu.org/software/diffutils/}}.
\newblock
\newblock
\shownote{[version 3.8; Online; accessed 08-May-2024]}.


\bibitem[Ding et~al\mbox{.}(2019)]%
        {ding2019asm2vec}
\bibfield{author}{\bibinfo{person}{Steven H.~H. Ding}, \bibinfo{person}{Benjamin C.~M. Fung}, {and} \bibinfo{person}{Philippe Charland}.} \bibinfo{year}{2019}\natexlab{}.
\newblock \showarticletitle{Asm2Vec: Boosting Static Representation Robustness for Binary Clone Search against Code Obfuscation and Compiler Optimization}. In \bibinfo{booktitle}{\emph{2019 IEEE Symposium on Security and Privacy (SP)}}. \bibinfo{pages}{472--489}.
\newblock
\urldef\tempurl%
\url{https://doi.org/10.1109/SP.2019.00003}
\showDOI{\tempurl}


\bibitem[Duan et~al\mbox{.}(2020)]%
        {duan2020deepbindiff}
\bibfield{author}{\bibinfo{person}{Yue Duan}, \bibinfo{person}{Xuezixiang Li}, \bibinfo{person}{Jinghan Wang}, {and} \bibinfo{person}{Heng Yin}.} \bibinfo{year}{2020}\natexlab{}.
\newblock \showarticletitle{Deepbindiff: Learning program-wide code representations for binary diffing}. In \bibinfo{booktitle}{\emph{Network and Distributed System Security Symposium}}.
\newblock


\bibitem[Dullien(2018)]%
        {funcsimsearch}
\bibfield{author}{\bibinfo{person}{Thomas Dullien}.} \bibinfo{year}{2018}\natexlab{}.
\newblock \bibinfo{title}{FuncSimSearch}.
\newblock \bibinfo{howpublished}{\url{https://github.com/googleprojectzero/functionsimsearch}}.
\newblock
\newblock
\shownote{[Online; accessed 13-July-2022]}.


\bibitem[Eschweiler et~al\mbox{.}(2016)]%
        {eschweiler2016discovre}
\bibfield{author}{\bibinfo{person}{Sebastian Eschweiler}, \bibinfo{person}{Khaled Yakdan}, \bibinfo{person}{Elmar Gerhards-Padilla}, {et~al\mbox{.}}} \bibinfo{year}{2016}\natexlab{}.
\newblock \showarticletitle{discovRE: Efficient Cross-Architecture Identification of Bugs in Binary Code.}. In \bibinfo{booktitle}{\emph{Ndss}}, Vol.~\bibinfo{volume}{52}. \bibinfo{pages}{58--79}.
\newblock


\bibitem[Farhadi et~al\mbox{.}(2014)]%
        {farhadi2014binclone}
\bibfield{author}{\bibinfo{person}{Mohammad~Reza Farhadi}, \bibinfo{person}{Benjamin~C.M. Fung}, \bibinfo{person}{Philippe Charland}, {and} \bibinfo{person}{Mourad Debbabi}.} \bibinfo{year}{2014}\natexlab{}.
\newblock \showarticletitle{BinClone: Detecting Code Clones in Malware}. In \bibinfo{booktitle}{\emph{2014 Eighth International Conference on Software Security and Reliability (SERE)}}. \bibinfo{pages}{78--87}.
\newblock
\urldef\tempurl%
\url{https://doi.org/10.1109/SERE.2014.21}
\showDOI{\tempurl}


\bibitem[Feng et~al\mbox{.}(2016a)]%
        {Feng2016Genius}
\bibfield{author}{\bibinfo{person}{Qian Feng}, \bibinfo{person}{Rundong Zhou}, \bibinfo{person}{Chengcheng Xu}, \bibinfo{person}{Yao Cheng}, \bibinfo{person}{Brian Testa}, {and} \bibinfo{person}{Heng Yin}.} \bibinfo{year}{2016}\natexlab{a}.
\newblock \showarticletitle{Scalable Graph-Based Bug Search for Firmware Images} \emph{(\bibinfo{series}{CCS '16})}. \bibinfo{publisher}{Association for Computing Machinery}, \bibinfo{address}{New York, NY, USA}, \bibinfo{pages}{480–491}.
\newblock
\showISBNx{9781450341394}
\urldef\tempurl%
\url{https://doi.org/10.1145/2976749.2978370}
\showDOI{\tempurl}


\bibitem[Feng et~al\mbox{.}(2016b)]%
        {qian2016bugsearch}
\bibfield{author}{\bibinfo{person}{Qian Feng}, \bibinfo{person}{Rundong Zhou}, \bibinfo{person}{Chengcheng Xu}, \bibinfo{person}{Yao Cheng}, \bibinfo{person}{Brian Testa}, {and} \bibinfo{person}{Heng Yin}.} \bibinfo{year}{2016}\natexlab{b}.
\newblock \showarticletitle{Scalable Graph-Based Bug Search for Firmware Images}. In \bibinfo{booktitle}{\emph{Proceedings of the 2016 ACM SIGSAC Conference on Computer and Communications Security}} (Vienna, Austria) \emph{(\bibinfo{series}{CCS '16})}. \bibinfo{publisher}{Association for Computing Machinery}, \bibinfo{address}{New York, NY, USA}, \bibinfo{pages}{480–491}.
\newblock
\showISBNx{9781450341394}
\urldef\tempurl%
\url{https://doi.org/10.1145/2976749.2978370}
\showDOI{\tempurl}


\bibitem[Findutils(2024)]%
        {findutils}
\bibfield{author}{\bibinfo{person}{Findutils}.} \bibinfo{year}{2024}\natexlab{}.
\newblock \bibinfo{title}{{GNU Findutils}}.
\newblock \bibinfo{howpublished}{\url{https://www.gnu.org/software/findutils/}}.
\newblock
\newblock
\shownote{[version 4.9; Online; accessed 08-May-2024]}.


\bibitem[Gao et~al\mbox{.}(2018)]%
        {gao2018vulseeker}
\bibfield{author}{\bibinfo{person}{Jian Gao}, \bibinfo{person}{Xin Yang}, \bibinfo{person}{Ying Fu}, \bibinfo{person}{Yu Jiang}, {and} \bibinfo{person}{Jiaguang Sun}.} \bibinfo{year}{2018}\natexlab{}.
\newblock \showarticletitle{VulSeeker: A Semantic Learning Based Vulnerability Seeker for Cross-Platform Binary}. In \bibinfo{booktitle}{\emph{2018 33rd IEEE/ACM International Conference on Automated Software Engineering (ASE)}}. \bibinfo{pages}{896--899}.
\newblock
\urldef\tempurl%
\url{https://doi.org/10.1145/3238147.3240480}
\showDOI{\tempurl}


\bibitem[Ghidra({[n.\,d.]})]%
        {Ghidra}
Ghidra \bibinfo{year}{[n.\,d.]}\natexlab{}.
\newblock \bibinfo{title}{Ghidra: Software Reverse Engineering Framework}.
\newblock \bibinfo{howpublished}{\url{https://ghidra-sre.org/}}.
\newblock


\bibitem[Gorchakov et~al\mbox{.}(2023)]%
        {Gorchakov23}
\bibfield{author}{\bibinfo{person}{Artyom~V. Gorchakov}, \bibinfo{person}{Liliya~A. Demidova}, {and} \bibinfo{person}{Peter~N. Sovietov}.} \bibinfo{year}{2023}\natexlab{}.
\newblock \showarticletitle{Analysis of Program Representations Based on Abstract Syntax Trees and Higher-Order Markov Chains for Source Code Classification Task}.
\newblock \bibinfo{journal}{\emph{Future Internet}} \bibinfo{volume}{15}, \bibinfo{number}{9} (\bibinfo{year}{2023}).
\newblock
\showISSN{1999-5903}
\urldef\tempurl%
\url{https://doi.org/10.3390/fi15090314}
\showDOI{\tempurl}


\bibitem[Guo et~al\mbox{.}(2019)]%
        {Guo19DeepVSA}
\bibfield{author}{\bibinfo{person}{Wenbo Guo}, \bibinfo{person}{Dongliang Mu}, \bibinfo{person}{Xinyu Xing}, \bibinfo{person}{Min Du}, {and} \bibinfo{person}{Dawn Song}.} \bibinfo{year}{2019}\natexlab{}.
\newblock \showarticletitle{{DEEPVSA}: Facilitating Value-set Analysis with Deep Learning for Postmortem Program Analysis}. In \bibinfo{booktitle}{\emph{28th USENIX Security Symposium (USENIX Security 19)}}. \bibinfo{publisher}{USENIX Association}, \bibinfo{address}{Santa Clara, CA}, \bibinfo{pages}{1787--1804}.
\newblock
\showISBNx{978-1-939133-06-9}
\urldef\tempurl%
\url{https://www.usenix.org/conference/usenixsecurity19/presentation/guo}
\showURL{%
\tempurl}


\bibitem[Gzip(2024)]%
        {gzip}
\bibfield{author}{\bibinfo{person}{Gzip}.} \bibinfo{year}{2024}\natexlab{}.
\newblock \bibinfo{title}{{GNU Gzip}}.
\newblock \bibinfo{howpublished}{\url{https://www.gnu.org/software/gzip/}}.
\newblock
\newblock
\shownote{[version 1.12; Online; accessed 08-May-2024]}.


\bibitem[Han et~al\mbox{.}(2018)]%
        {han2018openke}
\bibfield{author}{\bibinfo{person}{Xu Han}, \bibinfo{person}{Shulin Cao}, \bibinfo{person}{Xin Lv}, \bibinfo{person}{Yankai Lin}, \bibinfo{person}{Zhiyuan Liu}, \bibinfo{person}{Maosong Sun}, {and} \bibinfo{person}{Juanzi Li}.} \bibinfo{year}{2018}\natexlab{}.
\newblock \showarticletitle{{O}pen{KE}: An Open Toolkit for Knowledge Embedding}. In \bibinfo{booktitle}{\emph{Proceedings of the 2018 Conference on Empirical Methods in Natural Language Processing: System Demonstrations}}. \bibinfo{publisher}{Association for Computational Linguistics}, \bibinfo{address}{Brussels, Belgium}, \bibinfo{pages}{139--144}.
\newblock
\urldef\tempurl%
\url{https://doi.org/10.18653/v1/D18-2024}
\showDOI{\tempurl}


\bibitem[Haq and Caballero(2021)]%
        {haq2021binsurvey}
\bibfield{author}{\bibinfo{person}{Irfan~Ul Haq} {and} \bibinfo{person}{Juan Caballero}.} \bibinfo{year}{2021}\natexlab{}.
\newblock \showarticletitle{A Survey of Binary Code Similarity}.
\newblock \bibinfo{journal}{\emph{ACM Comput. Surv.}} \bibinfo{volume}{54}, \bibinfo{number}{3}, Article \bibinfo{articleno}{51} (\bibinfo{date}{apr} \bibinfo{year}{2021}), \bibinfo{numpages}{38}~pages.
\newblock
\showISSN{0360-0300}
\urldef\tempurl%
\url{https://doi.org/10.1145/3446371}
\showDOI{\tempurl}


\bibitem[He et~al\mbox{.}(2024)]%
        {He24}
\bibfield{author}{\bibinfo{person}{Haojie He}, \bibinfo{person}{Xingwei Lin}, \bibinfo{person}{Ziang Weng}, \bibinfo{person}{Ruijie Zhao}, \bibinfo{person}{Shuitao Gan}, \bibinfo{person}{Libo Chen}, \bibinfo{person}{Yuede Ji}, \bibinfo{person}{Jiashui Wang}, {and} \bibinfo{person}{Zhi Xue}.} \bibinfo{year}{2024}\natexlab{}.
\newblock \showarticletitle{Code is not Natural Language: Unlock the Power of Semantics-Oriented Graph Representation for Binary Code Similarity Detection}. In \bibinfo{booktitle}{\emph{Security}}. \bibinfo{publisher}{USENIX}.
\newblock


\bibitem[Hendrycks and Gimpel(2016)]%
        {Hendrycks2016Silu}
\bibfield{author}{\bibinfo{person}{Dan Hendrycks} {and} \bibinfo{person}{Kevin Gimpel}.} \bibinfo{year}{2016}\natexlab{}.
\newblock \showarticletitle{Bridging Nonlinearities and Stochastic Regularizers with Gaussian Error Linear Units}.
\newblock \bibinfo{journal}{\emph{CoRR}}  \bibinfo{volume}{abs/1606.08415} (\bibinfo{year}{2016}).
\newblock
\showeprint[arXiv]{1606.08415}
\urldef\tempurl%
\url{http://arxiv.org/abs/1606.08415}
\showURL{%
\tempurl}


\bibitem[Hex-Rays({[n.\,d.]})]%
        {IDAPro}
\bibfield{author}{\bibinfo{person}{Hex-Rays}.} \bibinfo{year}{[n.\,d.]}\natexlab{}.
\newblock \bibinfo{title}{IDA Pro}.
\newblock \bibinfo{howpublished}{\url{https://hex-rays.com/ida-pro/}}.
\newblock


\bibitem[Hunt and Szymanski(1977)]%
        {Hunt77}
\bibfield{author}{\bibinfo{person}{James~W. Hunt} {and} \bibinfo{person}{Thomas~G. Szymanski}.} \bibinfo{year}{1977}\natexlab{}.
\newblock \showarticletitle{A Fast Algorithm for Computing Longest Common Subsequences}.
\newblock \bibinfo{journal}{\emph{Commun. ACM}} \bibinfo{volume}{20}, \bibinfo{number}{5} (\bibinfo{date}{may} \bibinfo{year}{1977}), \bibinfo{pages}{350–353}.
\newblock
\showISSN{0001-0782}
\urldef\tempurl%
\url{https://doi.org/10.1145/359581.359603}
\showDOI{\tempurl}


\bibitem[Ioffe and Szegedy(2015)]%
        {pmlr-v37-ioffe15-batchnorm}
\bibfield{author}{\bibinfo{person}{Sergey Ioffe} {and} \bibinfo{person}{Christian Szegedy}.} \bibinfo{year}{2015}\natexlab{}.
\newblock \showarticletitle{Batch Normalization: Accelerating Deep Network Training by Reducing Internal Covariate Shift}. In \bibinfo{booktitle}{\emph{Proceedings of the 32nd International Conference on Machine Learning}} \emph{(\bibinfo{series}{Proceedings of Machine Learning Research}, Vol.~\bibinfo{volume}{37})}, \bibfield{editor}{\bibinfo{person}{Francis Bach} {and} \bibinfo{person}{David Blei}} (Eds.). \bibinfo{publisher}{PMLR}, \bibinfo{address}{Lille, France}, \bibinfo{pages}{448--456}.
\newblock
\urldef\tempurl%
\url{http://proceedings.mlr.press/v37/ioffe15.html}
\showURL{%
\tempurl}


\bibitem[Jin et~al\mbox{.}(2012)]%
        {Jin12}
\bibfield{author}{\bibinfo{person}{Wesley Jin}, \bibinfo{person}{Sagar Chaki}, \bibinfo{person}{Cory Cohen}, \bibinfo{person}{Arie Gurfinkel}, \bibinfo{person}{Jeffrey Havrilla}, \bibinfo{person}{Charles Hines}, {and} \bibinfo{person}{Priya Narasimhan}.} \bibinfo{year}{2012}\natexlab{}.
\newblock \showarticletitle{Binary Function Clustering Using Semantic Hashes}. In \bibinfo{booktitle}{\emph{ICMLA}}. \bibinfo{publisher}{IEEE Computer Society}, \bibinfo{address}{USA}, \bibinfo{pages}{386–391}.
\newblock
\showISBNx{9780769549132}
\urldef\tempurl%
\url{https://doi.org/10.1109/ICMLA.2012.70}
\showDOI{\tempurl}


\bibitem[Junod et~al\mbox{.}(2015)]%
        {ollvm-JunodRWM}
\bibfield{author}{\bibinfo{person}{Pascal Junod}, \bibinfo{person}{Julien Rinaldini}, \bibinfo{person}{Johan Wehrli}, {and} \bibinfo{person}{Julie Michielin}.} \bibinfo{year}{2015}\natexlab{}.
\newblock \showarticletitle{Obfuscator-{LLVM} -- Software Protection for the Masses}. In \bibinfo{booktitle}{\emph{Proceedings of the {IEEE/ACM} 1st International Workshop on Software Protection, {SPRO'15}, Firenze, Italy, May 19th, 2015}}, \bibfield{editor}{\bibinfo{person}{Brecht Wyseur}} (Ed.). \bibinfo{publisher}{IEEE}, \bibinfo{pages}{3--9}.
\newblock
\urldef\tempurl%
\url{https://doi.org/10.1109/SPRO.2015.10}
\showDOI{\tempurl}


\bibitem[Jurafsky(2000)]%
        {jurafsky2000nlp-book}
\bibfield{author}{\bibinfo{person}{Dan Jurafsky}.} \bibinfo{year}{2000}\natexlab{}.
\newblock \bibinfo{booktitle}{\emph{Speech \& language processing}}.
\newblock \bibinfo{publisher}{Pearson Education}.
\newblock


\bibitem[Karamitas and Kehagias(2018)]%
        {Karamitas18}
\bibfield{author}{\bibinfo{person}{Chariton Karamitas} {and} \bibinfo{person}{Athanasios Kehagias}.} \bibinfo{year}{2018}\natexlab{}.
\newblock \showarticletitle{Efficient features for function matching between binary executables}. In \bibinfo{booktitle}{\emph{SANER}}, \bibfield{editor}{\bibinfo{person}{Rocco Oliveto}, \bibinfo{person}{Massimiliano~Di Penta}, {and} \bibinfo{person}{David~C. Shepherd}} (Eds.). \bibinfo{publisher}{{IEEE} Computer Society}, \bibinfo{pages}{335--345}.
\newblock
\urldef\tempurl%
\url{https://doi.org/10.1109/SANER.2018.8330221}
\showDOI{\tempurl}


\bibitem[Kim et~al\mbox{.}(2023)]%
        {kim23binaryfeatures}
\bibfield{author}{\bibinfo{person}{Dongkwan Kim}, \bibinfo{person}{Eunsoo Kim}, \bibinfo{person}{Sang~Kil Cha}, \bibinfo{person}{Sooel Son}, {and} \bibinfo{person}{Yongdae Kim}.} \bibinfo{year}{2023}\natexlab{}.
\newblock \showarticletitle{Revisiting Binary Code Similarity Analysis Using Interpretable Feature Engineering and Lessons Learned}.
\newblock \bibinfo{journal}{\emph{IEEE Trans. Softw. Eng.}} \bibinfo{volume}{49}, \bibinfo{number}{4} (\bibinfo{date}{apr} \bibinfo{year}{2023}), \bibinfo{pages}{1661–1682}.
\newblock
\showISSN{0098-5589}
\urldef\tempurl%
\url{https://doi.org/10.1109/TSE.2022.3187689}
\showDOI{\tempurl}


\bibitem[Kim et~al\mbox{.}(2022)]%
        {Kim22}
\bibfield{author}{\bibinfo{person}{Geunwoo Kim}, \bibinfo{person}{Sanghyun Hong}, \bibinfo{person}{Michael Franz}, {and} \bibinfo{person}{Dokyung Song}.} \bibinfo{year}{2022}\natexlab{}.
\newblock \showarticletitle{Improving Cross-Platform Binary Analysis Using Representation Learning via Graph Alignment}. In \bibinfo{booktitle}{\emph{Proceedings of the 31st ACM SIGSOFT International Symposium on Software Testing and Analysis}} (Virtual, South Korea) \emph{(\bibinfo{series}{ISSTA 2022})}. \bibinfo{publisher}{Association for Computing Machinery}, \bibinfo{address}{New York, NY, USA}, \bibinfo{pages}{151–163}.
\newblock
\showISBNx{9781450393799}
\urldef\tempurl%
\url{https://doi.org/10.1145/3533767.3534383}
\showDOI{\tempurl}


\bibitem[Koch et~al\mbox{.}(2015)]%
        {koch2015siamese}
\bibfield{author}{\bibinfo{person}{Gregory Koch}, \bibinfo{person}{Richard Zemel}, \bibinfo{person}{Ruslan Salakhutdinov}, {et~al\mbox{.}}} \bibinfo{year}{2015}\natexlab{}.
\newblock \showarticletitle{Siamese neural networks for one-shot image recognition}. In \bibinfo{booktitle}{\emph{ICML deep learning workshop}}, Vol.~\bibinfo{volume}{2}. Lille.
\newblock


\bibitem[L{\'a}szl{\'o} and Kiss(2009)]%
        {laszlo2009obfuscating}
\bibfield{author}{\bibinfo{person}{T{\i}mea L{\'a}szl{\'o}} {and} \bibinfo{person}{{\'A}kos Kiss}.} \bibinfo{year}{2009}\natexlab{}.
\newblock \showarticletitle{Obfuscating C++ programs via control flow flattening}.
\newblock \bibinfo{journal}{\emph{Annales Universitatis Scientarum Budapestinensis de Rolando E{\"o}tv{\"o}s Nominatae, Sectio Computatorica}} \bibinfo{volume}{30}, \bibinfo{number}{1} (\bibinfo{year}{2009}), \bibinfo{pages}{3--19}.
\newblock


\bibitem[Lattner and Adve(2004)]%
        {Lattner:2004:llvm}
\bibfield{author}{\bibinfo{person}{Chris Lattner} {and} \bibinfo{person}{Vikram Adve}.} \bibinfo{year}{2004}\natexlab{}.
\newblock \showarticletitle{LLVM: A Compilation Framework for Lifelong Program Analysis \& Transformation}. In \bibinfo{booktitle}{\emph{Proceedings of the International Symposium on Code Generation and Optimization: Feedback-Directed and Runtime Optimization}} (Palo Alto, California) \emph{(\bibinfo{series}{CGO '04})}. \bibinfo{publisher}{IEEE Computer Society}, \bibinfo{address}{USA}, \bibinfo{pages}{75}.
\newblock
\showISBNx{0769521029}


\bibitem[Li et~al\mbox{.}(2018)]%
        {li2018asha}
\bibfield{author}{\bibinfo{person}{Lisha Li}, \bibinfo{person}{Kevin Jamieson}, \bibinfo{person}{Afshin Rostamizadeh}, \bibinfo{person}{Katya Gonina}, \bibinfo{person}{Moritz Hardt}, \bibinfo{person}{Benjamin Recht}, {and} \bibinfo{person}{Ameet Talwalkar}.} \bibinfo{year}{2018}\natexlab{}.
\newblock \bibinfo{title}{Massively Parallel Hyperparameter Tuning}.
\newblock
\newblock
\urldef\tempurl%
\url{https://openreview.net/forum?id=S1Y7OOlRZ}
\showURL{%
\tempurl}


\bibitem[Li et~al\mbox{.}(2021)]%
        {li2021palmtree}
\bibfield{author}{\bibinfo{person}{Xuezixiang Li}, \bibinfo{person}{Yu Qu}, {and} \bibinfo{person}{Heng Yin}.} \bibinfo{year}{2021}\natexlab{}.
\newblock \showarticletitle{PalmTree: Learning an Assembly Language Model for Instruction Embedding} \emph{(\bibinfo{series}{CCS '21})}. \bibinfo{publisher}{Association for Computing Machinery}, \bibinfo{address}{New York, NY, USA}, \bibinfo{pages}{3236–3251}.
\newblock
\showISBNx{9781450384544}
\urldef\tempurl%
\url{https://doi.org/10.1145/3460120.3484587}
\showDOI{\tempurl}


\bibitem[Li et~al\mbox{.}(2019)]%
        {Li19GraphMatchingNetworks}
\bibfield{author}{\bibinfo{person}{Yujia Li}, \bibinfo{person}{Chenjie Gu}, \bibinfo{person}{Thomas Dullien}, \bibinfo{person}{Oriol Vinyals}, {and} \bibinfo{person}{Pushmeet Kohli}.} \bibinfo{year}{2019}\natexlab{}.
\newblock \showarticletitle{Graph Matching Networks for Learning the Similarity of Graph Structured Objects}. In \bibinfo{booktitle}{\emph{Proceedings of the 36th International Conference on Machine Learning}} \emph{(\bibinfo{series}{Proceedings of Machine Learning Research}, Vol.~\bibinfo{volume}{97})}, \bibfield{editor}{\bibinfo{person}{Kamalika Chaudhuri} {and} \bibinfo{person}{Ruslan Salakhutdinov}} (Eds.). \bibinfo{publisher}{PMLR}, \bibinfo{pages}{3835--3845}.
\newblock
\urldef\tempurl%
\url{https://proceedings.mlr.press/v97/li19d.html}
\showURL{%
\tempurl}


\bibitem[Liaw et~al\mbox{.}(2018)]%
        {Liaw2018RayTune}
\bibfield{author}{\bibinfo{person}{Richard Liaw}, \bibinfo{person}{Eric Liang}, \bibinfo{person}{Robert Nishihara}, \bibinfo{person}{Philipp Moritz}, \bibinfo{person}{Joseph~E. Gonzalez}, {and} \bibinfo{person}{Ion Stoica}.} \bibinfo{year}{2018}\natexlab{}.
\newblock \showarticletitle{Tune: {A} Research Platform for Distributed Model Selection and Training}.
\newblock \bibinfo{journal}{\emph{CoRR}}  \bibinfo{volume}{abs/1807.05118} (\bibinfo{year}{2018}).
\newblock
\showeprint[arXiv]{1807.05118}
\urldef\tempurl%
\url{http://arxiv.org/abs/1807.05118}
\showURL{%
\tempurl}


\bibitem[Liu et~al\mbox{.}(2020)]%
        {shigang2020cyberVulnerability}
\bibfield{author}{\bibinfo{person}{Shigang Liu}, \bibinfo{person}{Mahdi Dibaei}, \bibinfo{person}{Yonghang Tai}, \bibinfo{person}{Chao Chen}, \bibinfo{person}{Jun Zhang}, {and} \bibinfo{person}{Yang Xiang}.} \bibinfo{year}{2020}\natexlab{}.
\newblock \showarticletitle{Cyber Vulnerability Intelligence for Internet of Things Binary}.
\newblock \bibinfo{journal}{\emph{IEEE Transactions on Industrial Informatics}} \bibinfo{volume}{16}, \bibinfo{number}{3} (\bibinfo{year}{2020}), \bibinfo{pages}{2154--2163}.
\newblock
\urldef\tempurl%
\url{https://doi.org/10.1109/TII.2019.2942800}
\showDOI{\tempurl}


\bibitem[Liu et~al\mbox{.}(2019)]%
        {Liu2019Roberta}
\bibfield{author}{\bibinfo{person}{Yinhan Liu}, \bibinfo{person}{Myle Ott}, \bibinfo{person}{Naman Goyal}, \bibinfo{person}{Jingfei Du}, \bibinfo{person}{Mandar Joshi}, \bibinfo{person}{Danqi Chen}, \bibinfo{person}{Omer Levy}, \bibinfo{person}{Mike Lewis}, \bibinfo{person}{Luke Zettlemoyer}, {and} \bibinfo{person}{Veselin Stoyanov}.} \bibinfo{year}{2019}\natexlab{}.
\newblock \showarticletitle{RoBERTa: {A} Robustly Optimized {BERT} Pretraining Approach}.
\newblock \bibinfo{journal}{\emph{CoRR}}  \bibinfo{volume}{abs/1907.11692} (\bibinfo{year}{2019}).
\newblock
\showeprint[arXiv]{1907.11692}
\urldef\tempurl%
\url{http://arxiv.org/abs/1907.11692}
\showURL{%
\tempurl}


\bibitem[LLVM(2018)]%
        {LLVM-LangRef}
\bibfield{author}{\bibinfo{person}{LLVM}.} \bibinfo{year}{2018}\natexlab{}.
\newblock \bibinfo{title}{{LLVM Language Reference}}.
\newblock \bibinfo{howpublished}{\url{https://llvm.org/docs/LangRef.html}}.
\newblock
\newblock
\shownote{Accessed 2019-08-20}.


\bibitem[lua(2024)]%
        {lua}
\bibfield{author}{\bibinfo{person}{lua}.} \bibinfo{year}{2024}\natexlab{}.
\newblock \bibinfo{title}{{lua}}.
\newblock \bibinfo{howpublished}{\url{https://www.lua.org/source/}}.
\newblock
\newblock
\shownote{[version 5.4.4; Online; accessed 08-May-2024]}.


\bibitem[Luo et~al\mbox{.}(2023)]%
        {Luo23}
\bibfield{author}{\bibinfo{person}{Zhenhao Luo}, \bibinfo{person}{Pengfei Wang}, \bibinfo{person}{Baosheng Wang}, \bibinfo{person}{Yong Tang}, \bibinfo{person}{Wei Xie}, \bibinfo{person}{Xu Zhou}, \bibinfo{person}{Danjun Liu}, {and} \bibinfo{person}{Kai Lu}.} \bibinfo{year}{2023}\natexlab{}.
\newblock \showarticletitle{VulHawk: Cross-architecture Vulnerability Detection with Entropy-based Binary Code Search}. In \bibinfo{booktitle}{\emph{30th Annual Network and Distributed System Security Symposium, {NDSS} 2023, San Diego, California, USA, February 27 - March 3, 2023}}. \bibinfo{publisher}{The Internet Society}.
\newblock
\urldef\tempurl%
\url{https://www.ndss-symposium.org/ndss-paper/vulhawk-cross-architecture-vulnerability-detection-with-entropy-based-binary-code-search/}
\showURL{%
\tempurl}


\bibitem[Marcelli et~al\mbox{.}(2022)]%
        {marcelli2022usenix}
\bibfield{author}{\bibinfo{person}{Andrea Marcelli}, \bibinfo{person}{Mariano Graziano}, \bibinfo{person}{Xabier Ugarte-Pedrero}, \bibinfo{person}{Yanick Fratantonio}, \bibinfo{person}{Mohamad Mansouri}, {and} \bibinfo{person}{Davide Balzarotti}.} \bibinfo{year}{2022}\natexlab{}.
\newblock \showarticletitle{How Machine Learning Is Solving the Binary Function Similarity Problem}. In \bibinfo{booktitle}{\emph{31st USENIX Security Symposium (USENIX Security 22)}}. \bibinfo{publisher}{USENIX Association}, \bibinfo{address}{Boston, MA}.
\newblock
\urldef\tempurl%
\url{https://www.usenix.org/conference/usenixsecurity22/presentation/marcelli}
\showURL{%
\tempurl}


\bibitem[Mart\'{\i}nez et~al\mbox{.}(2023)]%
        {Martinez23}
\bibfield{author}{\bibinfo{person}{Pablo~Antonio Mart\'{\i}nez}, \bibinfo{person}{Jackson Woodruff}, \bibinfo{person}{Jordi Armengol-Estap\'{e}}, \bibinfo{person}{Gregorio Bernab\'{e}}, \bibinfo{person}{Jos\'{e}~Manuel Garc\'{\i}a}, {and} \bibinfo{person}{Michael F.~P. O’Boyle}.} \bibinfo{year}{2023}\natexlab{}.
\newblock \showarticletitle{Matching Linear Algebra and Tensor Code to Specialized Hardware Accelerators}. In \bibinfo{booktitle}{\emph{International Conference on Compiler Construction}} (Montr\'{e}al, QC, Canada). \bibinfo{publisher}{Association for Computing Machinery}, \bibinfo{address}{New York, NY, USA}, \bibinfo{pages}{85–97}.
\newblock
\showISBNx{9798400700880}
\urldef\tempurl%
\url{https://doi.org/10.1145/3578360.3580262}
\showDOI{\tempurl}


\bibitem[Massarelli et~al\mbox{.}(2019)]%
        {massarelli2019safe}
\bibfield{author}{\bibinfo{person}{Luca Massarelli}, \bibinfo{person}{Giuseppe Antonio~Di Luna}, \bibinfo{person}{Fabio Petroni}, \bibinfo{person}{Roberto Baldoni}, {and} \bibinfo{person}{Leonardo Querzoni}.} \bibinfo{year}{2019}\natexlab{}.
\newblock \showarticletitle{Safe: Self-attentive function embeddings for binary similarity}. In \bibinfo{booktitle}{\emph{International Conference on Detection of Intrusions and Malware, and Vulnerability Assessment}}. Springer, \bibinfo{pages}{309--329}.
\newblock


\bibitem[Mikolov et~al\mbox{.}(2013)]%
        {mikolov2013word2vec}
\bibfield{author}{\bibinfo{person}{T Mikolov}, \bibinfo{person}{K Chen}, \bibinfo{person}{G Corrado}, {and} \bibinfo{person}{J Dean}.} \bibinfo{year}{2013}\natexlab{}.
\newblock \showarticletitle{Efficient estimation of word representations in vector space}.
\newblock \bibinfo{journal}{\emph{arXiv preprint arXiv:1301.3781}} (\bibinfo{year}{2013}).
\newblock


\bibitem[Ming et~al\mbox{.}(2016)]%
        {Ming2016PlagDetection}
\bibfield{author}{\bibinfo{person}{Jiang Ming}, \bibinfo{person}{Fangfang Zhang}, \bibinfo{person}{Dinghao Wu}, \bibinfo{person}{Peng Liu}, {and} \bibinfo{person}{Sencun Zhu}.} \bibinfo{year}{2016}\natexlab{}.
\newblock \showarticletitle{Deviation-Based Obfuscation-Resilient Program Equivalence Checking With Application to Software Plagiarism Detection}.
\newblock \bibinfo{journal}{\emph{IEEE Transactions on Reliability}} \bibinfo{volume}{65}, \bibinfo{number}{4} (\bibinfo{year}{2016}), \bibinfo{pages}{1647--1664}.
\newblock
\urldef\tempurl%
\url{https://doi.org/10.1109/TR.2016.2570554}
\showDOI{\tempurl}


\bibitem[Morris(1973)]%
        {Lockwood-Morris-Semantics-10.1145/512927.512941}
\bibfield{author}{\bibinfo{person}{F.~Lockwood Morris}.} \bibinfo{year}{1973}\natexlab{}.
\newblock \showarticletitle{Advice on Structuring Compilers and Proving Them Correct}. In \bibinfo{booktitle}{\emph{Proceedings of the 1st Annual ACM SIGACT-SIGPLAN Symposium on Principles of Programming Languages}} (Boston, Massachusetts) \emph{(\bibinfo{series}{POPL '73})}. \bibinfo{publisher}{Association for Computing Machinery}, \bibinfo{address}{New York, NY, USA}, \bibinfo{pages}{144–152}.
\newblock
\showISBNx{9781450373494}
\urldef\tempurl%
\url{https://doi.org/10.1145/512927.512941}
\showDOI{\tempurl}


\bibitem[Muchnick(1997)]%
        {muchnick1997advanced}
\bibfield{author}{\bibinfo{person}{Steven~S. Muchnick}.} \bibinfo{year}{1997}\natexlab{}.
\newblock \bibinfo{booktitle}{\emph{Advanced Compiler Design and Implementation}}.
\newblock \bibinfo{publisher}{Morgan Kaufmann Publishers Inc.}, \bibinfo{address}{San Francisco, CA, USA}.
\newblock
\showISBNx{1-55860-320-4}


\bibitem[Nethercote and Seward(2007)]%
        {nicholas2007valgrind}
\bibfield{author}{\bibinfo{person}{Nicholas Nethercote} {and} \bibinfo{person}{Julian Seward}.} \bibinfo{year}{2007}\natexlab{}.
\newblock \showarticletitle{Valgrind: A Framework for Heavyweight Dynamic Binary Instrumentation}. In \bibinfo{booktitle}{\emph{Proceedings of the 28th ACM SIGPLAN Conference on Programming Language Design and Implementation}} (San Diego, California, USA) \emph{(\bibinfo{series}{PLDI '07})}. \bibinfo{publisher}{Association for Computing Machinery}, \bibinfo{address}{New York, NY, USA}, \bibinfo{pages}{89–100}.
\newblock
\showISBNx{9781595936332}
\urldef\tempurl%
\url{https://doi.org/10.1145/1250734.1250746}
\showDOI{\tempurl}


\bibitem[Panchenko et~al\mbox{.}(2019)]%
        {Panchenko19}
\bibfield{author}{\bibinfo{person}{Maksim Panchenko}, \bibinfo{person}{Rafael Auler}, \bibinfo{person}{Bill Nell}, {and} \bibinfo{person}{Guilherme Ottoni}.} \bibinfo{year}{2019}\natexlab{}.
\newblock \showarticletitle{BOLT: A Practical Binary Optimizer for Data Centers and Beyond}. In \bibinfo{booktitle}{\emph{Proceedings of the 2019 IEEE/ACM International Symposium on Code Generation and Optimization}} (Washington, DC, USA) \emph{(\bibinfo{series}{CGO 2019})}. \bibinfo{publisher}{IEEE Press}, \bibinfo{pages}{2–14}.
\newblock
\showISBNx{9781728114361}


\bibitem[Paulheim(2017)]%
        {paulheim2017knowledge}
\bibfield{author}{\bibinfo{person}{Heiko Paulheim}.} \bibinfo{year}{2017}\natexlab{}.
\newblock \showarticletitle{Knowledge graph refinement: A survey of approaches and evaluation methods}.
\newblock \bibinfo{journal}{\emph{Semantic web}} \bibinfo{volume}{8}, \bibinfo{number}{3} (\bibinfo{year}{2017}), \bibinfo{pages}{489--508}.
\newblock


\bibitem[Pedregosa et~al\mbox{.}(2011)]%
        {scikit-learn}
\bibfield{author}{\bibinfo{person}{F. Pedregosa}, \bibinfo{person}{G. Varoquaux}, \bibinfo{person}{A. Gramfort}, \bibinfo{person}{V. Michel}, \bibinfo{person}{B. Thirion}, \bibinfo{person}{O. Grisel}, \bibinfo{person}{M. Blondel}, \bibinfo{person}{P. Prettenhofer}, \bibinfo{person}{R. Weiss}, \bibinfo{person}{V. Dubourg}, \bibinfo{person}{J. Vanderplas}, \bibinfo{person}{A. Passos}, \bibinfo{person}{D. Cournapeau}, \bibinfo{person}{M. Brucher}, \bibinfo{person}{M. Perrot}, {and} \bibinfo{person}{E. Duchesnay}.} \bibinfo{year}{2011}\natexlab{}.
\newblock \showarticletitle{Scikit-learn: Machine Learning in {P}ython}.
\newblock \bibinfo{journal}{\emph{Journal of Machine Learning Research}}  \bibinfo{volume}{12} (\bibinfo{year}{2011}), \bibinfo{pages}{2825--2830}.
\newblock


\bibitem[Pei et~al\mbox{.}(2020)]%
        {pei2020trex}
\bibfield{author}{\bibinfo{person}{Kexin Pei}, \bibinfo{person}{Zhou Xuan}, \bibinfo{person}{Junfeng Yang}, \bibinfo{person}{Suman Jana}, {and} \bibinfo{person}{Baishakhi Ray}.} \bibinfo{year}{2020}\natexlab{}.
\newblock \showarticletitle{Trex: Learning execution semantics from micro-traces for binary similarity}.
\newblock \bibinfo{journal}{\emph{arXiv preprint arXiv:2012.08680}} (\bibinfo{year}{2020}).
\newblock


\bibitem[Peng et~al\mbox{.}(2021)]%
        {peng2021oscar}
\bibfield{author}{\bibinfo{person}{Dinglan Peng}, \bibinfo{person}{Shuxin Zheng}, \bibinfo{person}{Yatao Li}, \bibinfo{person}{Guolin Ke}, \bibinfo{person}{Di He}, {and} \bibinfo{person}{Tie-Yan Liu}.} \bibinfo{year}{2021}\natexlab{}.
\newblock \showarticletitle{How could Neural Networks understand Programs?}. In \bibinfo{booktitle}{\emph{International Conference on Machine Learning}}. PMLR, \bibinfo{pages}{8476--8486}.
\newblock


\bibitem[Pewny et~al\mbox{.}(2015)]%
        {pewny2015sp}
\bibfield{author}{\bibinfo{person}{Jannik Pewny}, \bibinfo{person}{Behrad Garmany}, \bibinfo{person}{Robert Gawlik}, \bibinfo{person}{Christian Rossow}, {and} \bibinfo{person}{Thorsten Holz}.} \bibinfo{year}{2015}\natexlab{}.
\newblock \showarticletitle{Cross-Architecture Bug Search in Binary Executables}. In \bibinfo{booktitle}{\emph{2015 IEEE Symposium on Security and Privacy}}. \bibinfo{pages}{709--724}.
\newblock
\urldef\tempurl%
\url{https://doi.org/10.1109/SP.2015.49}
\showDOI{\tempurl}


\bibitem[PuTTY(2024)]%
        {putty}
\bibfield{author}{\bibinfo{person}{PuTTY}.} \bibinfo{year}{2024}\natexlab{}.
\newblock \bibinfo{title}{{PuTTY}}.
\newblock \bibinfo{howpublished}{\url{https://www.chiark.greenend.org.uk/~sgtatham/putty/}}.
\newblock
\newblock
\shownote{[version 0.76; Online; accessed 08-May-2024]}.


\bibitem[Qasem et~al\mbox{.}(2023)]%
        {Qasem2023Binfinder-AsiaCCS}
\bibfield{author}{\bibinfo{person}{Abdullah Qasem}, \bibinfo{person}{Mourad Debbabi}, \bibinfo{person}{Bernard Lebel}, {and} \bibinfo{person}{Marthe Kassouf}.} \bibinfo{year}{2023}\natexlab{}.
\newblock \showarticletitle{Binary Function Clone Search in the Presence of Code Obfuscation and Optimization over Multi-CPU Architectures}. In \bibinfo{booktitle}{\emph{Proceedings of the 2023 ACM Asia Conference on Computer and Communications Security}} (<conf-loc>, <city>Melbourne</city>, <state>VIC</state>, <country>Australia</country>, </conf-loc>) \emph{(\bibinfo{series}{ASIA CCS '23})}. \bibinfo{publisher}{Association for Computing Machinery}, \bibinfo{address}{New York, NY, USA}, \bibinfo{pages}{443–456}.
\newblock
\showISBNx{9798400700989}
\urldef\tempurl%
\url{https://doi.org/10.1145/3579856.3582818}
\showDOI{\tempurl}


\bibitem[Rastello and Tichadou(2022)]%
        {rastello2022ssa}
\bibfield{author}{\bibinfo{person}{Fabrice Rastello} {and} \bibinfo{person}{Florent~Bouchez Tichadou}.} \bibinfo{year}{2022}\natexlab{}.
\newblock \bibinfo{booktitle}{\emph{SSA-based Compiler Design}}.
\newblock \bibinfo{publisher}{Springer Nature}.
\newblock


\bibitem[Redmond et~al\mbox{.}(2018)]%
        {Redmond18}
\bibfield{author}{\bibinfo{person}{Kimberly Redmond}, \bibinfo{person}{Lannan Luo}, {and} \bibinfo{person}{Qiang Zeng}.} \bibinfo{year}{2018}\natexlab{}.
\newblock \bibinfo{title}{A Cross-Architecture Instruction Embedding Model for Natural Language Processing-Inspired Binary Code Analysis}.
\newblock
\newblock
\showeprint[arxiv]{1812.09652}~[cs.CR]


\bibitem[Rice(1953)]%
        {Rice53}
\bibfield{author}{\bibinfo{person}{H.~G. Rice}.} \bibinfo{year}{1953}\natexlab{}.
\newblock \showarticletitle{Classes of Recursively Enumerable Sets and Their Decision Problems}.
\newblock \bibinfo{journal}{\emph{Trans. Amer. Math. Soc.}} \bibinfo{volume}{74}, \bibinfo{number}{2} (\bibinfo{year}{1953}), \bibinfo{pages}{358--366}.
\newblock
\showISSN{00029947}
\urldef\tempurl%
\url{http://www.jstor.org/stable/1990888}
\showURL{%
\tempurl}


\bibitem[{S. VenkataKeerthy, Yashas Andaluri, Sayan Dey, Rinku Shah, Praveen Tammana, Ramakrishna Upadrasta}(2023)]%
        {VenkataKeerthy-2023-packet_processing}
\bibfield{author}{\bibinfo{person}{{S. VenkataKeerthy, Yashas Andaluri, Sayan Dey, Rinku Shah, Praveen Tammana, Ramakrishna Upadrasta}}.} \bibinfo{year}{2023}\natexlab{}.
\newblock \showarticletitle{Packet Processing Algorithm Identification using Program Embeddings}. In \bibinfo{booktitle}{\emph{Proceedings of the 6th Asia-Pacific Workshop on Networking}} (Fuzhou, China) \emph{(\bibinfo{series}{APNet '22})}. \bibinfo{publisher}{Association for Computing Machinery}, \bibinfo{address}{New York, NY, USA}, \bibinfo{pages}{76–82}.
\newblock
\showISBNx{9781450397483}
\urldef\tempurl%
\url{https://doi.org/10.1145/3542637.3542649}
\showDOI{\tempurl}


\bibitem[Schroff et~al\mbox{.}(2015)]%
        {schroff2015facenet}
\bibfield{author}{\bibinfo{person}{Florian Schroff}, \bibinfo{person}{Dmitry Kalenichenko}, {and} \bibinfo{person}{James Philbin}.} \bibinfo{year}{2015}\natexlab{}.
\newblock \showarticletitle{Facenet: A unified embedding for face recognition and clustering}. In \bibinfo{booktitle}{\emph{Proceedings of the IEEE conference on computer vision and pattern recognition}}. \bibinfo{pages}{815--823}.
\newblock


\bibitem[Shalev and Partush(2018)]%
        {zeek2018plas}
\bibfield{author}{\bibinfo{person}{Noam Shalev} {and} \bibinfo{person}{Nimrod Partush}.} \bibinfo{year}{2018}\natexlab{}.
\newblock \showarticletitle{Binary Similarity Detection Using Machine Learning}. In \bibinfo{booktitle}{\emph{Proceedings of the 13th Workshop on Programming Languages and Analysis for Security}} (Toronto, Canada) \emph{(\bibinfo{series}{PLAS '18})}. \bibinfo{publisher}{Association for Computing Machinery}, \bibinfo{address}{New York, NY, USA}, \bibinfo{pages}{42–47}.
\newblock
\showISBNx{9781450359931}
\urldef\tempurl%
\url{https://doi.org/10.1145/3264820.3264821}
\showDOI{\tempurl}


\bibitem[Shoshitaishvili et~al\mbox{.}(2015)]%
        {shoshitaishvili2015PyVex}
\bibfield{author}{\bibinfo{person}{Yan Shoshitaishvili}, \bibinfo{person}{Ruoyu Wang}, \bibinfo{person}{Christophe Hauser}, \bibinfo{person}{Christopher Kruegel}, {and} \bibinfo{person}{Giovanni Vigna}.} \bibinfo{year}{2015}\natexlab{}.
\newblock \showarticletitle{Firmalice - Automatic Detection of Authentication Bypass Vulnerabilities in Binary Firmware}. In \bibinfo{booktitle}{\emph{22nd Annual Network and Distributed System Security Symposium, {NDSS} 2015, San Diego, California, USA, February 8-11, 2015}}. \bibinfo{publisher}{The Internet Society}.
\newblock
\urldef\tempurl%
\url{https://www.ndss-symposium.org/ndss2015/firmalice-automatic-detection-authentication-bypass-vulnerabilities-binary-firmware}
\showURL{%
\tempurl}


\bibitem[Stenberg(2024)]%
        {curl}
\bibfield{author}{\bibinfo{person}{Daniel Stenberg}.} \bibinfo{year}{2024}\natexlab{}.
\newblock \bibinfo{title}{{cURL}}.
\newblock \bibinfo{howpublished}{\url{https://curl.se/}}.
\newblock
\newblock
\shownote{[version 7.83.0; Online; accessed 08-May-2024]}.


\bibitem[Tang et~al\mbox{.}(2018)]%
        {wei2018bcfinder}
\bibfield{author}{\bibinfo{person}{W. Tang}, \bibinfo{person}{D. Chen}, {and} \bibinfo{person}{P. Luo}.} \bibinfo{year}{2018}\natexlab{}.
\newblock \showarticletitle{BCFinder: A Lightweight and Platform-Independent Tool to Find Third-Party Components in Binaries}. In \bibinfo{booktitle}{\emph{2018 25th Asia-Pacific Software Engineering Conference (APSEC)}}. \bibinfo{publisher}{IEEE Computer Society}, \bibinfo{address}{Los Alamitos, CA, USA}, \bibinfo{pages}{288--297}.
\newblock
\urldef\tempurl%
\url{https://doi.org/10.1109/APSEC.2018.00043}
\showDOI{\tempurl}


\bibitem[Team(2017)]%
        {Radare}
\bibfield{author}{\bibinfo{person}{Radare2 Team}.} \bibinfo{year}{2017}\natexlab{}.
\newblock \bibinfo{booktitle}{\emph{Radare2 Book}}.
\newblock \bibinfo{publisher}{GitHub}.
\newblock


\bibitem[Tian et~al\mbox{.}(2024)]%
        {tian2024collapseaware}
\bibfield{author}{\bibinfo{person}{Qiwei Tian}, \bibinfo{person}{Chenhao Lin}, \bibinfo{person}{Zhengyu Zhao}, \bibinfo{person}{Qian Li}, {and} \bibinfo{person}{Chao Shen}.} \bibinfo{year}{2024}\natexlab{}.
\newblock \showarticletitle{Collapse-Aware Triplet Decoupling for Adversarially Robust Image Retrieval}. In \bibinfo{booktitle}{\emph{Forty-first International Conference on Machine Learning}}.
\newblock
\urldef\tempurl%
\url{https://openreview.net/forum?id=cy3JBZKCw1}
\showURL{%
\tempurl}


\bibitem[Ullmann(1976)]%
        {Ullmann76}
\bibfield{author}{\bibinfo{person}{J.~R. Ullmann}.} \bibinfo{year}{1976}\natexlab{}.
\newblock \showarticletitle{An Algorithm for Subgraph Isomorphism}.
\newblock \bibinfo{journal}{\emph{J. ACM}} \bibinfo{volume}{23}, \bibinfo{number}{1} (\bibinfo{date}{jan} \bibinfo{year}{1976}), \bibinfo{pages}{31–42}.
\newblock
\showISSN{0004-5411}
\urldef\tempurl%
\url{https://doi.org/10.1145/321921.321925}
\showDOI{\tempurl}


\bibitem[van~der Maaten and Hinton(2008)]%
        {van2008tsne}
\bibfield{author}{\bibinfo{person}{Laurens van~der Maaten} {and} \bibinfo{person}{Geoffrey Hinton}.} \bibinfo{year}{2008}\natexlab{}.
\newblock \showarticletitle{Visualizing Data using t-SNE}.
\newblock \bibinfo{journal}{\emph{Journal of Machine Learning Research}} \bibinfo{volume}{9}, \bibinfo{number}{86} (\bibinfo{year}{2008}), \bibinfo{pages}{2579--2605}.
\newblock
\urldef\tempurl%
\url{http://jmlr.org/papers/v9/vandermaaten08a.html}
\showURL{%
\tempurl}


\bibitem[Vaswani et~al\mbox{.}(2017)]%
        {Vaswani2017attentionTransformers}
\bibfield{author}{\bibinfo{person}{Ashish Vaswani}, \bibinfo{person}{Noam Shazeer}, \bibinfo{person}{Niki Parmar}, \bibinfo{person}{Jakob Uszkoreit}, \bibinfo{person}{Llion Jones}, \bibinfo{person}{Aidan~N Gomez}, \bibinfo{person}{\L~ukasz Kaiser}, {and} \bibinfo{person}{Illia Polosukhin}.} \bibinfo{year}{2017}\natexlab{}.
\newblock \showarticletitle{Attention is All you Need}. In \bibinfo{booktitle}{\emph{Advances in Neural Information Processing Systems}}, \bibfield{editor}{\bibinfo{person}{I.~Guyon}, \bibinfo{person}{U.~Von Luxburg}, \bibinfo{person}{S.~Bengio}, \bibinfo{person}{H.~Wallach}, \bibinfo{person}{R.~Fergus}, \bibinfo{person}{S.~Vishwanathan}, {and} \bibinfo{person}{R.~Garnett}} (Eds.), Vol.~\bibinfo{volume}{30}. \bibinfo{publisher}{Curran Associates, Inc.}
\newblock
\urldef\tempurl%
\url{https://proceedings.neurips.cc/paper/2017/file/3f5ee243547dee91fbd053c1c4a845aa-Paper.pdf}
\showURL{%
\tempurl}


\bibitem[VenkataKeerthy et~al\mbox{.}(2020)]%
        {VenkataKeerthy-2020-IR2Vec}
\bibfield{author}{\bibinfo{person}{S. VenkataKeerthy}, \bibinfo{person}{R Aggarwal}, \bibinfo{person}{S Jain}, \bibinfo{person}{M~S Desarkar}, \bibinfo{person}{R Upadrasta}, {and} \bibinfo{person}{Y.~N. Srikant}.} \bibinfo{year}{2020}\natexlab{}.
\newblock \showarticletitle{{IR2Vec: LLVM IR Based Scalable Program Embeddings}}.
\newblock \bibinfo{journal}{\emph{ACM Trans. Archit. Code Optim.}} \bibinfo{volume}{17}, \bibinfo{number}{4}, Article \bibinfo{articleno}{32} (\bibinfo{date}{Dec.} \bibinfo{year}{2020}), \bibinfo{numpages}{27}~pages.
\newblock
\showISSN{1544-3566}
\urldef\tempurl%
\url{https://doi.org/10.1145/3418463}
\showDOI{\tempurl}


\bibitem[Wang and Shoshitaishvili(2017)]%
        {wang2017angr}
\bibfield{author}{\bibinfo{person}{F. Wang} {and} \bibinfo{person}{Y. Shoshitaishvili}.} \bibinfo{year}{2017}\natexlab{}.
\newblock \showarticletitle{Angr - The Next Generation of Binary Analysis}. In \bibinfo{booktitle}{\emph{2017 IEEE Cybersecurity Development (SecDev)}}. \bibinfo{publisher}{IEEE Computer Society}, \bibinfo{address}{Los Alamitos, CA, USA}, \bibinfo{pages}{8--9}.
\newblock
\urldef\tempurl%
\url{https://doi.org/10.1109/SecDev.2017.14}
\showDOI{\tempurl}


\bibitem[Wang et~al\mbox{.}(2023)]%
        {wang2023sem2vec}
\bibfield{author}{\bibinfo{person}{Huaijin Wang}, \bibinfo{person}{Pingchuan Ma}, \bibinfo{person}{Shuai Wang}, \bibinfo{person}{Qiyi Tang}, \bibinfo{person}{Sen Nie}, {and} \bibinfo{person}{Shi Wu}.} \bibinfo{year}{2023}\natexlab{}.
\newblock \showarticletitle{sem2vec: Semantics-aware Assembly Tracelet Embedding}.
\newblock \bibinfo{journal}{\emph{ACM Trans. Softw. Eng. Methodol.}} \bibinfo{volume}{32}, \bibinfo{number}{4}, Article \bibinfo{articleno}{90} (\bibinfo{date}{may} \bibinfo{year}{2023}), \bibinfo{numpages}{34}~pages.
\newblock
\showISSN{1049-331X}
\urldef\tempurl%
\url{https://doi.org/10.1145/3569933}
\showDOI{\tempurl}


\bibitem[Wang et~al\mbox{.}(2022)]%
        {Wang22jTrans}
\bibfield{author}{\bibinfo{person}{Hao Wang}, \bibinfo{person}{Wenjie Qu}, \bibinfo{person}{Gilad Katz}, \bibinfo{person}{Wenyu Zhu}, \bibinfo{person}{Zeyu Gao}, \bibinfo{person}{Han Qiu}, \bibinfo{person}{Jianwei Zhuge}, {and} \bibinfo{person}{Chao Zhang}.} \bibinfo{year}{2022}\natexlab{}.
\newblock \showarticletitle{JTrans: Jump-Aware Transformer for Binary Code Similarity Detection}. In \bibinfo{booktitle}{\emph{Proceedings of the 31st ACM SIGSOFT International Symposium on Software Testing and Analysis}} (Virtual, South Korea) \emph{(\bibinfo{series}{ISSTA 2022})}. \bibinfo{publisher}{Association for Computing Machinery}, \bibinfo{address}{New York, NY, USA}, \bibinfo{pages}{1–13}.
\newblock
\showISBNx{9781450393799}
\urldef\tempurl%
\url{https://doi.org/10.1145/3533767.3534367}
\showDOI{\tempurl}


\bibitem[{Wang} et~al\mbox{.}(2017)]%
        {knowledge-graph-embedding-survey}
\bibfield{author}{\bibinfo{person}{Q. {Wang}}, \bibinfo{person}{Z. {Mao}}, \bibinfo{person}{B. {Wang}}, {and} \bibinfo{person}{L. {Guo}}.} \bibinfo{year}{2017}\natexlab{}.
\newblock \showarticletitle{Knowledge Graph Embedding: A Survey of Approaches and Applications}.
\newblock \bibinfo{journal}{\emph{IEEE Transactions on Knowledge and Data Engineering}} \bibinfo{volume}{29}, \bibinfo{number}{12} (\bibinfo{date}{Dec} \bibinfo{year}{2017}), \bibinfo{pages}{2724--2743}.
\newblock
\showISSN{1041-4347}
\urldef\tempurl%
\url{https://doi.org/10.1109/TKDE.2017.2754499}
\showDOI{\tempurl}


\bibitem[Wang et~al\mbox{.}(2000)]%
        {Wang00}
\bibfield{author}{\bibinfo{person}{Zheng Wang}, \bibinfo{person}{Ken Pierce}, {and} \bibinfo{person}{Scott McFarling}.} \bibinfo{year}{2000}\natexlab{}.
\newblock \showarticletitle{Bmat-a binary matching tool for stale profile propagation}.
\newblock \bibinfo{journal}{\emph{The Journal of Instruction-Level Parallelism}}  \bibinfo{volume}{2} (\bibinfo{year}{2000}), \bibinfo{pages}{1--20}.
\newblock


\bibitem[Weiser(1984)]%
        {weiser1984programSlicing}
\bibfield{author}{\bibinfo{person}{Mark Weiser}.} \bibinfo{year}{1984}\natexlab{}.
\newblock \showarticletitle{Program Slicing}.
\newblock \bibinfo{journal}{\emph{IEEE Transactions on Software Engineering}} \bibinfo{volume}{SE-10}, \bibinfo{number}{4} (\bibinfo{year}{1984}), \bibinfo{pages}{352--357}.
\newblock
\urldef\tempurl%
\url{https://doi.org/10.1109/TSE.1984.5010248}
\showDOI{\tempurl}


\bibitem[Wu et~al\mbox{.}(2022)]%
        {GNNBook2022}
\bibfield{author}{\bibinfo{person}{Lingfei Wu}, \bibinfo{person}{Peng Cui}, \bibinfo{person}{Jian Pei}, {and} \bibinfo{person}{Liang Zhao}.} \bibinfo{year}{2022}\natexlab{}.
\newblock \bibinfo{booktitle}{\emph{Graph Neural Networks: Foundations, Frontiers, and Applications}}.
\newblock \bibinfo{publisher}{Springer Singapore}, \bibinfo{address}{Singapore}. 725 pages.
\newblock


\bibitem[Xu et~al\mbox{.}(2017)]%
        {xu2017gemini}
\bibfield{author}{\bibinfo{person}{Xiaojun Xu}, \bibinfo{person}{Chang Liu}, \bibinfo{person}{Qian Feng}, \bibinfo{person}{Heng Yin}, \bibinfo{person}{Le Song}, {and} \bibinfo{person}{Dawn Song}.} \bibinfo{year}{2017}\natexlab{}.
\newblock \showarticletitle{Neural Network-Based Graph Embedding for Cross-Platform Binary Code Similarity Detection}. In \bibinfo{booktitle}{\emph{Proceedings of the 2017 ACM SIGSAC Conference on Computer and Communications Security}} (Dallas, Texas, USA) \emph{(\bibinfo{series}{CCS '17})}. \bibinfo{publisher}{Association for Computing Machinery}, \bibinfo{address}{New York, NY, USA}, \bibinfo{pages}{363–376}.
\newblock
\showISBNx{9781450349468}
\urldef\tempurl%
\url{https://doi.org/10.1145/3133956.3134018}
\showDOI{\tempurl}


\bibitem[Xuan et~al\mbox{.}(2020b)]%
        {xuan2020hard}
\bibfield{author}{\bibinfo{person}{Hong Xuan}, \bibinfo{person}{Abby Stylianou}, \bibinfo{person}{Xiaotong Liu}, {and} \bibinfo{person}{Robert Pless}.} \bibinfo{year}{2020}\natexlab{b}.
\newblock \showarticletitle{Hard negative examples are hard, but useful}. In \bibinfo{booktitle}{\emph{European Conference on Computer Vision}}. Springer, \bibinfo{pages}{126--142}.
\newblock


\bibitem[Xuan et~al\mbox{.}(2020a)]%
        {xuan2020EasyHardMining}
\bibfield{author}{\bibinfo{person}{H. Xuan}, \bibinfo{person}{A. Stylianou}, {and} \bibinfo{person}{R. Pless}.} \bibinfo{year}{2020}\natexlab{a}.
\newblock \showarticletitle{Improved Embeddings with Easy Positive Triplet Mining}. In \bibinfo{booktitle}{\emph{2020 IEEE Winter Conference on Applications of Computer Vision (WACV)}}. \bibinfo{publisher}{IEEE Computer Society}, \bibinfo{address}{Los Alamitos, CA, USA}, \bibinfo{pages}{2463--2471}.
\newblock
\urldef\tempurl%
\url{https://doi.org/10.1109/WACV45572.2020.9093432}
\showDOI{\tempurl}


\bibitem[Xue et~al\mbox{.}(2018)]%
        {Xue18}
\bibfield{author}{\bibinfo{person}{Hongfa Xue}, \bibinfo{person}{Guru Venkataramani}, {and} \bibinfo{person}{Tian Lan}.} \bibinfo{year}{2018}\natexlab{}.
\newblock \showarticletitle{Clone-Hunter: Accelerated Bound Checks Elimination via Binary Code Clone Detection}. In \bibinfo{booktitle}{\emph{Proceedings of the 2nd ACM SIGPLAN International Workshop on Machine Learning and Programming Languages}} (Philadelphia, PA, USA) \emph{(\bibinfo{series}{MAPL 2018})}. \bibinfo{publisher}{Association for Computing Machinery}, \bibinfo{address}{New York, NY, USA}, \bibinfo{pages}{11–19}.
\newblock
\showISBNx{9781450358347}
\urldef\tempurl%
\url{https://doi.org/10.1145/3211346.3211347}
\showDOI{\tempurl}


\bibitem[Yang et~al\mbox{.}(2015)]%
        {yang2015DeepWalk}
\bibfield{author}{\bibinfo{person}{Cheng Yang}, \bibinfo{person}{Zhiyuan Liu}, \bibinfo{person}{Deli Zhao}, \bibinfo{person}{Maosong Sun}, {and} \bibinfo{person}{Edward~Y. Chang}.} \bibinfo{year}{2015}\natexlab{}.
\newblock \showarticletitle{Network representation learning with rich text information}. In \bibinfo{booktitle}{\emph{Proceedings of the 24th International Conference on Artificial Intelligence}} (Buenos Aires, Argentina) \emph{(\bibinfo{series}{IJCAI'15})}. \bibinfo{publisher}{AAAI Press}, \bibinfo{pages}{2111–2117}.
\newblock
\showISBNx{9781577357384}


\bibitem[Yang et~al\mbox{.}(2022)]%
        {jia2021codee}
\bibfield{author}{\bibinfo{person}{J. Yang}, \bibinfo{person}{C. Fu}, \bibinfo{person}{X. Liu}, \bibinfo{person}{H. Yin}, {and} \bibinfo{person}{P. Zhou}.} \bibinfo{year}{2022}\natexlab{}.
\newblock \showarticletitle{Codee: A Tensor Embedding Scheme for Binary Code Search}.
\newblock \bibinfo{journal}{\emph{IEEE Transactions on Software Engineering}} \bibinfo{volume}{48}, \bibinfo{number}{07} (\bibinfo{date}{jul} \bibinfo{year}{2022}), \bibinfo{pages}{2224--2244}.
\newblock
\showISSN{1939-3520}
\urldef\tempurl%
\url{https://doi.org/10.1109/TSE.2021.3056139}
\showDOI{\tempurl}


\bibitem[Yang et~al\mbox{.}(2021)]%
        {yang2021asteria}
\bibfield{author}{\bibinfo{person}{S. Yang}, \bibinfo{person}{L. Cheng}, \bibinfo{person}{Y. Zeng}, \bibinfo{person}{Z. Lang}, \bibinfo{person}{H. Zhu}, {and} \bibinfo{person}{Z. Shi}.} \bibinfo{year}{2021}\natexlab{}.
\newblock \showarticletitle{Asteria: Deep Learning-based AST-Encoding for Cross-platform Binary Code Similarity Detection}. In \bibinfo{booktitle}{\emph{2021 51st Annual IEEE/IFIP International Conference on Dependable Systems and Networks (DSN)}}. \bibinfo{publisher}{IEEE Computer Society}, \bibinfo{address}{Los Alamitos, CA, USA}, \bibinfo{pages}{224--236}.
\newblock
\urldef\tempurl%
\url{https://doi.org/10.1109/DSN48987.2021.00036}
\showDOI{\tempurl}


\bibitem[Yang et~al\mbox{.}(2023)]%
        {yang2023asteriapro}
\bibfield{author}{\bibinfo{person}{Shouguo Yang}, \bibinfo{person}{Chaopeng Dong}, \bibinfo{person}{Yang Xiao}, \bibinfo{person}{Yiran Cheng}, \bibinfo{person}{Zhiqiang Shi}, \bibinfo{person}{Zhi Li}, {and} \bibinfo{person}{Limin Sun}.} \bibinfo{year}{2023}\natexlab{}.
\newblock \bibinfo{title}{Asteria-Pro: Enhancing Deep-Learning Based Binary Code Similarity Detection by Incorporating Domain Knowledge}.
\newblock
\newblock
\showeprint[arxiv]{2301.00511}~[cs.SE]


\bibitem[Yu et~al\mbox{.}(2020)]%
        {Yu2020OrderMatters}
\bibfield{author}{\bibinfo{person}{Zeping Yu}, \bibinfo{person}{Rui Cao}, \bibinfo{person}{Qiyi Tang}, \bibinfo{person}{Sen Nie}, \bibinfo{person}{Junzhou Huang}, {and} \bibinfo{person}{Shi Wu}.} \bibinfo{year}{2020}\natexlab{}.
\newblock \showarticletitle{Order Matters: Semantic-Aware Neural Networks for Binary Code Similarity Detection}.
\newblock \bibinfo{journal}{\emph{Proceedings of the AAAI Conference on Artificial Intelligence}} \bibinfo{volume}{34}, \bibinfo{number}{01} (\bibinfo{date}{Apr.} \bibinfo{year}{2020}), \bibinfo{pages}{1145--1152}.
\newblock
\urldef\tempurl%
\url{https://doi.org/10.1609/aaai.v34i01.5466}
\showDOI{\tempurl}


\bibitem[Zhang et~al\mbox{.}(2020)]%
        {zhang2020similarity}
\bibfield{author}{\bibinfo{person}{Xiaochuan Zhang}, \bibinfo{person}{Wenjie Sun}, \bibinfo{person}{Jianmin Pang}, \bibinfo{person}{Fudong Liu}, {and} \bibinfo{person}{Zhen Ma}.} \bibinfo{year}{2020}\natexlab{}.
\newblock \showarticletitle{Similarity metric method for binary basic blocks of cross-instruction set architecture}. In \bibinfo{booktitle}{\emph{Proceedings 2020 Workshop on Binary Analysis Research, San Diego, CA}}.
\newblock


\bibitem[Zhu et~al\mbox{.}(2023)]%
        {zhu2023ktrans}
\bibfield{author}{\bibinfo{person}{Wenyu Zhu}, \bibinfo{person}{Hao Wang}, \bibinfo{person}{Yuchen Zhou}, \bibinfo{person}{Jiaming Wang}, \bibinfo{person}{Zihan Sha}, \bibinfo{person}{Zeyu Gao}, {and} \bibinfo{person}{Chao Zhang}.} \bibinfo{year}{2023}\natexlab{}.
\newblock \showarticletitle{kTrans: Knowledge-Aware Transformer for Binary Code Embedding}.
\newblock \bibinfo{journal}{\emph{arXiv preprint arXiv:2308.12659}} (\bibinfo{year}{2023}).
\newblock


\bibitem[Zuo et~al\mbox{.}(2019)]%
        {zuo2018innereye}
\bibfield{author}{\bibinfo{person}{Fei Zuo}, \bibinfo{person}{Xiaopeng Li}, \bibinfo{person}{Patrick Young}, \bibinfo{person}{Lannan Luo}, \bibinfo{person}{Qiang Zeng}, {and} \bibinfo{person}{Zhexin Zhang}.} \bibinfo{year}{2019}\natexlab{}.
\newblock \showarticletitle{Neural Machine Translation Inspired Binary Code Similarity Comparison beyond Function Pairs}. In \bibinfo{booktitle}{\emph{26th Annual Network and Distributed System Security Symposium, {NDSS} 2019, San Diego, California, USA, February 24-27, 2019}}. \bibinfo{publisher}{The Internet Society}.
\newblock
\urldef\tempurl%
\url{https://www.ndss-symposium.org/ndss-paper/neural-machine-translation-inspired-binary-code-similarity-comparison-beyond-function-pairs/}
\showURL{%
\tempurl}


\bibitem[Zynamics(2024)]%
        {bindiff}
\bibfield{author}{\bibinfo{person}{Zynamics}.} \bibinfo{year}{2024}\natexlab{}.
\newblock \bibinfo{title}{Bindiff7}.
\newblock
\newblock
\urldef\tempurl%
\url{https://www.zynamics.com/bindiff.html}
\showURL{%
\tempurl}
\newblock
\shownote{[Online; accessed 17-June-2024]}.


\end{thebibliography}
